\documentclass[onecolumn,aps,pr,superscriptaddress,preprintnumbers,nofootinbib,10pt]{revtex4-2}
\usepackage{xcolor}
\usepackage{soul}
\usepackage{amsfonts}
\usepackage{amssymb}
\usepackage{amsmath}
\usepackage{amsthm}
\usepackage{graphicx}
\usepackage{appendix}
\usepackage{bbold}
\usepackage{dsfont}
\usepackage{enumitem}
\usepackage{xcolor}
\usepackage[
colorlinks=true]{hyperref}

\usepackage{yfonts}
\usepackage{MnSymbol}

\usepackage[normalem]{ulem}
\newcommand\redsout{\bgroup\markoverwith{\textcolor{red}{\rule[0.4ex]{3pt}{0.7pt}}}\ULon}

\providecommand{\U}[1]{\protect\rule{.1in}{.1in}}

\newcommand{\tr}{\ensuremath{\mathrm{Tr}}}

\usepackage{color}
\usepackage{listings}
\definecolor{darkgreen}{rgb}{0,0.35,0}
\definecolor{Rood}{rgb}{1, 0, 0}

\newtheorem{theorem}{Theorem}[section]
\newtheorem{lemma}[theorem]{Lemma}
\newtheorem{proposition}[theorem]{Proposition}

\theoremstyle{definition}
\newtheorem{definition}{Definition}[section]

\usepackage{verbatim}
\usepackage{subcaption}

\usepackage{array}

\newcommand{\beq}{\begin{equation}}
\newcommand{\eeq}{\end{equation}}
\newcommand{\bea}{\begin{eqnarray}}
\newcommand{\eea}{\end{eqnarray}}

\newcommand{\m}[1]{\bar{#1}}

\graphicspath{{./figs/}}
\newbox{\ORCIDicon}
\sbox{\ORCIDicon}{\large
                  \includegraphics[width=0.8em]{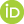}}

\begin{document}

\title{
Dually weighted multi-matrix models as a path to causal gravity-matter systems}

\author{Juan~L.~A. Abranches}
\email{juan.abranches@oist.jp}
\affiliation{Okinawa Institute of Science and Technology Graduate University, 1919-1, Tancha, Onna, Kunigami District, Okinawa 904-0495, Japan}
\author{Antonio~D.~Pereira\,\href{https://orcid.org/0000-0002-6952-2961}{\usebox{\ORCIDicon}}}
\email{adpjunior@id.uff.br}
\affiliation{Instituto de F\'isica, Universidade Federal Fluminense, Campus da Praia Vermelha, Av. Litor\^anea s/n, 24210-346, Niter\'oi, RJ, Brazil}
\affiliation{Institute for Mathematics, Astrophysics and Particle Physics (IMAPP),
Radboud University, Heyendaalseweg 135, 6525 AJ Nijmegen, The Netherlands}
\author{Reiko~Toriumi\,\href{https://orcid.org/0000-0003-1274-0463}{\usebox{\ORCIDicon}}}
\email{reiko.toriumi@oist.jp}
\affiliation{Okinawa Institute of Science and Technology Graduate University, 1919-1, Tancha, Onna, Kunigami District, Okinawa 904-0495, Japan}

\begin{abstract}
We introduce a dually-weighted multi-matrix model that for a suitable choice of weights reproduce two-dimensional Causal Dynamical Triangulations (CDT) coupled to the Ising model. When Ising degrees of freedom are removed, this model corresponds to the CDT-matrix model introduced by Benedetti and Henson [Phys. Lett. B 678, 222 (2009)]. We present exact as well as approximate results for the Gaussian averages of characters of a Hermitian matrix $A$ and $A^2$ for a given representation and establish the present limitations that prevent us to solve the model analytically. This sets the stage for the formulation of more sophisticated matter models coupled to two-dimensional CDT as dually weighted multi-matrix models providing a complementary view to the standard simplicial formulation of CDT-matter models.

\end{abstract}

\maketitle

\large
\section{Introduction}
\label{Sect:Intro}

The search for a consistent theory of quantum gravity that is valid up to arbitrarily short length scales remains an
open challenge in theoretical physics. The direct quantization of the gravitational field described by general relativity
using the standard perturbative field-theoretic techniques renders a perturbatively non-renormalizable quantum field
theory (QFT), see, e.g., \cite{tHooft:1974toh,Christensen:1979iy,Goroff:1985th}. In the history of the construction of quantum theories of the fundamental interactions, perturbative
non-renormalizability of a given well-grounded theoretical model was circumvented by the replacement of such a
model by a more fundamental description either by the addition of new fields or by the discovery that the degrees
of freedom to be quantized were different ones. As such, successful models were constructed and the fulfillment
of perturbative renormalizability became a paradigmatic prescription. Such an
attitude could be taken in the case of quantum gravity. Perhaps, the Einstein-Hilbert action is not the correct starting
point and just corresponds to an effective description of classical gravity at sufficiently low energies and the fundamental degrees of freedom are
completely different. Alternatively, one could add more fields that restore a well-behaved ultraviolet behavior for the
would-be quantum theory of gravity. In all those cases, there is an insistence in keeping the tools of perturbative
renormalization as the guiding principle in the construction of the underlying QFT. However, as it is well-known
nowadays, perturbative renormalizability is neither necessary nor sufficient to define a theory that is valid across
arbitrary length scales. The existence of Landau poles in scalar field theories in four dimensions or in quantum
electrodynamics shows that perturbatively renormalizable theories might require the introduction of a ultraviolet cutoff
at some finite energy scale in order to be non-trivial, see, e.g., \cite{Frohlich:1982tw,Gell-Mann:1954yli,Gockeler:1997dn,Callaway:1988ya}. In practice, however, such Landau poles live at scales far beyond our experimental capabilities
and therefore are harmless. Similarly, by introducing a cutoff for the standard QFT of general
relativity, one has a perfectly well-defined framework to compute quantum corrections to gravitational processes at
energy scales below the cutoff \cite{Donoghue:1993eb,Donoghue:1995cz,Burgess:2003jk}. Hence, constructing a fundamental theory of quantum gravity, as opposed to such an
effective field theory does not necessarily require the construction of a perturbatively renormalizable QFT.
One possibility to define a QFT that is valid up to arbitrarily high energies is by demanding that the coupling
constants that enter physical observables reach a renormalization group fixed point. When this happens, the theory
attains a scale-invariant regime and the ultraviolet cutoff can be safely removed. A paradigmatic example is quantum
chromodynamics which is asymptotically free, i.e., the theory reaches the free (or Gaussian) fixed point in the ultraviolet\footnote{In this case, the theory is also perturbatively renormalizable.}, see \cite{Politzer:1973fx,Gross:1973id}. In this case, since the coupling becomes sufficiently small at high energies,
perturbation theory is applicable and can probe the existence of such a fixed point. Alternatively, the theory could have an interacting (or non-Gaussian) fixed point in the ultraviolet. In this situation, the values of the couplings at the fixed point might not be small and perturbation theory is not sufficient to probe it. When such a non-trivial fixed point exists\footnote{Actually, besides having a non-trivial fixed point, it is necessary that just finitely many couplings are adjusted in order to reach it, otherwise the theory is not predictive.}, we say that the QFT is asymptotically safe. In \cite{Hawking:1979ig}, Weinberg put forward the conjecture that quantum
gravity could be realized as an asymptotically safe QFT. Evidence for that was obtained through the so-called 2 + $\epsilon$
expansion \cite{Kawai:1989yh,Kawai:1992np,Kawai:1993mb,Kawai:1995ju}. However, due to its non-perturbative nature, the search for such a fixed point needs different
techniques with different systematics in order to make its finding a robust claim. For many years, this conjecture
was not investigated systematically due to the lack of suitable technical tools. However, during the 1990's, this situation suffered a twist due to the development of two major frameworks: the use of functional or non-perturbative renormalization group equations to quantum gravity, see, e.g., \cite{Reuter:1996cp} and reviews \cite{Reuter:2012id,Eichhorn:2018yfc,Bonanno:2020bil,Pawlowski:2020qer,Eichhorn:2022gku,Saueressig:2023irs} and references therein and the proper understanding of how to put quantum
gravity on a lattice in a consistent fashion with causality constraints, see \cite{Ambjorn:1998xu,Ambjorn:2001cv} and reviews \cite{Ambjorn:2012jv,Loll:2019rdj} and references therein. Those computational techniques can be viewed as complementary
approaches to the search for a non-perturbative ultraviolet fixed point for quantum gravity and it has become usual to call the search for the fixed point using continuum techniques as the asymptotic safety approach while the search for a suitable continuum limit of discretized path integrals with causal constraints as the Causal Dynamical Triangulations (CDT) program. Both lines of research have provided a large body of evidence for a suitable continuum limit in four dimensions. We refer the reader to the references in the reviews above mentioned for a comprehensive list of the most recent results in those fields\footnote{We highlight that both approaches follow the inspiring idea put forward by Weinberg but there is no reason a priori to expect that the would-be continuum limits belong to the same universality class. In particular, CDT implements a Lorentzian path integral while the continuum computations employing the functional renormalization group are mostly Euclidean, with few exceptions, see, e.g., \cite{Manrique:2011jc,Biemans:2016rvp,Bonanno:2021squ,Fehre:2021eob,Saueressig:2023tfy}.}. 

Despite the tremendous progress achieved in the physically motivated four-dimensional case, the employment of non-perturbative techniques require substantial truncations and approximations that still need a lot of effort to produce quantitatively reliable results for those universal quantities, i.e., observables that could tell the existence of the continuum limit. Yet a fruitful playground is to consider two-dimensional quantum gravity. In this case, the interplay between continuum and discrete descriptions of the path integral of quantum gravity has witnessed a significant progress over the last decades ranging from the development of novel techniques to perform explicit calculations to rigorous mathematical results. We refer to, e.g., \cite{DiFrancesco:1993cyw,Nakayama:2004vk} for reviews on the topic. In particular, the Dynamical Triangulation program was born in two dimensions within an Euclidean setting while its continuum counterpart is encoded in Liouville quantum gravity. In particular, a successful implementation of the discretization of the path integral of quantum gravity and thus a sum over geometries and possibly topologies can be achieved by purely combinatorial means with the use of the so-called matrix models \cite{DiFrancesco:1993cyw,David:1984tx,DAVID1985543,Kazakov:1985ea,Douglas:1989ve,Brezin:1990rb,Gross:1989vs,Gross:1989aw,Ginsparg:1993is}. Triangles can be taken as dual representations of cubic-matrix vertices and Feynman diagrams of such a model are simply triangulations of two-dimensional surfaces. The perturbative expansion of matrix-models partition functions can be organized in powers of $1/N$ with $N$ standing for the size of the Hermitian random matrices and each power of such an expansion is associated with a specific genus $g$. Such a remarkable expansion allows for the investigation of continuum limits, e.g., where just spherical topologies contribute (the so-called planar limit) or where all topologies are taken into account (double-scaling limit). Thanks to the rich combinatorial framework behind the theory of random matrices, very powerful results could be established in such a discrete-to-continuum approach. The two-dimensional case is simple enough so that one can still solve the Dynamical Triangulation model with no need to make use of random matrices, but such a perspective opens up the possibility to think about higher-dimensional discrete approaches as theories of higher-order tensors. So far, the higher dimensional tensor models were not successful in producing a suitable continuum limit in higher dimensions. Primarily, the main obstacle was the lack of the analogue of the $1/N$-expansion in matrix models. Such a difficulty was lifted thanks to the works of Gurau that introduced the so-called colored tensor models \cite{Gurau:2009tw,Gurau:2010ba,Gurau:2011aq,Gurau:2011xp}. Yet the continuum limits obtained from the colored tensor models up to date do not feature an extended-geometry-like phase very much in agreement with the simulations performed for Dynamical Triangulations in dimensionality greater than two \cite{Gurau:2013cbh,Loll:1998aj,Ambjorn:1995dj,Bialas:1996wu,deBakker:1996zx}, see however \cite{Laiho:2011ya,Laiho:2016nlp,Dai:2021fqb,Bassler:2021pzt,Dai:2023tud,Asaduzzaman:2022kxz,Hamber:2009mt,Hamber:2012te,Hamber:2015jja}. This suggests that one could inspect more intricate continuum limits from those models, see, e.g., the realization of such a reasoning following the functional renormalization group \cite{Eichhorn:2017xhy,Eichhorn:2018ylk,Eichhorn:2018phj,Eichhorn:2019hsa}, or that a restriction to the configuration space should be implemented. This is precisely what is achieved with CDT which displays a robust body of evidence that such a suitable continuum limit exists in higher dimensions \cite{Ambjorn:2010ce,Ambjorn:1998xu,Ambjorn:2000dv,Ambjorn:2004qm,Ambjorn:2005qt,Ambjorn:2007jv,Ambjorn:2011ph,Ambjorn:2011cg,Jordan:2013iaa,Ambjorn:2020rcn,Benedetti:2022ots}. Nevertheless, the analogue of CDT in a tensor-model language is still unknown. In this paper, we explore a matrix model that would be such a realization in two dimensions introduced by Benedetti and Henson in \cite{Benedetti:2008hc}. In essence, this model implements the (non-local) global time foliation of CDT together with the avoidance of spatial topology change by means of  local constraints. This is achieved by the so-called dually weighted matrix models introduced in \cite{DiFrancesco:1992cn,Kazakov_1996,Kazakov:1995gm,Kazakov:1996zm}. On top of the CDT configurations generated by the Benedetti-Henson model, we will introduce degreess of freedom arising from the Ising model, i.e., we develop a (multi-)matrix model of the Ising model coupled to CDT in two dimensions. Alternatively, the combinatorial models can be enriched with more structure such as group-theoretic data giving birth to the so-called Group Field Theories which may or may not include a Lorentzian setting \cite{Oriti:2011jm,Jercher:2022mky,Marchetti:2022igl,Marchetti:2022nrf}. Thus, two-dimensional quantum gravity might be too simplistic for some purposes but certainly is an inspirational source for more sophisticated candidates in higher dimensions.  

Moreover, a realistic description of our Universe
must accommodate matter degrees of freedom whose 
fluctuations can affect the existence of a suitable continuum limit. Hence, it is conceivable that our comprehension of the potential mechanism that drives quantum gravity asymptotically safe (irrespective of the approach followed here) depends on a simultaneous treatment of gravitational as well as matter degrees of freedom. This ambitious goal has
witnessed significant progress over the last two decades both from the point of view of the functional renormalization
group as well as from lattice simulations and encouraging results were found \cite{Knorr:2022dsx,Eichhorn:2022gku,Pawlowski:2023gym,Loll:2019rdj,Ambjorn:2022naa}. In the present work we provide several results regarding a simple but still quite rich gravity-matter model: two-dimensional CDT coupled to the Ising model. There are several motivations to look at this model but we emphasize the following: Implementing the causality constraint at the level of a matrix model is a first step towards its implementation in tensor models. As it is known, such a model does not have a known exact solution and developing appropriate tools to deal with it is mandatory. The inclusion of the Ising model produces a multi-matrix model which is interesting on its own. Clearly, having the Ising model as the matter component is tremendously far from the rich structure of the Standard Model of Particle Physics coupled to quantum gravity. Yet the purpose here is to explore the impact of the dynamical (and causal) lattice to the Ising model and vice-versa. We emphasize that matrix models for Euclidean two-dimensional gravity coupled to the Ising model were investigated in the past \cite{Kazakov:1986hy,Boulatov:1986sb,Kazakov:1986hu,Brezin:1989db} as well as its higher-dimensional version \cite{Bonzom:2011ev,Sasakura:2014zwa,Lahoche:2019orv}. Moreover, the coupling between the Ising model and CDT was also investigated since the birth of two-dimensional CDT, see, e.g., \cite{Ambjorn:1999gi,Benedetti:2006zy,Benedetti:2006rv,Napolitano:2015poa,Sato:2017ccb,Cerda-Hernandez:2017jge,Ambjorn:2020kpr}. We establish several properties of such a CDT-Ising multi-matrix model paving the way for a future complete solution or making its structure well-grounded for studies using non-perturbative tools such as the functional renormalization group methods\footnote{Functional renormalization techniques were developed for matrix models, see, e.g., \cite{Eichhorn:2013isa,Eichhorn:2014xaa,Lahoche:2019ocf,Perez-Sanchez:2020kgq,Eichhorn:2020sla}. There are studies involving the functional renormalization group equation applied to the Benedetti-Henson model, see \cite{Castro:2020dzt}.} or numerical simulations.

This paper is organized as follows:
In Section \ref{Sect:CDTMM}, we will briefly review CDT-like matrix model as a dually weighted matrix model as introduced by Benedetti and Henson \cite{Benedetti:2008hc}.
In Section
\ref{Sect:BHTopOrient}, we develop an analysis concerning topologies which admit global foliations, specifically within the Benedetti-Henson pure CDT-like matrix model (in Section \ref{Sect:CDTMM}) and further on non-orientable symmetric matrix models.
In Section \ref{Sect:IM}, we discuss briefly the well-known two-matrix models which can be interpreted as Ising model on random two-surfaces.
In Section
\ref{sec:CDTmmIsing}, we present our CDT-matrix model coupled to the Ising model by combining the models in Sections \ref{Sect:CDTMM} and \ref{Sect:IM}.
In Section
\ref{cmsec}, we study the properties of the matrices $C_m$ which, in the case of $m=2$, is responsible for generating the global foliation structure on the dual triangulation graphs of the CDT-like matrix models. 
Theorem \ref{thm:Cm} is one of the main results of this paper.
In Section
\ref{sec:characterexp},
we present the character expansions for partition functions of the two models, i.e.,  pure CDT-like matrix model by Benedetti-Henson and CDT-like matrix model coupled to Ising model presented in Sections \ref{Sect:CDTMM} and \ref{sec:CDTmmIsing} respectively. We base our further analyses on these character expansions.
In Section \ref{sec:unitintmono}, we analyze the character expansion of the partition function for the latter model (i.e., CDT-like matrix model coupled to Ising model) and present an expression in Proposition \ref{unitaryint} in terms of Clebsch-Gordan coefficients using Weingarten calculus. 
In Section
\ref{sec:characters},
we evaluate Gaussian averages of characters for a given representation $r$, $\langle \chi_r(A) \rangle_0$ and $\langle \chi_r(A^2) \rangle_0$
for CDT-like matrix models presented in Section \ref{Sect:CDTMM},
using various techniques, e.g., Schur-Weyl duality, theory of symmetric group algebra, etc.
Theorem \ref{thm:chiA2R}, one of our main results, is related to Brauer algebra.
In Section \ref{sec:conc}, we conclude by revisiting and summarizing the various results reported in this work.

\section{Causal Dynamical Triangulations as a matrix model: a short review}
\label{Sect:CDTMM}

Let us describe via matrix model the Causal Dynamical Triangulation (CDT) in two dimensions. The CDT was first worked and solved in \cite{Ambj_rn_1998} as a lattice model. In \cite{Benedetti:2008hc}, a definition through a matrix model was given, and we follow this definition. 
Let us also refer to this model as pure{\footnote{We specifically refer {\it pure} (gravity) because later, we will talk about a matrix model which generate random two-surfaces (gravity) with {\it matter}. }} CDT-like matrix model of Benedetti-Henson.
The CDT-like matrix model graphs can be represented as ribbon graphs just like other matrix models. 
In fact, these CDT-like matrix models fall into a class of so called dually weighted models.
In \cite{Kazakov_1996}, Kazakov, Staudacher, and Wynter introduced a so called dually weighted graphs where different weights are assigned for vertices with different coordination numbers and for faces with different lengths.
As one will see below in elaborating the features of such CDT-like matrix models, our models generate the dually weighted graphs only with vertices with coordination number 3, and with faces with arbitrary lengths however, of particular type (restricting to only having two edges of a certain type, which we will name timelike).
\begin{definition}
    A {\it prime ribbon graph} is a set of topological discs and topological rectangles satisfying the following properties: no two rectangles or two discs intersect; each rectangle has exactly one pair of opposite sides that are contained in disc's circumferences, and this is the only intersection between rectangles and discs.
\end{definition}

\begin{definition}
    A rectangle's edge that is not contained in a disc is called a {\it strand}. 
\end{definition}

\begin{definition}
    A {\it ribbon graph} is the boundary of the union of the discs and rectangles of a prime ribbon graph.
\end{definition}

Additionally, one can represent such ribbon graphs as multigraphs. 
Given a ribbon graph, we assign a vertex to each disc and an edge to each rectangle. If a disc and a rectangle are connected, i.e., they intersect, then respective vertex and edge are connected. See Fig.\,\ref{fig:cdtfa} for an example. We may also call a ribbon graph's discs and rectangles as vertices and edges, respectively. We define a face as a closed curve formed by strands, and we carry the notion of faces from the ribbon graph to the multigraph. 
An example of a face of a ribbon graph is shown in Fig.\,\ref{ribface}.
\begin{figure}[t]
    \begin{subfigure}{0.4\textwidth}
    \centering
    \includegraphics[width=.7\linewidth]{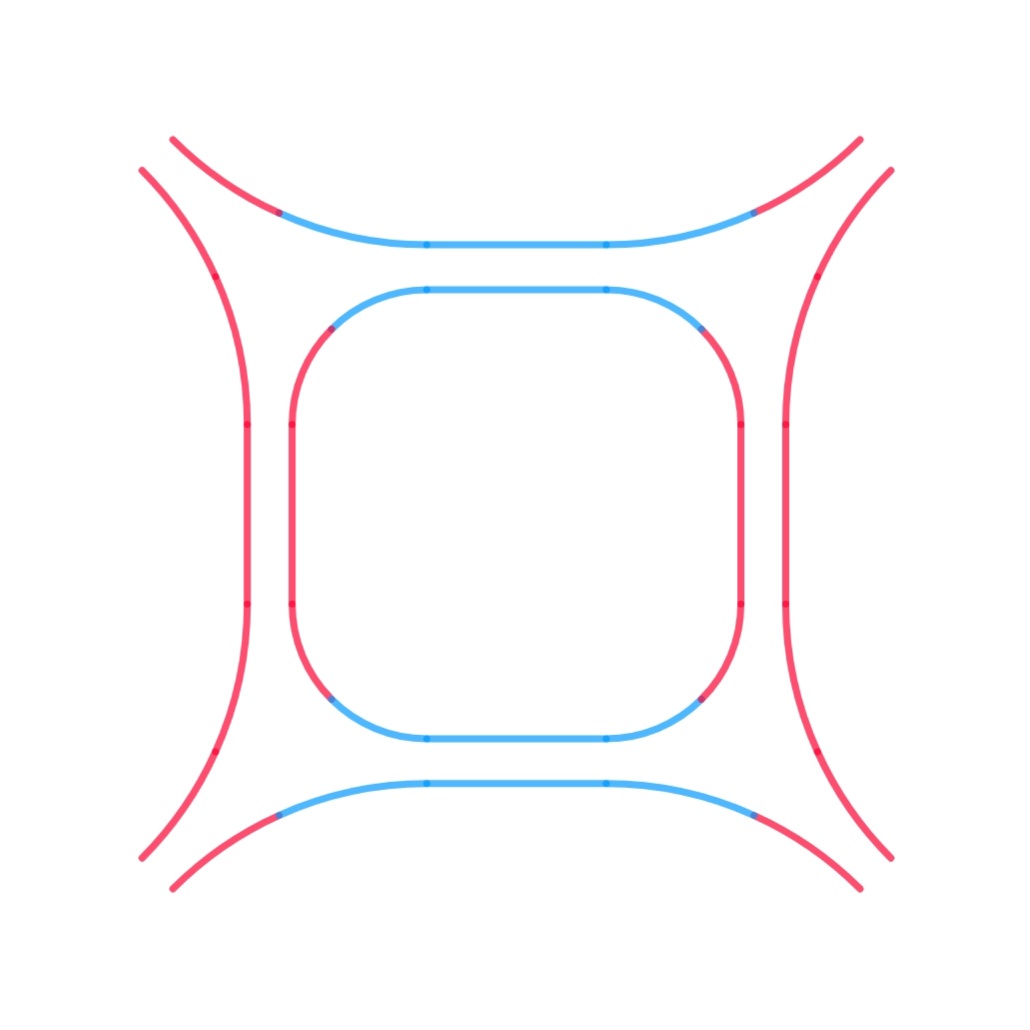}
    \caption{A closed loop is called a face of a ribbon graph. The timelike edges are shown in blue, and the spacelike edges are shown in red.}
    \label{ribface}
    \end{subfigure}
    \begin{subfigure}{0.4\textwidth}
    \centering
    \includegraphics[width=.7\linewidth]{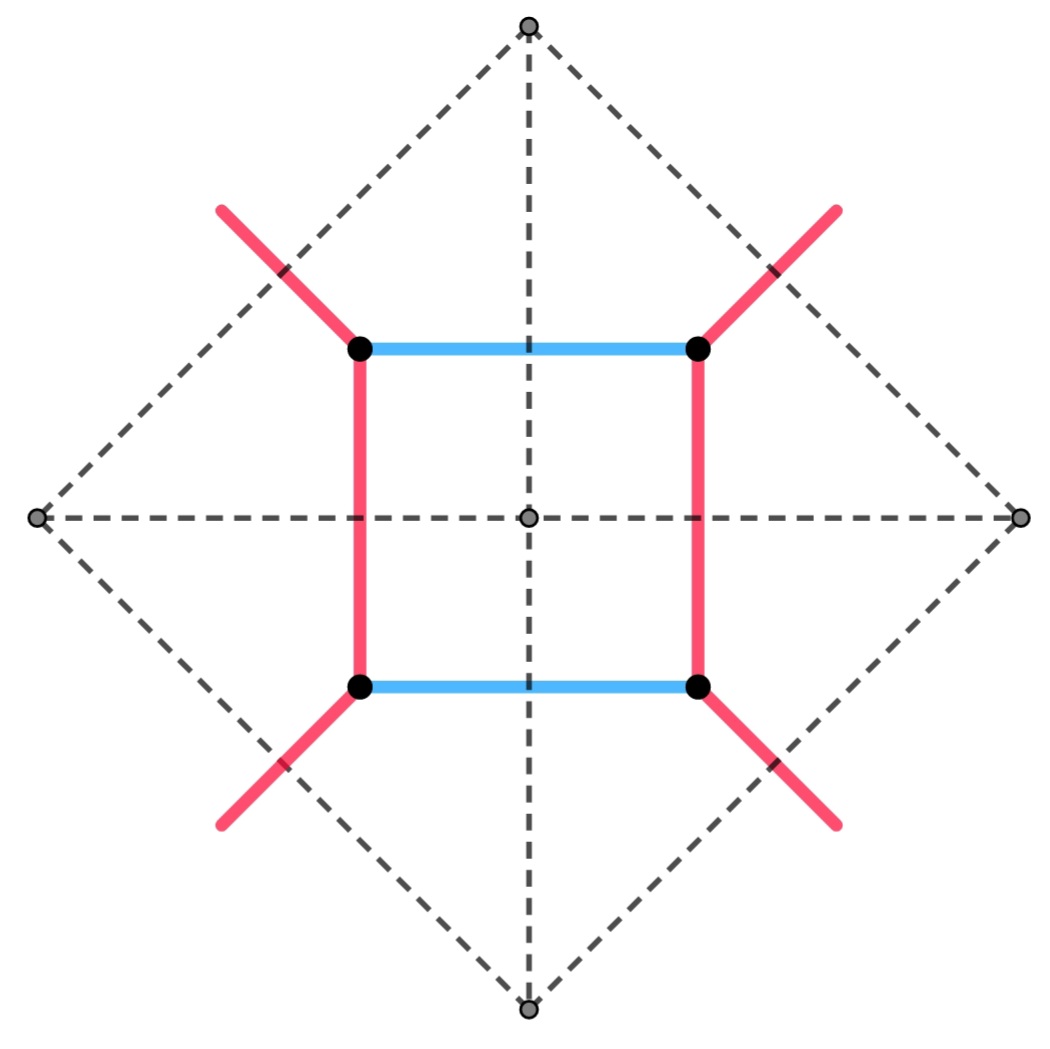}
    \caption{The multigraph representation of the ribbon graph in Fig.\,\ref{ribface} is shown in solid lines and black vertices. Dashed lines represent the edges of the dual graph and the gray dots are vertices of the dual graph. 
    }
    \label{graphndual}
    \end{subfigure}
\caption{A face of a CDT ribbon graph. }
\label{fig:cdtfa}
\end{figure}

The CDT-like matrix model ribbon graphs are edge-colored (not proper), partitioned into two sets: spacelike edges and timelike edges.
Together with the defining properties of a ribbon graph and the edge coloring, the CDT ribbon graph also has the following defining properties (CDT conditions):

\begin{enumerate}
    \item Every vertex is three valent and has exactly one timelike edge and two spacelike edges incident to it.
    \item Every face has either two timelike edges or none.
\end{enumerate}
\begin{figure}[t]
    \begin{subfigure}{0.3\textwidth}
    \centering
    \includegraphics[width=.5\linewidth]{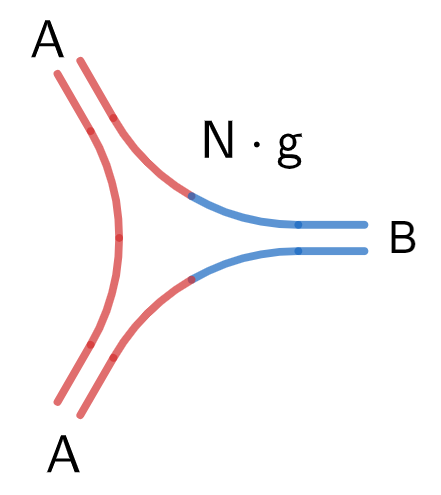}
    \caption{Vertex of a CDT ribbon graph.}
    \label{fig:cdtve}
    \end{subfigure}
    \begin{subfigure}{0.3\textwidth}
    \centering
    \includegraphics[width=.5\linewidth]{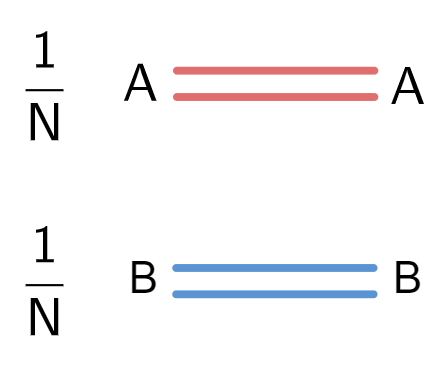}
    \caption{Edges of a CDT ribbon graph.}
    \label{fig:cdtprop}
    \end{subfigure}
\caption{Components of the CDT ribbon graph.}
\label{fig:cdtel}
\end{figure}
Notice that the faces can have any number of spacelike edges. Fig.\,\ref{fig:cdtel} shows a representation of the vertices and edges. 
A triangulation is obtained with the dual graph; the dual graph is defined by associating a vertex with each face of the graph, and if two faces share an edge, their respective vertices are connected by an edge (see Fig.\,\ref{fig:cdtprop}). This also leads to each vertex in the graph being associated with a face in the dual graph, and since all vertices are three valent, all faces in the dual graph are triangles. In Fig.\,\ref{:sph}, we show an example of a sphere triangulation that satisfies the CDT conditions given above.
\begin{figure}[t]
\begin{minipage}[t]{0.8\textwidth}
    \centering
    \includegraphics[width=.35\linewidth]{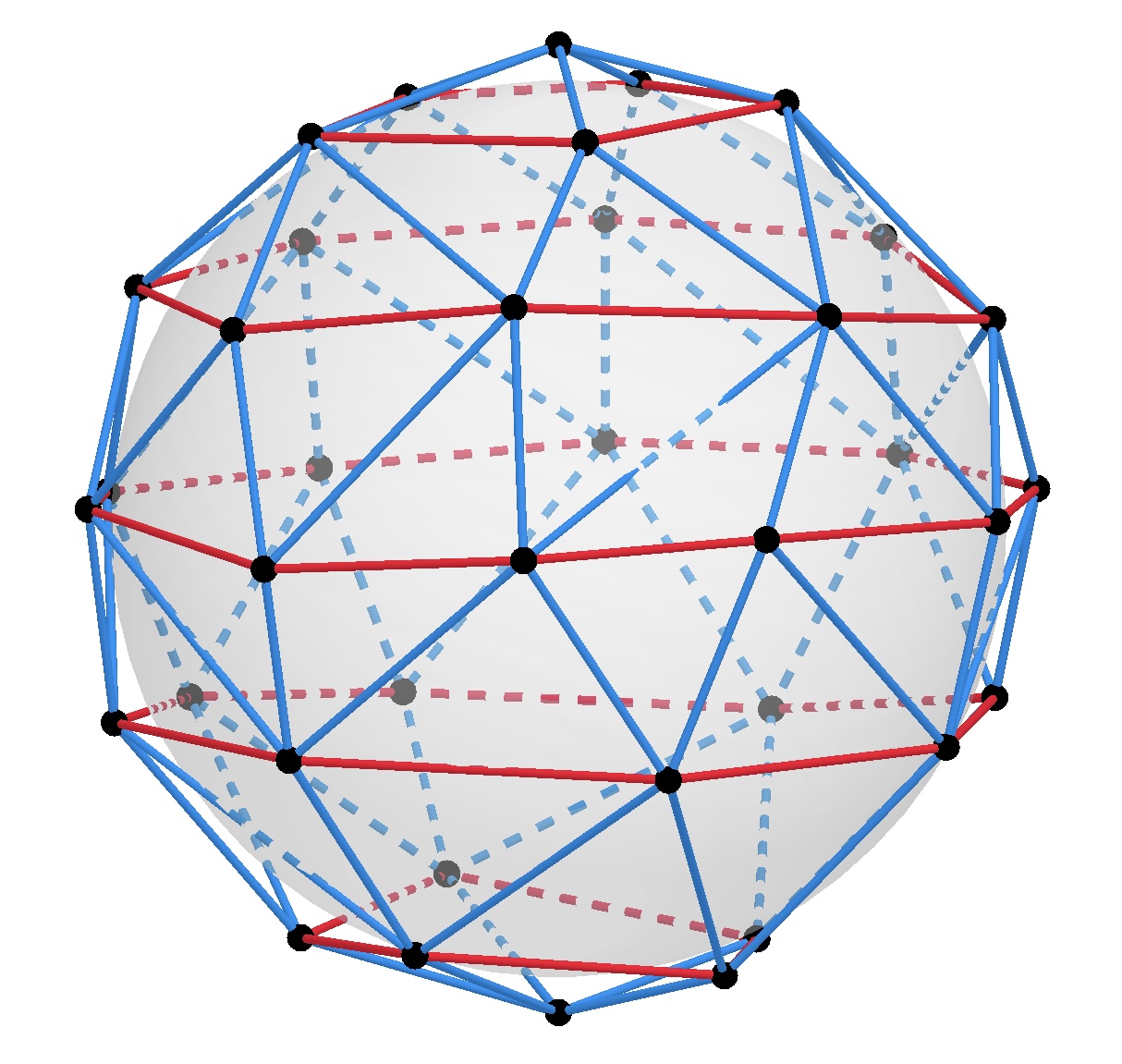}
    \caption{An example of a CDT triangulation graph (dual of a corresponding CDT ribbon graph). spacelike edges in the ribbon graph are dual to timelike edges in the dual triangulation graph, and vice-versa. 
    The vertices of the dual triangulation graph can have two or zero spacelike edges (in red) and any number of timelike edges (in blue).
    }
    \label{:sph}
\end{minipage}
\end{figure}
A matrix model that generates ribbon graphs,  satisfying these defining CDT conditions above can be constructed by associating a Hermitian matrix $A$ with spacelike edges of a ribbon graph and a Hermitian matrix $B$ to timelike edges of a ribbon graph. An auxiliary partition function for such a pure CDT-like matrix model is given by 
\begin{equation}
    \mathcal{Z} = \int dA dB \;e^{-N \tr\left[\frac{1}{2}A^2+\frac{1}{2}({C_2}^{-1}B)^2-g A^2B\right]}\,.
\label{eq:CDTMMDarioABC2}
\end{equation}
The action (the negative of the exponent in \eqref{eq:CDTMMDarioABC2}) of the model has the term $-Ng A^2B$, responsible for the property that every vertex have exactly two spacelike edges and one timelike edge. A matrix $C_2$ is introduced in the quadratic term of $B$, and it satisfies, for a positive integer $p$ and in the large $N$ limit,
\beq
\tr [(C_2)^p] = N\delta_{p,2}\,.
\label{eq:TrCn}
\eeq
This condition on $C_2$ is what sets the property that faces have only two or zero timelike edges. 
In Section \ref{cmsec}, we study this matrix $C_2$ in more detail. 
Notice that this constraint \eqref{eq:TrCn} on $C_2$ is overdetermined, once $p > N$. 
Therefore, the condition \eqref{eq:TrCn} on $C_2$ can only be imposed in the large $N$ limit. 
The Gaussian average of a function $f(A,B)$ is defined by
\begin{equation}
    \langle f(A,B)\rangle_0=\frac{\int dAdB \;f(A,B)\;e^{-N \tr\left[\frac{1}{2}A^2+\frac{1}{2}({C_2}^{-1}B)^2\right]}}{\int dAdB\;e^{-N \tr\left[\frac{1}{2}A^2+\frac{1}{2}({C_2}^{-1}B)^2\right]}}\;,
\end{equation}
and we find the following properties:
\begin{equation}
    \langle A_{ij}A_{kl} \rangle_0 = \frac{1}{N} \delta_{ik}\delta_{jl}\,,
\end{equation}
\begin{equation}
    \langle B_{ij}B_{kl} \rangle_0 = \frac{1}{N} {C_2}_{ik}{C_2}_{jl}\,,
\end{equation}
\begin{equation}
    \langle A_{ij}B_{kl} \rangle_0 = 0\;.
\end{equation}
Since the condition on $C_2$ in \eqref{eq:TrCn} can only be set at large $N$, the partition function for CDT is only be obtained in the large $N$ limit,
\begin{equation}
    \mathcal{Z}_{CDT}= \lim_{N\rightarrow \infty}\mathcal{Z}\;.
\end{equation}

\section{Properties of the Benedetti-Henson model: allowed topologies and (non-)orientability}
\label{Sect:BHTopOrient}

The two defining properties, albeit being local, put a global constraint on the topologies that our ribbon graphs (and in their dual triangulation) may represent. 
Studying the properties of the elements of the graph enables us to determine the values of Euler characteristic which are allowed.
The analysis can be done either with the ribbon graphs or with the triangulations dual to them. 
Here, let us work on the ribbon graphs.

Let us study some properties satisfied by the faces. 
A face that has no timelike edges is simply a set of connected spacelike edges. See Fig.\,\ref{fig:squa}.
For a face that has two timelike edges, since the vertices can only have one timelike edge, a timelike edge in a face is always a neighbor of two spacelike edges of the face. See Fig.\,\ref{fig:pent}. 
This way, the edges of a given face of the latter type form a sequence composed by a timelike edge, a several spacelike edges, a timelike edge, a several spacelike edges, and then connecting back to the first timelike edge. 
Thus, this type of face has two disjoint sets of connected spacelike edges, and these two sets are separated by timelike edges. Let us call these the two sets as {\it spacelike boundaries} of a face.

\begin{definition}
A {\it strip} is defined as a set that contains faces which are sequentially connected by timelike edges, and also contains the edges and the vertices of these faces. Each face in a strip is glued by timelike edges to either two faces in the same strip (it can also glue to itself) or none, and one timelike edge is shared exactly by two faces (these two faces are possibly the same face). 
If a spacelike edge has multiplicity two in one face or if it belongs to two faces of the same strip, consider those two appearances as distinct elements of a strip.
\end{definition}
\begin{definition}\label{def0} A
{\it boundary of a strip} is defined as a set of spacelike edges that are sequentially connected by vertices.
\end{definition}
\begin{definition}\label{def00} A
{\it interior of a strip} is defined as the set of faces and timelike edges in a strip. 
\end{definition}
\begin{definition}\label{def2}
A {\it regular strip} is a strip that has two boundaries (see Fig.\,\ref{fig:restri}).
\end{definition}
\begin{definition}\label{def1}
A {\it singular strip} is defined as a strip composed by a face with no timelike edges (see Fig.\,\ref{fig:sistri}).
\end{definition}
\begin{definition}\label{def3}
A \it M{\"o}bius strip is a non-singular strip that has one boundary (see Fig.\,\ref{fig:mostri}).
\end{definition}

Some important properties satisfied by CDT ribbon graphs are:
\begin{enumerate}
    \item Every face belongs to a strip and is in only one strip. 
    \item A strip has only spacelike edges in its boundary, and all spacelike edges are in boundaries  of strips.
    \item Since there is a finite number of faces, there is a finite number of strips.
    \item Every non-singular strip is a periodic sequence which alternates between faces and timelike edges; and every boundary of a strip is a periodic sequence which alternates between spacelike edges and vertices.
    \item A strip has only one or two boundaries.
    \item A boundary of a strip either bounds one other strip or connects the strip to itself.
\end{enumerate}
\begin{figure}[t]
    \begin{subfigure}{0.3\textwidth}
    \centering
    \includegraphics[width=.7\linewidth]{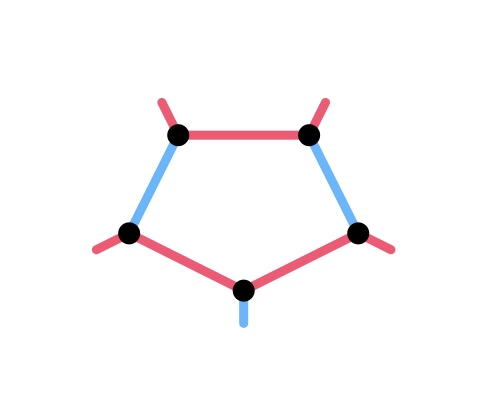}
    \caption{A face with two timelike edges.}
    \label{fig:pent}
    \end{subfigure}
    \begin{subfigure}{0.3\textwidth}
    \centering
    \includegraphics[width=.7\linewidth]{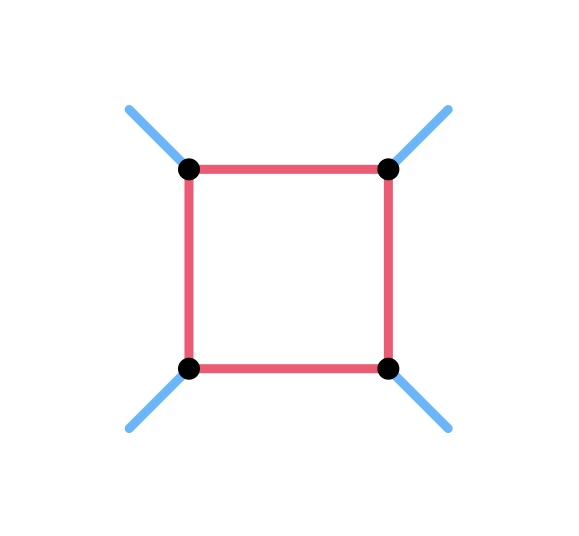}
    \caption{A face with no timelike edge.}
    \label{fig:squa}
    \end{subfigure}
\caption{Components of the CDT ribbon graph. Spacelike edges are shown in red, while timelike edges are shown in blue.}
\label{fig:cdtfa1}
\end{figure}
Properties 1 to 4 are easy to see, but let us take a closer look at property 5. 
The singular strip is just a face with only spacelike edges, thus as a strip it has only one boundary, as shown in Fig.\,\ref{fig:sistri}. 
\begin{figure}[t]
    \begin{subfigure}{0.4\textwidth}
    \centering
    \includegraphics[width=.7\linewidth]{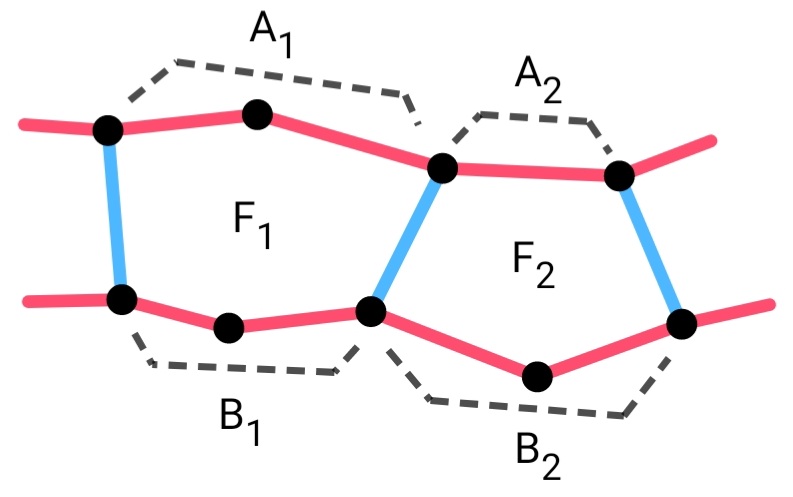}
    \caption{The face $F_1$ has spacelike boundaries $A_1$ and $B_1$, and the face $F_2$ has the spacelike boundaries $A_2$ and $B_2$. Together, they have the spacelike boundaries $A_1\cup A_2$ and $B_1\cup B_2$.}
    \label{strinex}
    \end{subfigure}
    \begin{subfigure}{0.4\textwidth}
    \centering
    \includegraphics[width=.7\linewidth]{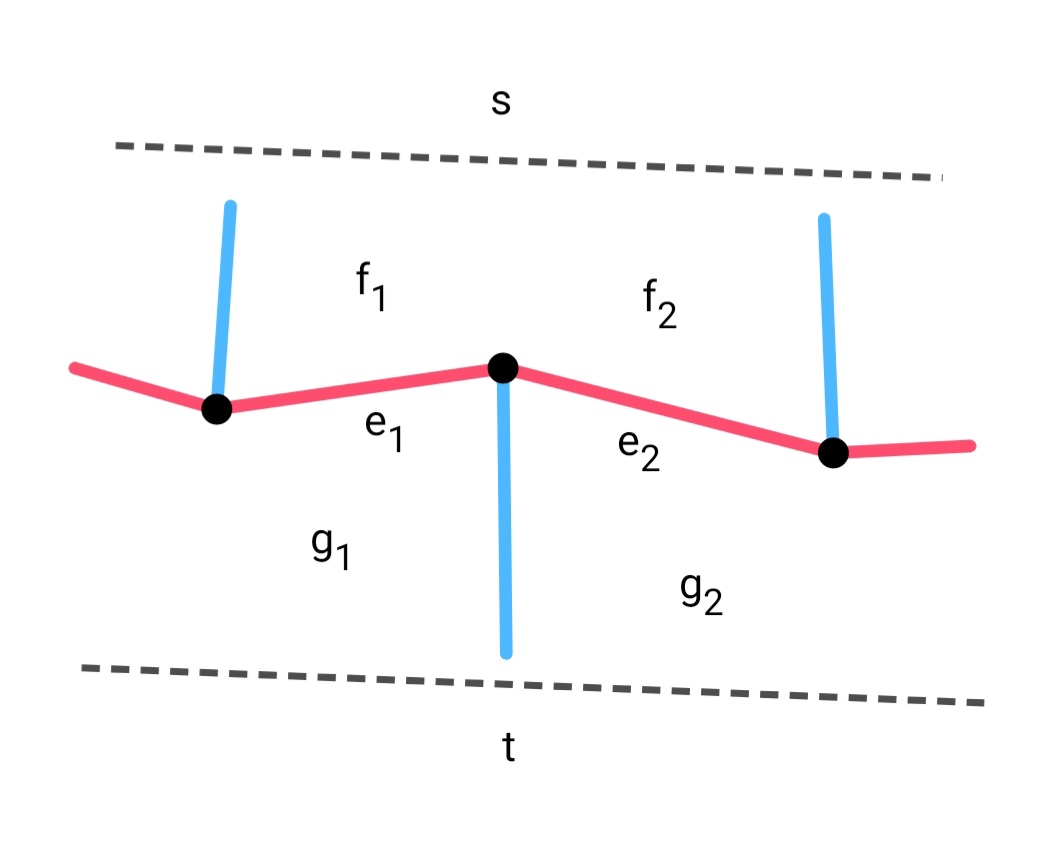}
    \caption{ The edge $e_1$ is in the faces $f_1$ and $g_1$. The edge $e_2$ is in the faces $f_2$ and $g_2$. Both faces $f_1$ and $f_2$ are in  the strip $s$ and both faces $g_1$ and $g_2$ are in the strip $t$. In this example, $f_1=f_2$.}
    \label{striboex}
    \end{subfigure}
\caption{ Steps of the demonstration of properties 5 and 6 of CDT ribbon graphs.}
\label{striex}
\end{figure}
For a regular strip or a M{\"o}bius strip, we follow the  steps below:
\begin{itemize}
    \item Consider one of its faces, call it $F_1$. $F_1$ has two spacelike boundaries, call them $A_1$ and $B_1$. 
    \item Consider also a face, call it $F_2$, that is connected to the previously considered face $F_1$ above by a timelike edge. $F_2$ also has two spacelike boundaries. Call them $A_2$ and $B_2$. (See Fig.\,\ref{strinex}.)
    \item These two boundaries are each connected to a boundary of the previous face by the vertices the faces share. Together, the two faces create two extended spacelike boundaries, $A_1\cup A_2$ and $B_1\cup B_2$. 
    \item This process is inductively repeated to $k$ connected faces in the strip, and the number of spacelike boundaries is kept two, $A_1\cup A_2\cup\cdots\cup A_k$, and $B_1\cup B_2\cup\cdots\cup B_k$.
    \item Call $l$ the number of faces in the strip. Considering property 4 of CDT ribbon graphs, when the first and the last faces are connected, there are two possibilities:
    \begin{enumerate}
         \item The two extended boundaries close into $A_1\cup A_2\cup\cdots\cup A_l$ and $B_1\cup B_2\cup\cdots\cup B_l$, forming a regular strip, as shown in Fig.\,\ref{fig:restri}. 
         \item The extended boundaries connect to each other, $A_1\cup A_2\cup\cdots\cup A_l\cup B_1\cup B_2\cup\cdots\cup B_l$, creating only one closed boundary, forming the M{\"o}bius strip as shown in \ref{fig:mostri}. 
    \end{enumerate}
\end{itemize}    

We can obtain the property 6 in a similar manner as we did property 5 above:
\begin{itemize}
    \item Given a spacelike edge $e_1$, consider the two faces (that might be the same) that share $e_1$. Call these faces $f_1$ and $g_1$. 
    \item Consider the strips (that might be the same) that contain the faces $f_1$ and $g_1$ respectively. Call them $s$ and $t$ respectively. 
    \item Consider now a spacelike edge, call it $e_2$ that has a common vertex with the previous spacelike edge $e_1$. Call $f_2$ and $g_2$ the two faces that $e_2$ belongs to.
    \item The faces $f_2$ and $g_2$ either satisfy $f_2\in s$ and $g_2\in t$, or satisfy $f_2\in t$ and $g_2\in s$. Without the loss of generality, 
    assume the first case. (See Fig.\,\ref{striboex}.)
    \item The two neighboring spacelike edges have the same two strips at its sides, $s$ and $t$.
    \item Inductively, when considering the entire closed cycle of spacelike edges, either $s \neq t$ or $s=t$. Therefore, there are either two or one strip at its sides.
\end{itemize}
\begin{figure}[t]
\begin{minipage}[t]{0.7\textwidth}
    \begin{subfigure}{0.3\textwidth}
    \centering
    \includegraphics[width=1\linewidth]{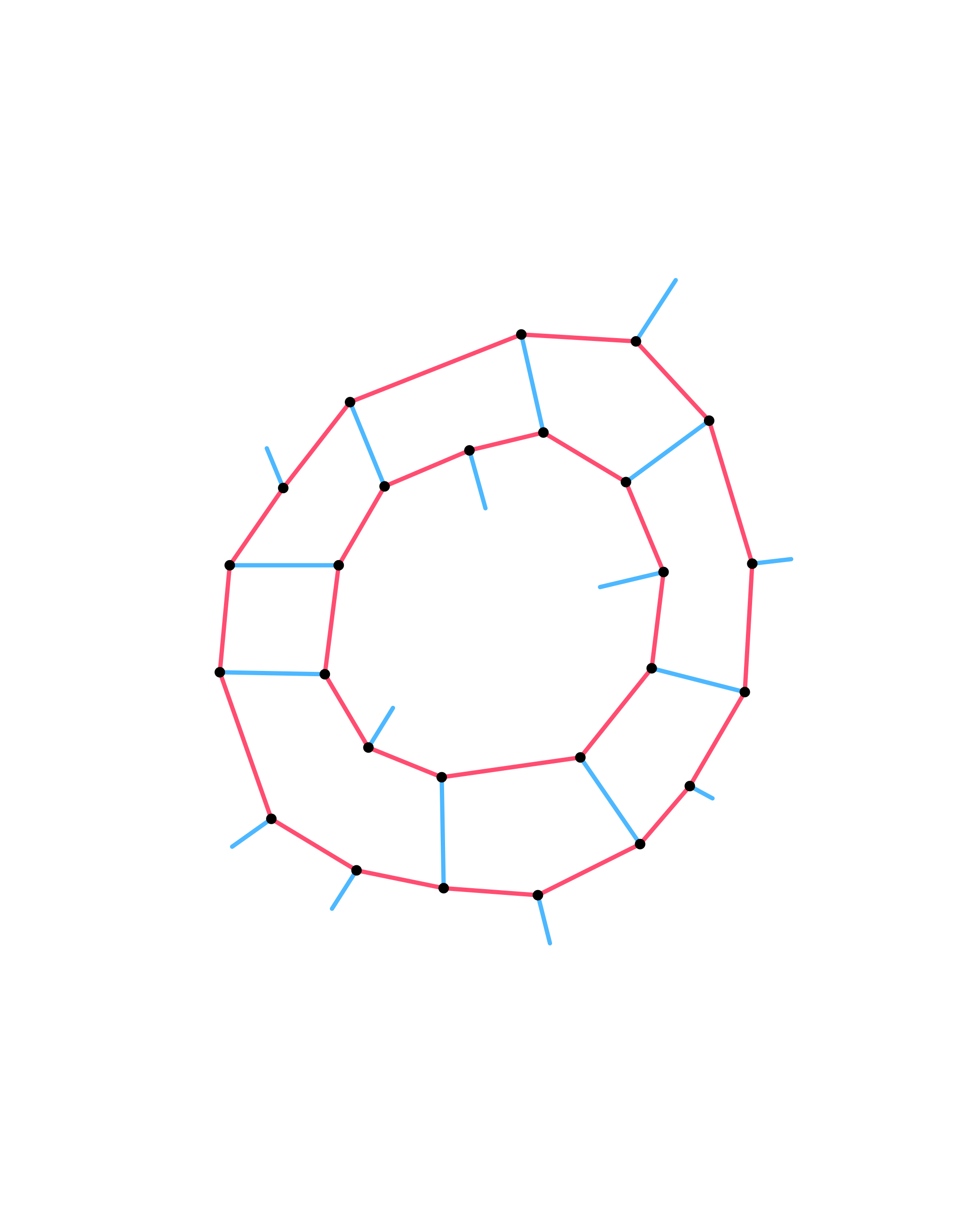}
    \caption{A regular strip.}
    \label{fig:restri}
    \end{subfigure}
    \begin{subfigure}{0.3\textwidth}
    \centering
    \includegraphics[width=1\linewidth]{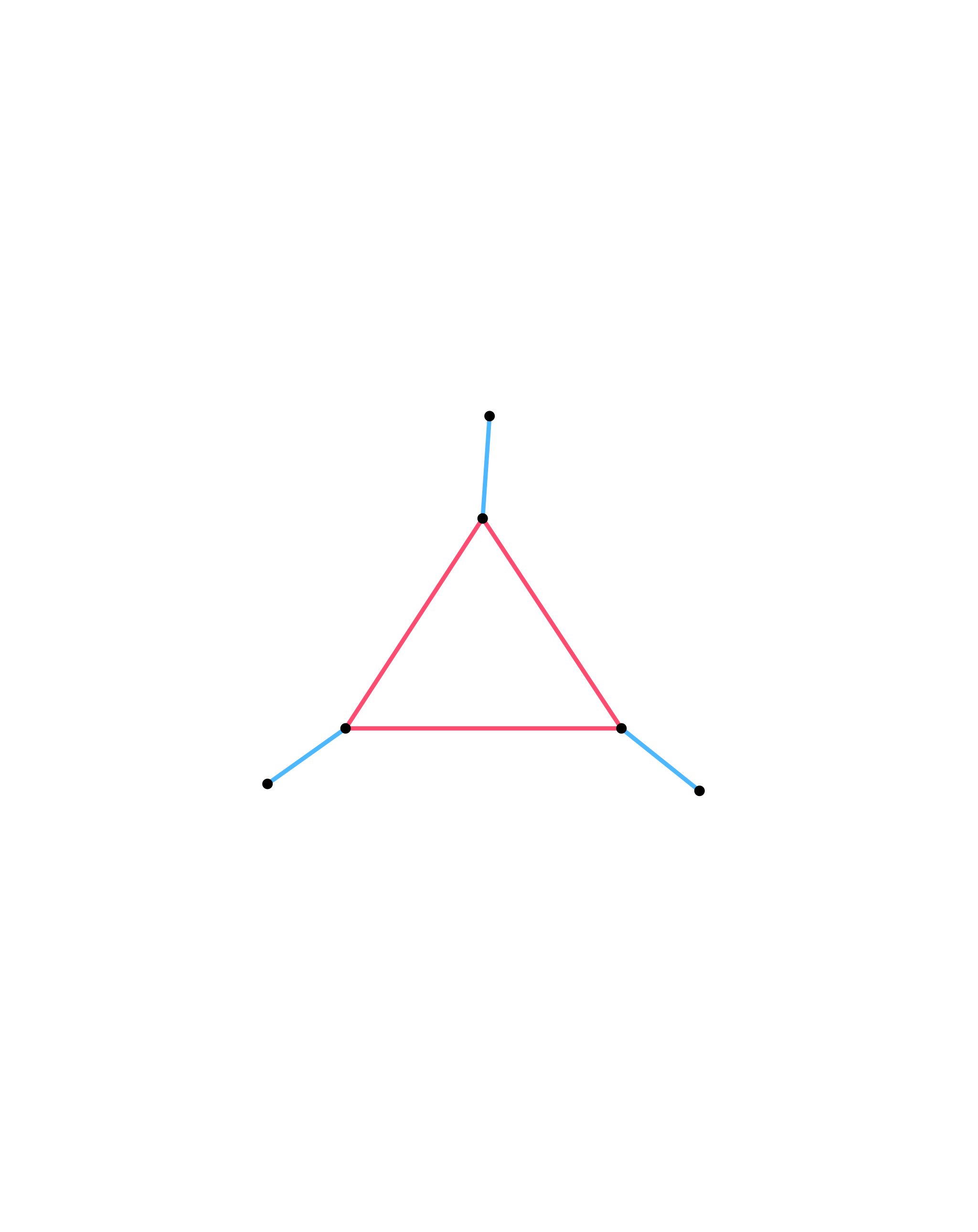}
    \caption{A singular strip.}
    \label{fig:sistri}
    \end{subfigure}
    \begin{subfigure}{0.3\textwidth}
    \centering
    \includegraphics[width=1\linewidth]{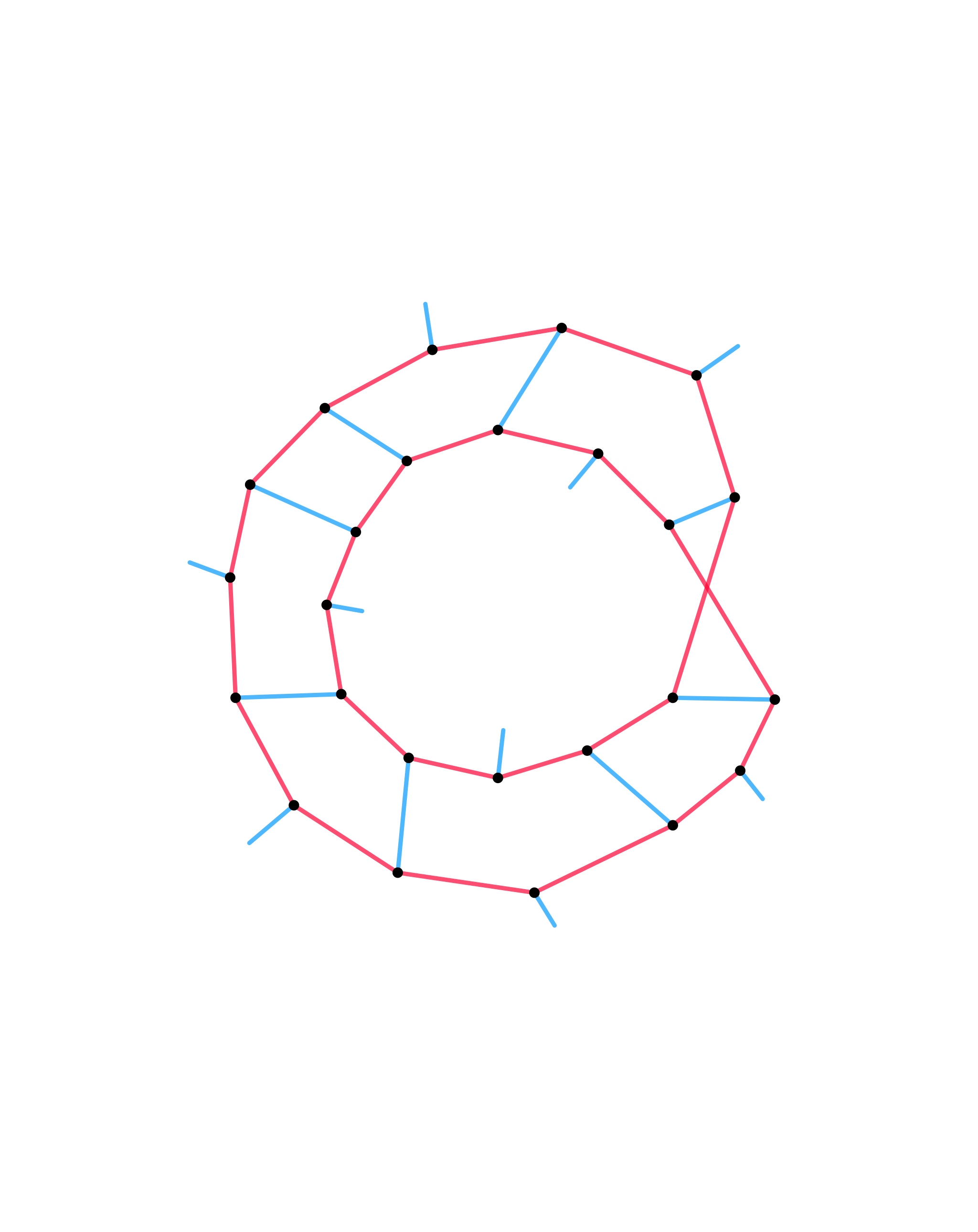}
    \caption{A M{\"o}bius strip}.
    \label{fig:mostri}
    \end{subfigure}
\caption{ Types of strips present in a CDT ribbon graph. spacelike edges are shown in red, while timelike edges are shown in blue.}
\label{fig:cdtst}
\end{minipage}
\end{figure}
Consequently, similarly to the way faces are connected by the timelike edges to form a strip, the entire graph is construted by strips connected by spacelike edges. 
Together with property 6, since the graph is connected, every strip is connected to each other by a sequence of strips. 
Therefore, one can say that a graph is a sequence of strips, and there are three types of strips defined in \ref{def1}, \ref{def2}, and \ref{def3}.
\\

\noindent {\bf {Possible topologies.}}
\\
Let us now take a look at the Euler characteristics of the CDT ribbon graphs. 
The elements of these ribbon graph can be decomposed into two disjoint sets: boundaries of strips (defined in \ref{def0}) and interiors of strips (defined in \ref{def00}).
Setting $V$ as the number of vertices, $F$ the number of faces, $E_S$ the number of spacelike edges, $E_T$ the number of timelike edges, and $E$ the total number of edges, the Euler characteristics can be written as:
\begin{equation}\label{ece}
    \chi = V-E+F = V-E_S-E_T+F=\chi_B+\chi_I,
\end{equation}
with $\chi_B=V-E_S$ being the contribution to the Euler characteristics by the boundaries of strips, and $\chi_I=F-E_T$ being the contribution by the interiors of strips. It is always true that $\chi_B=0$, since every vertex has two spacelike edges, and every edge connects two vertices, thus $2V=2E_S$. 
For $\chi_I$, we need to differentiate between different types of strips. 
For the regular strip and the M{\"o}bius strip, since their faces have two timelike edges and every edge is shared by two faces, these strips have the property $2F=2E_T$, thus $\chi_I=0$. 
As for a singular strip, since it is composed of only one face and there are no timelike edges, $\chi_I=1$. Therefore, the Euler characteristic equals the number of singular strips:
\begin{equation}
    \chi=F_s\,,
\end{equation}
where $F_s$ is the number of faces with no timelike edges, therefore it is equal to the number of singular strips. 

From properties 5 and 6 we also get an important fact: The strips can be ordered by sequencing them according to having a common boundary.
This sequencing induces global foliation, which is one of the fundamental properties of the definition of CDT \cite{Ambj_rn_1998}.
This sequence may be periodic (and therefore called infinite), just as the strips themselves are a periodic sequence of faces.
The sequence also may be finite, starting with a singular or M{\"o}bius strip, as they can only share a boundary with one other strip, and may end also with a singular or M{\"o}bius strip.
The singular strip, then, can only appear at most twice. 
Therefore, since the Euler characteristics equals the number of singular strips, the Euler characteristics can only be 0, 1 and 2.

\begin{itemize}
    \item Considering only orientable surfaces{\footnote{One can realize generating orientable two-surfaces by considering Hermitian matrix models.}}, we know that the Euler characteristics can be expressed as $\chi=2-2g$, where $g$ is the genus of the surface. The possibilities follow:
    \begin{enumerate}
        \item When the graph has two singular strips, we find that the genus is 0, thus the surface is a sphere. This graph has the sequence of strips starting with a singular strip, having some regular strips in the middle, and then ending in another singular strip. 
        
        \item When the graph has one singular strip, there is no solution since there is no orientable surface with $\chi=1$ (corresponding to $g=1/2$).
        
        \item Lastly, when the graph has no singular strips, the genus is 1, therefore the surface is a torus, composed of a periodic sequence of   regular strips.
    \end{enumerate}

    \item  When we consider nonorientable surfaces{\footnote{One can realize generating nonorientable two-surfaces by considering symmetric matrix models.}}, the Euler characteristics now assumes the form 
    \begin{equation}
    \label{eq:chinonorient}
        \chi=2-(2g+c)\,,
    \end{equation}
    where $c$ is the number of cross caps. 
    \begin{enumerate}
        \item  For a graph with two singular strips, \eqref{eq:chinonorient} can only be satisfied for $g=0$ and $c=0$. Then, the topology is of a sphere. 
        
        \item When the graph has only one singular strip, the equation \eqref{eq:chinonorient} can now be satisfied with $g=0$ and $c=1$, thus we have the topology of the projective plane. 
        
        \item Lastly, when the graph has no singular strips, the equation \eqref{eq:chinonorient} has two solutions. One solution is $g=1$ with $c=0$, implying it is a torus. The other solution is $g=0$ with $c=2$, showing the possibility of the topology of the Klein bottle.
    \end{enumerate}
\end{itemize}

\noindent {\bf {Existence and construction of topologies.}}
\\
We have shown above that these topologies are allowed, but it does not yet mean that they indeed exist. For this, let us define orientation of faces and strips:
\begin{definition}
    {\it Orientation of a face}: The edges of a face form a closed curve and thus have the usual notion of orientation. Two faces that share an edge are said to have a compatible orientation if the directions of the closed curves are opposite in the common edge. We extend this notion to any two faces by transitivity. 
\end{definition}
\begin{definition}
    {\it Orientation of a strip}: A regular strip admits a compatible orientation among all of its faces and we define the orientation of a strip as the orientation of one of its faces. A M{\"o}bius strip does not admit a compatible orientation among all of its faces. Two strips are said to have a compatible orientation if a face from one strip has compatible orientation to a face of the other strip.
\end{definition}
In the following, we argue that indeed all the possible topologies discussed above can exist by a simple construction of putting strips together to form the sequence of strips. 
\begin{enumerate}
    \item Sphere: Start the sequence with a singular strip, follow it by a sequence of regular strips, and end the sequence with another singular strip. See Fig.\,\ref{dias}.

\begin{figure}
\begin{minipage}[t]{0.8\textwidth}
    \centering
    \includegraphics[width=.5\linewidth]{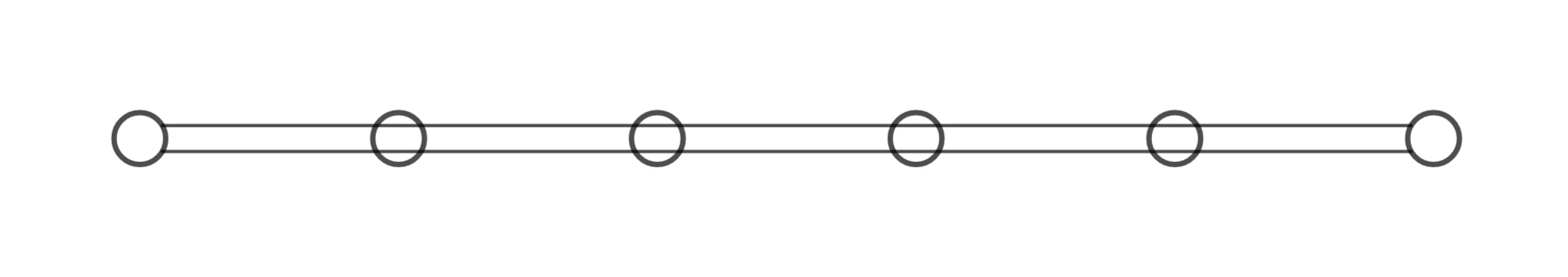}
    \caption{Construction of a sphere graph. The empty circle represents a singular strip, the circle with two lines inside represents a regular strip, and the two non-intersecting lines connecting two circles represent a boundary that keeps the orientation compatible. This example shows four regular strips, but any non-negative integer is possible.}
\label{dias}
\end{minipage}
\end{figure}
    
    \item Torus: Given a finite sequence of regular strips, glue the first strip to the last in way that the orientation is compatible. See Fig.\,\ref{diat}.

\begin{figure}
\begin{minipage}[t]{0.7\textwidth}
    \centering
    \includegraphics[width=.25\linewidth]{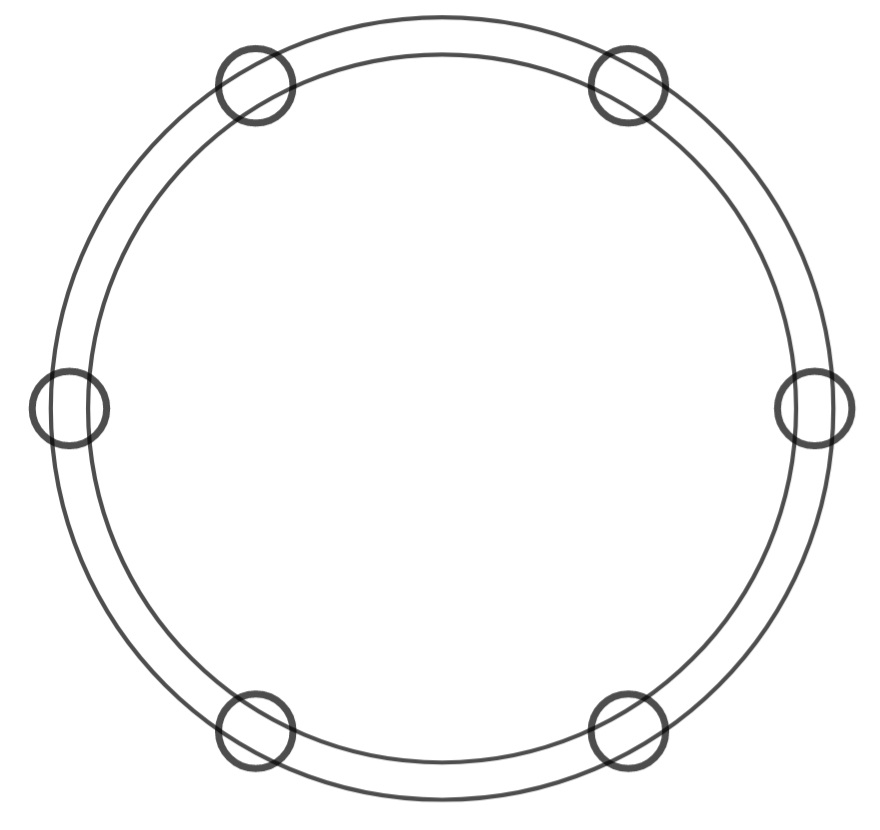}
    \caption{Construction of a torus graph. The circle with two lines represents a regular strip, and the two non-intersecting lines connecting two circles represent a boundary that keeps the orientation compatible. This example shows six regular strips, but any non-negative integer is possible.}
\label{diat}
\end{minipage}
\end{figure}

    \item Projective plane: (i) Start the sequence with a singular strip, follow it by a sequence of regular strips (or none), and end it with a M{\"o}bius strip. See Fig.\,\ref{diap1}. (ii) Alternatively, start the sequence with a singular strip, follow it by a sequence of regular strips, and glue the last boundary to itself as a cross-cap. See Fig.\,\ref{diap2}. The last boundary must have even length for this to be possible.

\begin{figure}
\begin{minipage}[t]{0.8\textwidth}
    \begin{subfigure}{.45\textwidth}
    \centering
    \includegraphics[width=\linewidth]{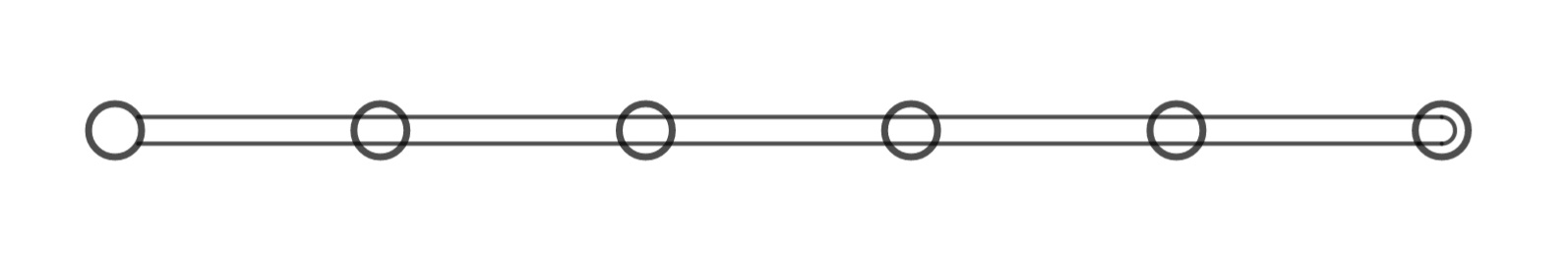}
    \caption{
    }
    \label{diap1}
    \end{subfigure}
    \begin{subfigure}{.45\textwidth}
    \centering
    \includegraphics[width=\linewidth]{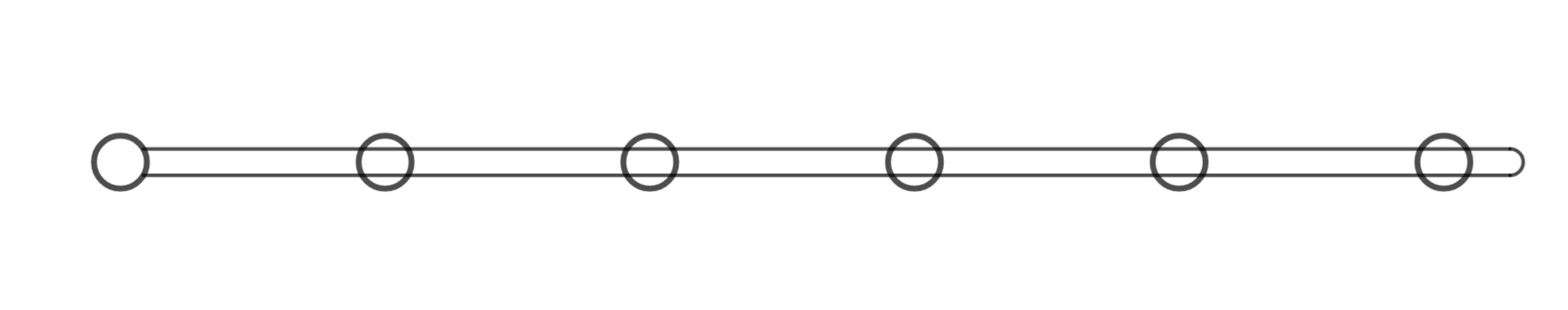}
    \caption{
    }
    \label{diap2}
    \end{subfigure}
\caption{Two different ways of constructing the projective plane. The empty circle represents a singular strip, the circle with two lines represents a regular strip, the two non-intersecting lines connecting two circles represent a boundary that keeps the orientation compatible, and the line going in and out of the same circle represents a cross-cap.
Both examples show four regular strips, but any non-negative integer is possible.
}
\label{diap}
\end{minipage}
\end{figure}

    \item Klein bottle: (i) Start the sequence with a M{\"o}bius strip, follow it by a sequence of regular strips (or none) and end it with another M{\"o}bius strip. See Fig.\,\ref{diak1}. (ii) Alternatively, we can change one or both of these M{\"o}bius strips for (a) cross-cap(s). See Figures \ref{diak2} and \ref{diak3}.
    (iii) Furthermore, we can also have a periodic sequence of regular strips, 
    and we glue the first strip to the last one in a way that their orientations are not compatible.  See Fig.\,\ref{diak4}.

\begin{figure}
\begin{minipage}[t]{0.8\textwidth}
    \begin{subfigure}{.45\textwidth}
    \centering
    \includegraphics[width=\linewidth]{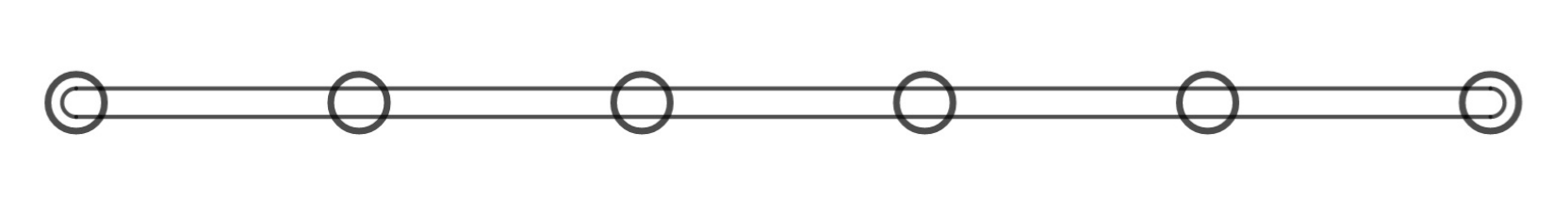}
    \caption{
    }
    \label{diak1}
    \end{subfigure}
    \begin{subfigure}{.45\textwidth}
    \centering
    \includegraphics[width=\linewidth]{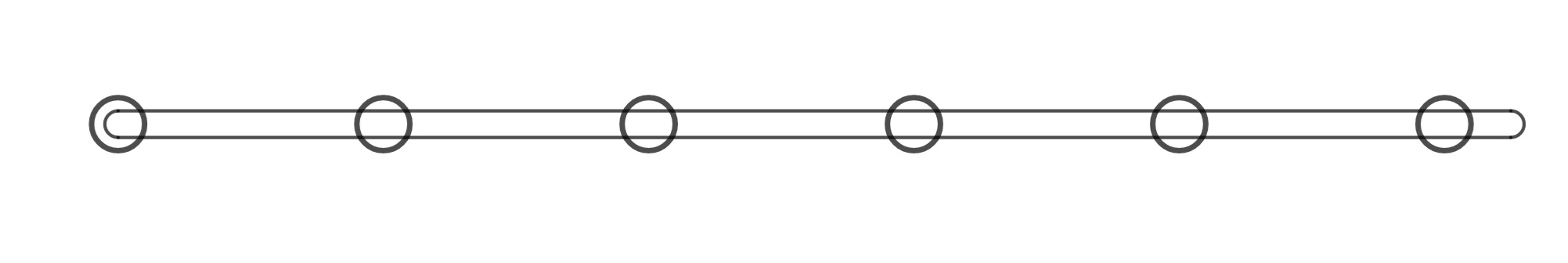}
    \caption{
    }
    \label{diak2}
    \end{subfigure}
    \begin{subfigure}{.45\textwidth}
    \centering
    \includegraphics[width=\linewidth]{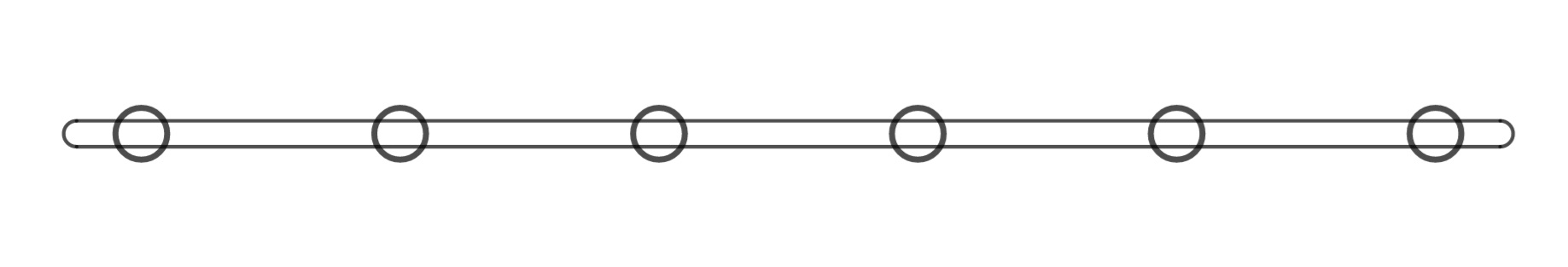}
    \caption{
    }
    \label{diak3}
    \end{subfigure}
    \begin{subfigure}{.45\textwidth}
    \centering
    \includegraphics[width=\linewidth]{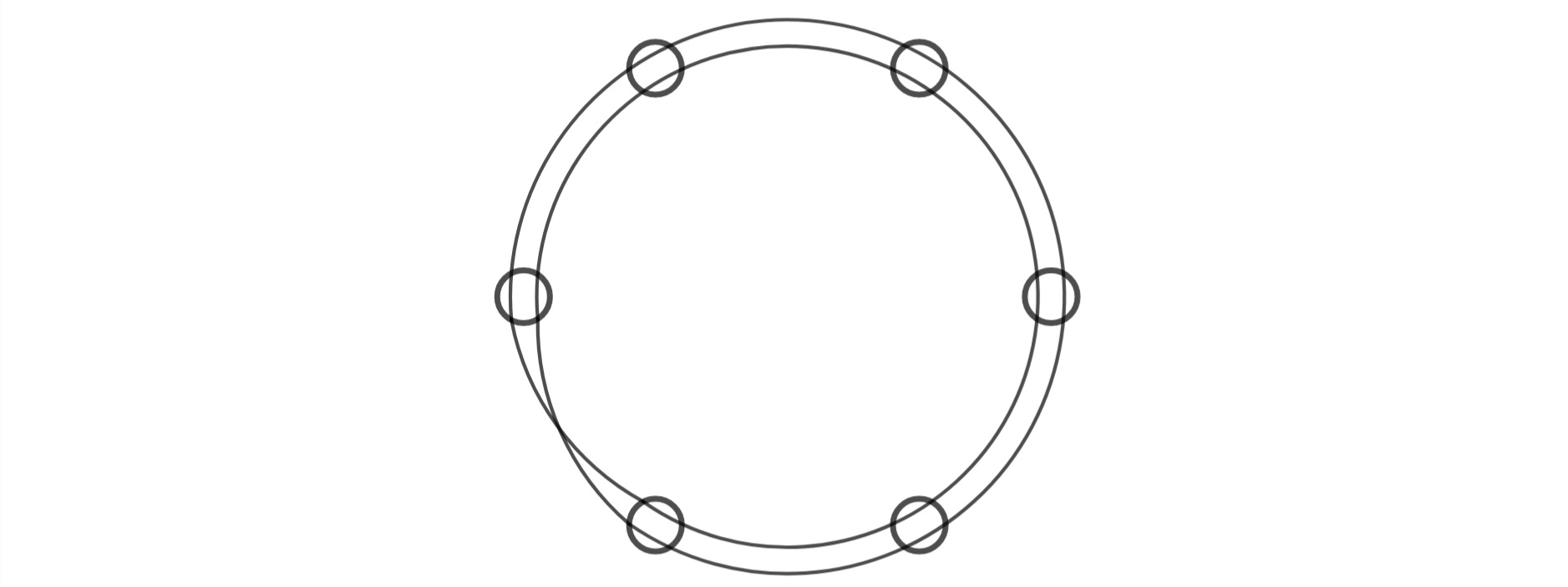}
    \caption{
    }
    \label{diak4}
    \end{subfigure}
\caption{Four different ways of constructing the Klein bottle.
The circle with two lines represents a regular strip, the circle with one line represents a M{\"o}bius strip, and the two non-intersecting lines connecting two circles represent a boundary that keeps the orientation compatible. These examples are constructed with four regular strips, but any number of them is possible.}
\label{diak}
\end{minipage}
\end{figure}
\end{enumerate}

\section{Two-matrix model representing Ising model on random two-surfaces}
\label{Sect:IM}

\begin{figure}
    \begin{subfigure}{0.3\textwidth}
    \centering
    \includegraphics[width=.5\linewidth]{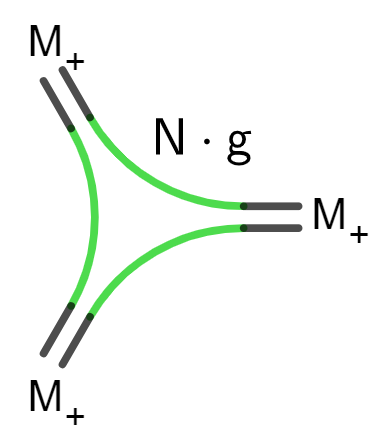}
    \caption{Spin up vertex}
    \label{fig:verve}
    \end{subfigure}
    \begin{subfigure}{0.3\textwidth}
    \centering
    \includegraphics[width=.5\linewidth]{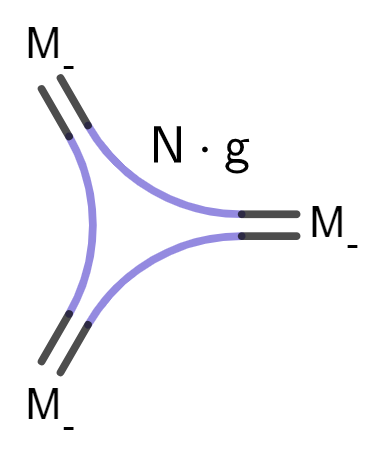}
    \caption{Spin down vertex}
    \label{fig:verpu}
    \end{subfigure}
    \begin{subfigure}{0.3\textwidth}
    \centering
    \includegraphics[width=.6\linewidth]{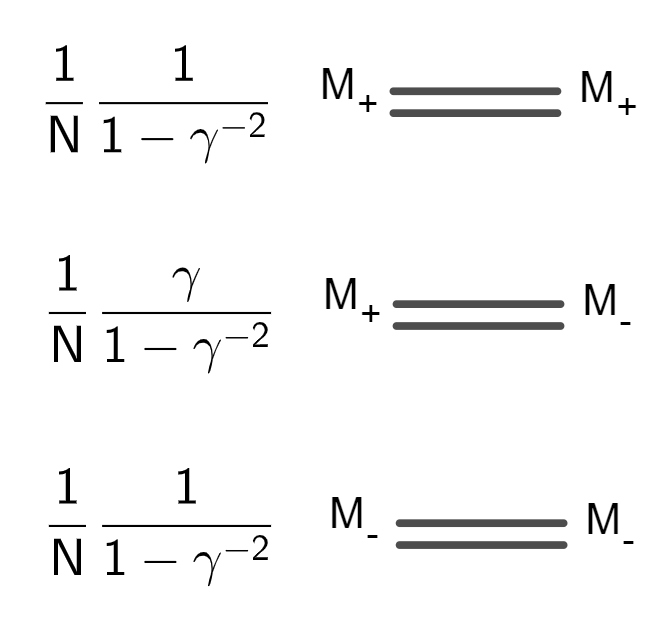}
    \caption{Edges}
    \label{fig:prop}
    \end{subfigure}
\caption{Components of the two-matrix model ribbon graph.}
\label{fig:imf}
\end{figure}

Here, let us introduce two-matrix models which have been widely studied starting with Itzykson and Zuber \cite{Itzykson:1979fi}, followed by many others \cite{Mehta:1990te, EYNARD1997633, Eynard:2002kg, Bertola:2003iw}.

Consider the partition function,
\begin{equation}
    \mathcal{Z}_{\rm IM}=N^{-2}\ln\int dM_+ dM_{-} e^{-S_{IM}}\,,
    \label{eq:partitionfcnIsing}
\end{equation}
which is a generating function of connected ribbon graphs. Here $M_+$ and $M_-$ are $N\times N$ Hermitian matrices and for a Hermitian matrix $M$, the measure applied is $dM=\prod_{i<j}\mathrm{Re}(M_{ij}) \mathrm{Im}(M_{ij}) \prod_i M_{ii}$.
For couplings $\gamma$ and $g$,
\begin{equation}
    S_{\rm IM}=N \;\mathrm{Tr}\left[\frac{1}{2}M_+^2+\frac{1}{2}M_-^2-\gamma^{-2} M_+M_--gM_+^3-gM_-^3\right]
    \label{IMaction}
\end{equation}
is said to be the action of the system \cite{KAZAKOV1986140, Boulatov:1986sb}. 
One can interpret that the spins are located on the vertices (of Feynman ribbon graphs generated perturbatively by this action) represented by the terms in the action $M_+^3$ and $M_-^3$. Fig.\,\ref{fig:imf} shows how these elements appear in the ribbon graph, and Fig.\,\ref{fig:img} shows an example of a ribbon graph represented as a graph. 
\begin{figure}[t]
\begin{minipage}[t]{0.8\textwidth}
    \centering
    \includegraphics[width=.4\linewidth]{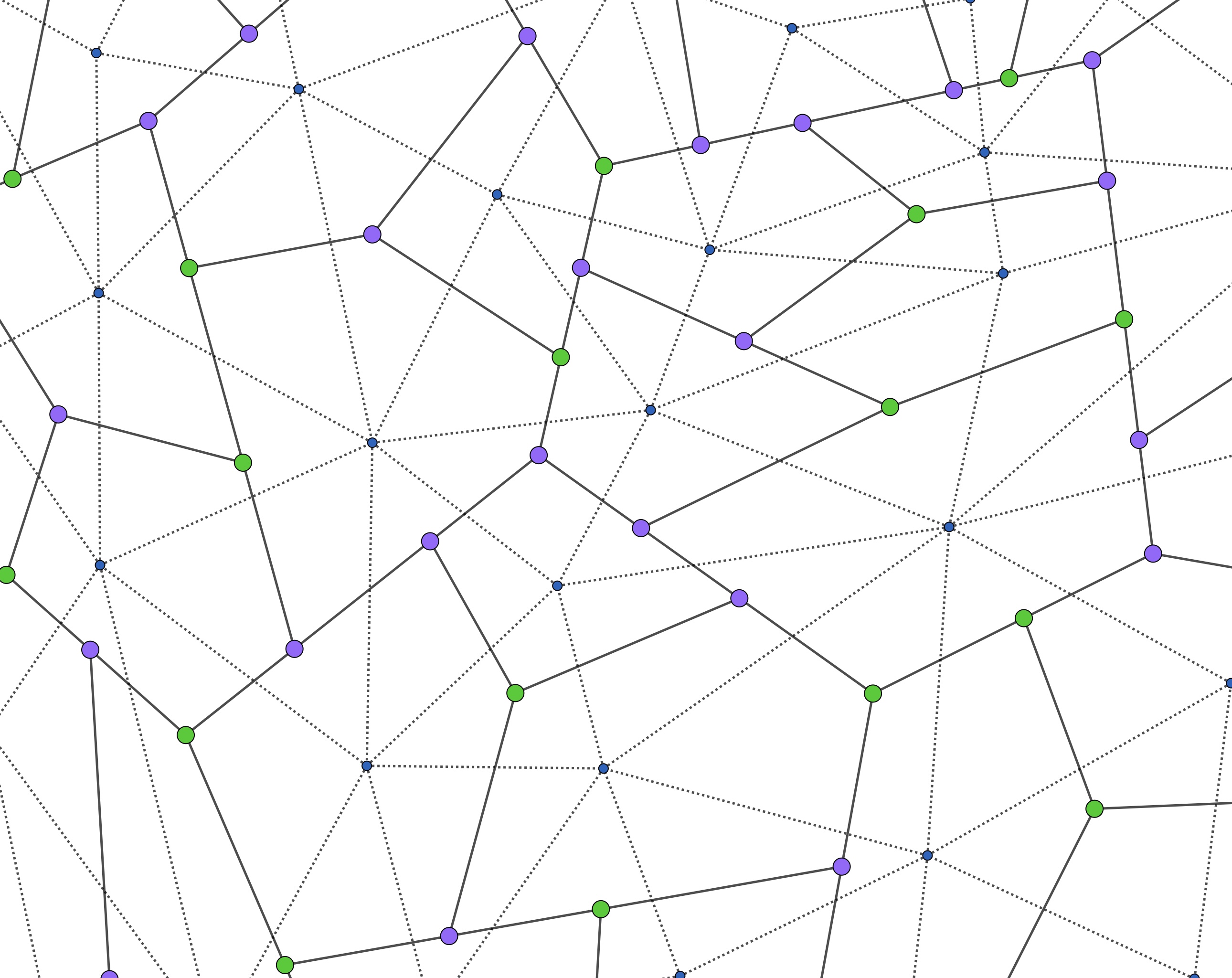}
    \caption{An example of a part of an two-matrix model graph described by the action \eqref{IMaction}. The solid lines represent the ribbon graph. The dashed lines represent the dual triangulation. The coloring of the vertices represent the spin.}
    \label{fig:img}
\end{minipage}
\end{figure}

We can interpret $g$ as a cosmological constant, since it is related to the volume of the universe. In the ribbon graph, the volume appears as the number of vertices, while in the dual graph, the triangulation, it is the number of faces. The coupling $\gamma$ can be related to the Ising temperature $T$ by $\gamma=e^{-T^{-1}}$.
For a Gaussian model with partition function
\begin{equation}
    Z_0=\int dA\;e^{-N\mathrm{Tr}\frac{1}{2}A^2}
    \,,
\end{equation}
with a Hermitian matrix $A$, 
consider the average 
\begin{equation}\label{avdef}
    \langle f(A) \rangle_0 
    = 
    \frac{1}{Z_0}
    {\int dA\;f(A)\;e^{-N\mathrm{Tr}\frac{1}{2}A^2}}
    \,,
\end{equation}
where $0$ denotes the Gaussian average. Wick's theorem says that the computation of this function can be decomposed by taking the products of the propagators related to each edge in our Feynman graphs, where the propagators are simply the mean values taken for $g=0$,
\begin{equation}
    \langle {M_{+}}_{ij}{M_{+}}_{kl} \rangle_0=\langle {M_{-}}_{ij}{M_{-}}_{kl} \rangle_0=\frac{1}{N}\frac{1}{1-\gamma^{-4}}\delta_{il}\delta_{kj}\,,
\end{equation}
and
\begin{equation}
    \langle {M_{+}}_{ij}{M_{-}}_{kl} \rangle_0=\frac{1}{N}\frac{\gamma^{-2}}{1-\gamma^{-4}}\delta_{il}\delta_{kj}\,.
\end{equation}
    Diagonalizing the kinetic term in \eqref{IMaction} takes us to another interesting version of the model, where now one can interpret that the spins are localized on the faces of the matrix model graphs (i.e., ribbon graphs) rather than on the vertices of the matrix model graphs. We achieve this by a few change of variables. First change to the Hermitian matrices $K$ and $L$ given by
\begin{equation}\label{trf1}
    M_{+}=\frac{1}{\sqrt{2}}(K+L)\;\qquad \mathrm{and}\qquad M_{-}=\frac{1}{\sqrt{2}}(K-L)\,,
\end{equation}
so that
\begin{equation}
    S_{\rm IM}=N \;\mathrm{Tr}\left[\frac{1}{2}(1-\gamma^{-2})K^2+\frac{1}{2}(1+\gamma^{-2})L^2-\sqrt{2}gK^3-3\sqrt{2}gKL^2\right]\,,
\end{equation}
and with a last substitution to the couplings $\gamma'$, $g'$, and Hermitian matrices $U$ and $V$ given by
\begin{equation}\label{trf2}
    \gamma'^{-2}=\frac{\gamma-\gamma^{-1}}{\gamma+\gamma^{-1}},\;g'=\gamma'^{-3/2}\frac{\sqrt{2}g}{(1-\gamma^{-2})^{3/2}},\;U^2=\gamma'(1-\gamma^{-2})K^2,\; V^2=\gamma'(1-\gamma^{-2})L^2,  
\end{equation}
we obtain
\begin{equation}\label{imf}
    S_{\rm IM}=N \;\mathrm{Tr}\left[\frac{1}{2}\gamma'^{-1}U^2+\frac{1}{2}\gamma' V^2-g'U^3-3g'UV^2\right]\,,
\end{equation}
which is the Ising model where the spins,  therefore $\pm$ signs,
are assigned for each face of  Feynman ribbon graphs. Here, $U$ (resp. $V$) can be thought of as being associated with a ribbon graph edge  shared by ribbon graph faces of the same (resp. opposite) parity. For a third degree vertex, since there is a face between each neighboring edges, there are three faces (some may be the same) around it.
Either the three faces have the same parity, hence the term $U^3$, or one of the three has a parity different from the other two, hence the term $3 \, U \, V^2$. For topological reasons, this interpretation is only valid for planar graphs. For example, consider the torus graph shown in Fig.\,\ref{torcon}. This graph can be generated by \eqref{imf} if we do not restrict ourselves to planar graphs. Edges associated with $U$ are in green, and edges associated to $V$ are in orange. In this configuration, the placement of the green edges around almost all faces would imply that all faces have the same spin, while the existence of orange edges would imply that some faces have opposite spins, thus a contradiction.
\begin{figure}[htb]
\begin{minipage}[t]{0.6\textwidth}
    \centering
    \includegraphics[width=.5\linewidth]{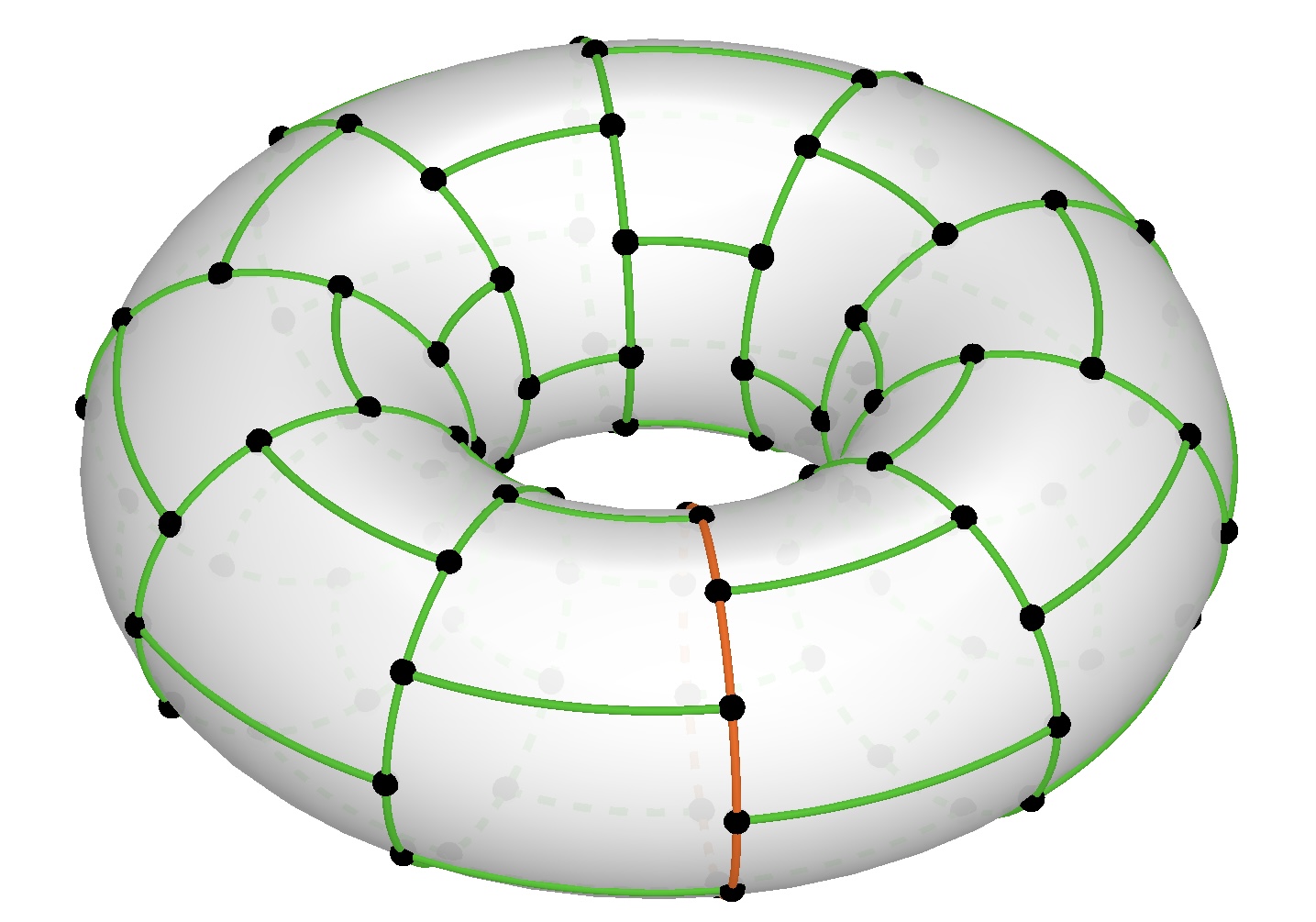}
    \caption{A torus graph with an edge configuration that does not allow an assignment of spin to the faces
    of ribbon graphs.
    }
    \label{torcon}
\end{minipage}
\end{figure}
Note however, that the action \eqref{imf} after simply the change of variables  should contain exactly the same information as the original action \eqref{IMaction}. The apparent restriction discussed above is only associated with the interpretation that we employed for $U$ and $V$ representing the same or opposite parities respectively.

\vskip 30pt

\section{Ising model coupled to CDT via matrix model}
\label{sec:CDTmmIsing}

We present a matrix model that describes Ising model on a two-dimensional manifold which has global foliation, by combining the ideas of the two models (CDT-like matrix model and Ising model) we introduced above in Sections \ref{Sect:CDTMM} and \ref{Sect:IM}. 
We would like to define an action of the matrix model which generates ribbon graphs satisfying desired properties of both CDT and Ising model. 
Hence, we require four types of half-edges to be considered: spacelike half-edges associated with spin-up vertices, spacelike half-edges associated with spin-down vertices, timelike half-edges associated with spin-up vertices, and timelike half-edges associated with spin-down vertices. 
Therefore, we introduce four matrices, $A_+$, $A_-$, $B_+$, and $B_-$ so that they induce notions of the space, time, and parity as desired. Each of these matrices is represented by a specific type of half-edges.
The following conditions are sufficient in order to implement the desired properties we described above:
To impose the existence of spacelike edges connecting equal parities, we set
\begin{equation}\label{eq1}
    \langle {A_{+}}_{ij}{A_{+}}_{kl} \rangle_0=\langle {A_{-}}_{ij}{A_{-}}_{kl} \rangle_0=\frac{1}{N}\frac{1}{1-\gamma^2}\delta_{il}\delta_{kj}\,.
\end{equation}
Similarly, for timelike edges connecting equal parities, we impose 
\begin{equation}\label{eq2}
    \langle {B_{+}}_{ij}{B_{+}}_{kl} \rangle_0=\langle {B_{-}}_{ij}{B_{-}}_{kl} \rangle_0=\frac{1}{N}\frac{1}{1-\gamma^2}{C_2}_{il}{C_2}_{kj}\,.
\end{equation}
As for spacelike edges connecting unequal parities, we write
\begin{equation}\label{eq3}
    \langle {A_{+}}_{ij}{A_{-}}_{kl} \rangle_0=\frac{1}{N}\frac{\gamma}{1-\gamma^2}\delta_{il}\delta_{kj}\,.
\end{equation}
Lastly, for timelike edges connecting unequal parities, we put
\begin{equation}\label{eq4}
    \langle {B_{+}}_{ij}{B_{-}}_{kl} \rangle_0=\frac{1}{N}\frac{\gamma}{1-\gamma^2}{C_2}_{il}{C_2}_{kj}\,.
\end{equation}
In order to make sure that only the four types of edges above \eqref{eq1}, \eqref{eq2}, \eqref{eq3}, and \eqref{eq4} exist, we set
\begin{equation}\label{eq5}
    \langle {A_{+}}_{ij}{B_{+}}_{kl}\rangle=\langle {A_{+}}_{ij}{B_{-}}_{kl}\rangle=\langle {A_{-}}_{ij}{B_{+}}_{kl}\rangle=\langle {A_{-}}_{ij}{B_{-}}_{kl}\rangle=0\,.
\end{equation}
The presence of spin-up and spin-down vertices is respectively imposed by introducing the terms $A_+^2B_+$ and $A_-^2B_-$ in the action. Therefore, 
\begin{multline}
    S_{\rm CDTIM}=\\
    N\;\mathrm{Tr}\left[\frac{1}{2}A_+^2+\frac{1}{2}\left({C_2}^{-1}B_+\right)^2+\frac{1}{2}A_-^2+\frac{1}{2}\left({C_2}^{-1}B_-\right)^2-\gamma A_+A_--\gamma({C_2}^{-1}B_+)({C_2}^{-1}B_-)-gA_+^2B_+-gA_-^2B_-\right ]\,.
\end{multline}
The partition function is then given by
\begin{equation}
    \mathcal{Z}=\lim_{N\rightarrow \infty}\ln\int dA_+dB_+dA_-dB_-\; e^{-S_{\rm CDTIM}}\;.
\end{equation}
We can also use similar transformations to \eqref{trf1} and \eqref{trf2}, that lead to \eqref{imf}, to arrive at the version where the spins are on the faces. After the transformation the action is
\begin{align}
    S_{\rm CDTIM}= N \tr\left[\frac{1}{2}\gamma^{-1}U_s^2+\frac{1}{2}\gamma^{-1}(C_2^{-1}U_t)^2+\frac{1}{2}\gamma V_s^2+\frac{1}{2}\gamma (C_2^{-1}V_t)^2
    -gU_s^2U_t-gV_s^2U_t-gU_sV_sV_t-gV_sU_sV_t\right]\;,
\end{align}
where we may interpret that $U_s$ generates spacelike edges between faces of the same spin, $U_t$ timelike edges between faces of the same spin, $V_s$ spacelike edges between faces of opposite spins, and $V_t$ timelike edges between faces of opposite spins.

\section{Properties of the matrices $C_m$}
\label{cmsec}

Let us study the properties of $C_m$ and find its expression explicitly.
 
\begin{theorem}
\label{thm:Cm}
Let $C_m$ be a $GL(N)$ matrix that satisfies $\mathrm{Tr}[(C_m)^q]=N\delta_{m,q}$ for $q=1,...,N$, with $N>0$ and $N=0 \;\mathrm{mod}\; m$. Its eigenvalues can be approximated by
\begin{equation}\label{rl}
    \lambda_{ts}=e^{\frac{2\pi}{m}is} \; \mathrm{W}(-e^{\frac{2\pi m}{N} i(t-1/2)-1})^{-\frac{1}{m}}\; \Big(1+{\mathcal O}(N^{-1/2})\Big)\,,\qquad t=1,...,N/m\;\; \text{and} \;\; s=1,...,m\;.
\end{equation}
\end{theorem}
\begin{proof}
We show \eqref{rl} by finding the characteristic polynomial of $C_m$. 
Recall \cite{mehta2004random} that, if the characteristic polynomial of a $GL(N)$ matrix $C$ is given by
\begin{equation}\label{cmp}
    p(\lambda)=\det(C-\lambda\mathbb{1})=\sum_{k=0}^N \pi_{k}(-\lambda)^{N-k}\,,
\end{equation}
then the coefficients $\pi_k$ are related to the power traces of $C$, $t_k=\mathrm{Tr}\,(C^k)$, by Girard-Newton formulas
\begin{equation}
    \pi_k= (k!)^{-1}\det[a_{ij}]_{i,j=1,...,k}\,,
\end{equation}
where
\begin{equation}
a_{ij}=\begin{cases}
t_{j-i+1} & \text{ if } i\leq j \\
j & \text{ if } i=j+1\\
0 & \text{ if } i\geq j+1\;.
\end{cases}
\end{equation}
Evaluating for $C=C_m$, we have $\pi_k=0$ for $k\neq 0$ mod $m$, and for $k=mr$, with $r$ an integer satisfying $0\leq r\leq k/m$, we get
\begin{equation}
    \pi_{mr}=\frac{(-1)^{mr}N^r}{(-m)^r r!}\;.
\end{equation}
Therefore, imposing that $N$ is divisible by $m$,
\begin{equation}\label{cha}
    \begin{split}
        p_m(\lambda)&=\sum_{r=0}^{N/m}\pi_{mr}(-\lambda)^{N-mr}=\sum_{r=0}^{N/m}\frac{(-1)^{mr}N^r}{(-m)^r r!}(-\lambda)^{N-mr}\\
        &=(-\lambda)^N\sum_{r=0}^{N/m}\frac{1}{r!}\left(-\frac{N}{m}\lambda^{-m}\right)^{r}
        =(-\lambda)^N s_{N/m}\left(-\frac{N}{m}\lambda^{-m}\right)\,,
    \end{split}
\end{equation}
where $s_n$ is the nth partial exponential sum polynomial, and we only need to investigate its zeros because we are interested in eigenvalues $\lambda_{ts}$. In the large n limit, the polynomials $s_n(n z)$ have been shown \cite{Walker, Stephen} to have the roots lying on the curve $|z e^{1-z}|=1$ with $|z|\leq1$, and  $\mathrm{arg}(z e^{1-z})$ increases monotonically from each zero. The zeroes $z_t$, $t=1,...,n$, are approximately given by $z_t=\tilde{z}_t(1+{\mathcal O}(N^{-1/2}))$, with
\begin{equation}
    {\tilde z}_t e^{1-{\tilde z}_t}=e^{\frac{2\pi}{n} i(t-1/2)}\,,\qquad t=1,...,n\,.
    \label{eq:zeros}
\end{equation}
Using the principal branch of the Lambert $W$ function, $W$, defined by $W^{-1}(z)=z e^z$, we invert \eqref{eq:zeros}
\begin{equation}
    {\tilde z}_t=-\mathrm{W}(-e^{\frac{2\pi}{n} i(t-1/2)-1})\;.
    \label{eq:zt}
\end{equation}
Now, in order to go back to eigenvalues $\lambda_{ts}$, we use $n=N/m$ and $-\lambda_{ts}^{-m}={\tilde z}_t(1+{\mathcal O}(N^{-1/2}))$ in \eqref{eq:zt}, and
\begin{equation}
    \lambda_{ts}=e^{\frac{2\pi}{m}is}\mathrm{W}(-e^{\frac{2\pi m}{N} i(t-1/2)-1})^{-\frac{1}{m}}(1+{\mathcal O}(N^{-1/2}))\,,\qquad t=1,...,N/m\;\; \text{and} \;\; s=1,...,m\;.
    \label{eq:lambdatsapp}
\end{equation}
We note that the eigenvalues are discrete, and the set is finite, with $N$ eigenvalues in total. 
\end{proof}

We let Mathematica solve numerically the zeros of the polynomial \eqref{cha}. We compare those with \eqref{eq:lambdatsapp}. We show in Fig.\,\ref{fig:c1e} and in Fig.\,\ref{fig:C2} a comparison for different values of $N$ for $C_1$ and $C_2$.
\begin{figure}[t]
\centering
    \begin{subfigure}{0.24\textwidth}
    \centering
    \includegraphics[width=1\linewidth]{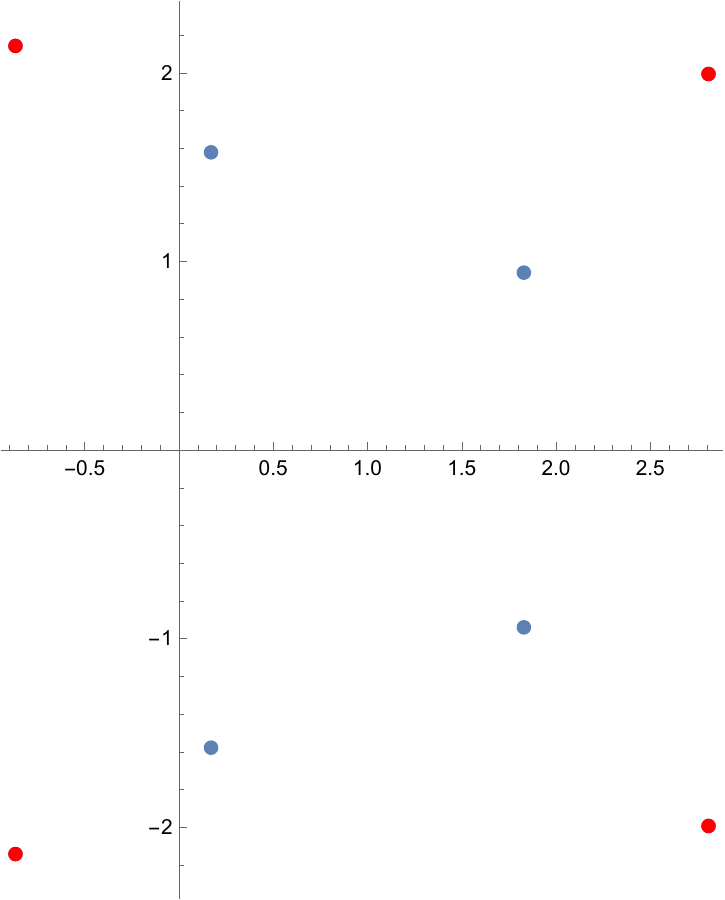}
    \caption{$N=4$}
    \label{fig:p4}
    \end{subfigure}
    \begin{subfigure}{0.24\textwidth}
    \centering
    \includegraphics[width=1\linewidth]{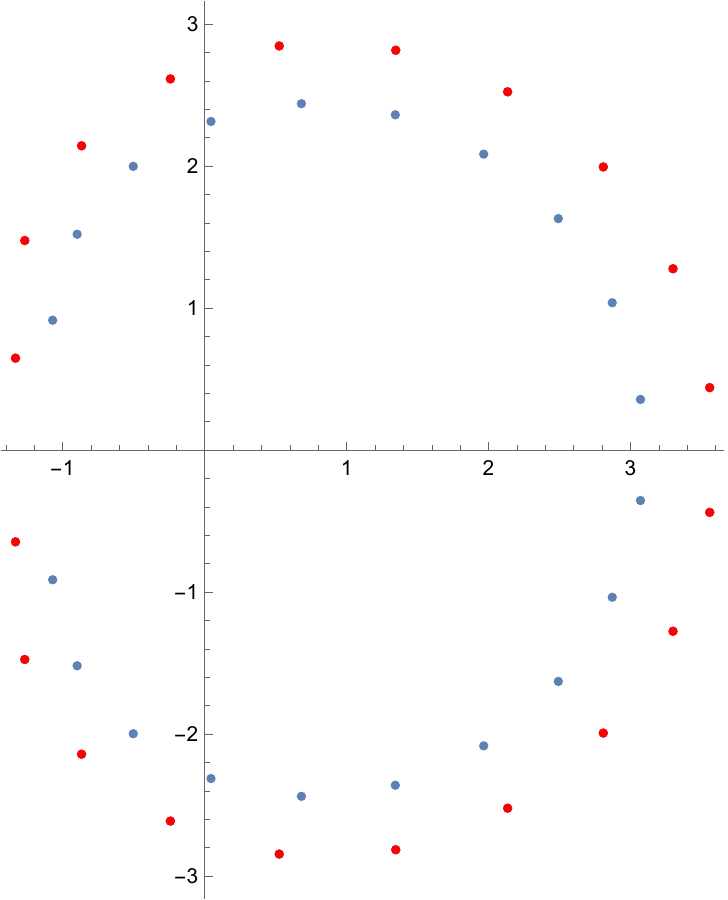}
    \caption{$N=20$}
    \label{fig:p20}
    \end{subfigure}
    \begin{subfigure}{0.24\textwidth}
    \centering
    \includegraphics[width=1\linewidth]{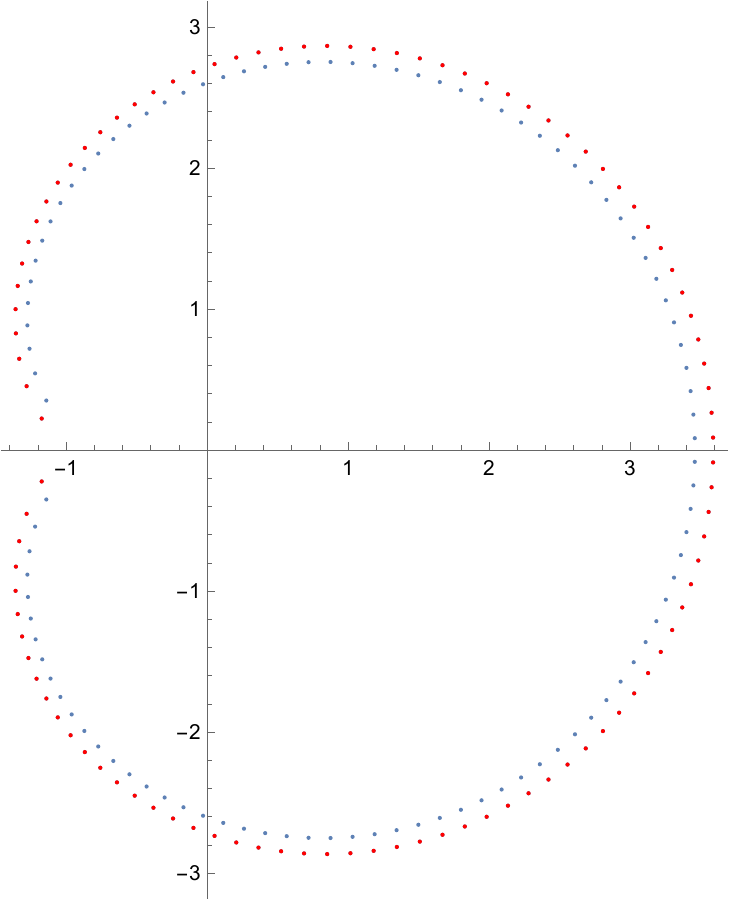}
    \caption{$N=100$}
    \label{fig:p100}
    \end{subfigure}
    \begin{subfigure}{0.24\textwidth}
    \centering
    \includegraphics[width=1\linewidth]{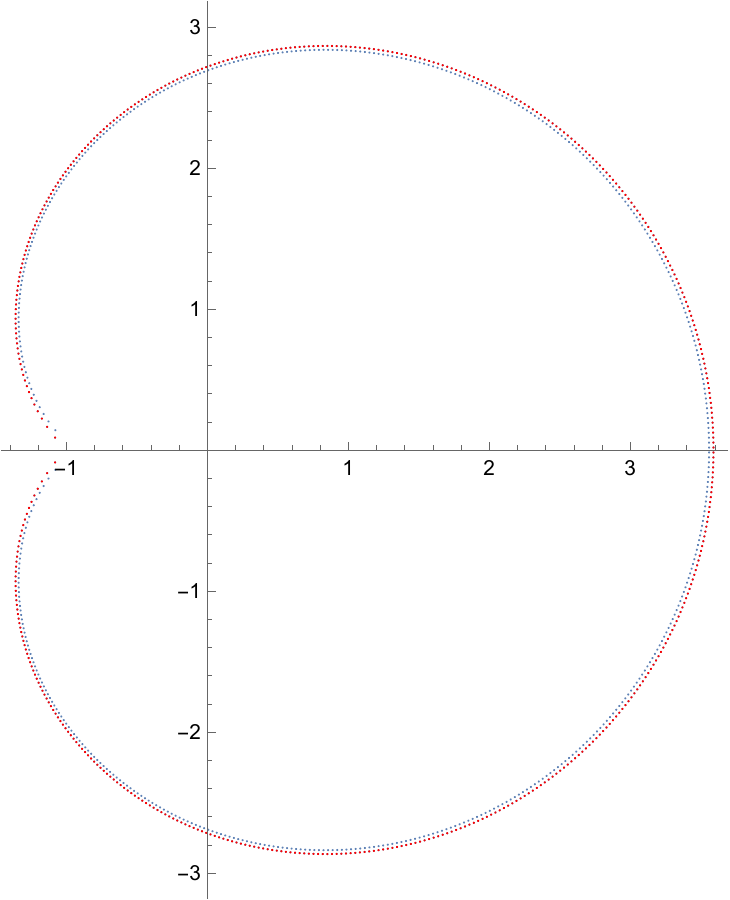}
    \caption{$N=500$}
    \label{fig:p500}
    \end{subfigure}
\caption{Eigenvalues of $C_1$ in their complex plane, with vertical axis being $\mathrm{Im}(\lambda)$ and the horizontal $\mathrm{Re}(\lambda)$. In blue, we plot the numerical values of the zeros of the polynomial $p_1(\lambda)$ in \eqref{cha} using Mathematica, and in red  we plot the values obtained from \eqref{eq:lambdatsapp} for large $N$.}
\label{fig:c1e}
\end{figure}

\begin{figure}[t]
\centering
\begin{minipage}[t]{0.8\textwidth}
    \centering
    \includegraphics[width=0.8\linewidth]{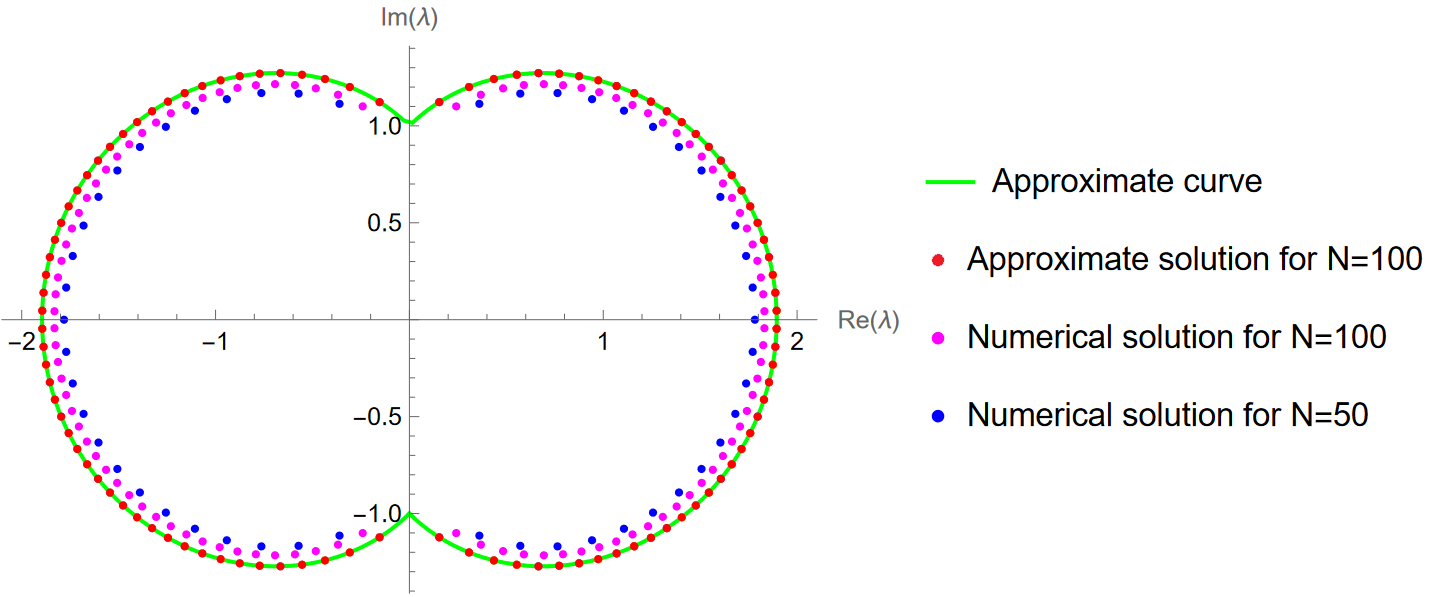}
\caption{Eigenvalues of $C_2$. We obtained the approximate curve in green by letting $t$ and $s$ continuous parameters in \eqref{eq:lambdatsapp}}.
\label{fig:C2}
\end{minipage}
\end{figure}
Taking the values of the roots of \eqref{cha} computed by Mathematica and evaluating the power traces $\tr (C_1)$ and $\tr (C_2)^2$, we indeed obtain the correct results within the machine error.
Now, let us also check the validity of the expression \eqref{eq:lambdatsapp}, by inserting the values of $\lambda_{ts}$ back to compute the power traces $\tr (C_1)$ and $\tr (C_2)^2$. We found that for $m \ne q$, we obtain $0$ up to certain power ($q_{max}$) of $C_m$, but the value of the power traces quickly increases after we reach $q_{max}$. For example, $q_{max}$ is about $20$ for $N = 10000$. We find that the value of $q_{max}$ also increases as $N$ increases.

\medskip
Let us consider the resolvent of the matrix $C_m$
\begin{equation}\label{inv1}
    \frac{1}{N}\tr\left[\frac{1}{1-\mu \, C_m}\right]=\frac{1}{N}\sum_{k=0}^\infty\tr[(\mu \, C_m)^k]=1+\mu^{m}\,.
\end{equation}
The resolvent \eqref{inv1} evaluates to 
\begin{equation}\label{inv2}
    \lim_{N \rightarrow \infty} 
    \sum_{t=1}^N \frac{1}{N} \frac{1}{1-\mu \lambda_t}
    =
    \int_0^1 dx \frac{1}{1-\mu\lambda(x)}=\int_{\gamma}d\lambda\frac{\rho(\lambda)}{1-\mu\lambda}\,, 
\end{equation}
where $x=t/N$ and $\rho(\lambda)=dx/d\lambda$ is called the spectral density, and the indices $s$ and $t$ are joined to only $t$. 
By equating \eqref{inv1}  and \eqref{inv2}, we can solve for $\rho(\lambda)$.

Let us change the variables, $\alpha=\lambda^{-1}$;
\begin{equation}\label{resv}
    \int_{\gamma}d\lambda\frac{\rho(\lambda)}{1-\mu\lambda}=\int_{\gamma'}d\alpha\,\frac{-\alpha^{-1}\rho(\alpha^{-1})}{\alpha-\mu}\,.
\end{equation}
Assuming $\gamma'$ is a simple closed curve enclosing $\mu$, using Cauchy's theorem we can use that
\begin{equation}
    1+\mu^m=\frac{1}{2\pi i}\int_{\gamma'} d\alpha \frac{1+\alpha^{m}}{\alpha-\mu}\;.
    \label{eq:cauchy}
\end{equation}
Equating \eqref{resv} and \eqref{eq:cauchy}, we find a solution
\begin{equation}
    -\alpha^{-1}\rho(\alpha^{-1})=\frac{1}{2\pi i}(1+\alpha^{m})\,, 
\end{equation}
which, in terms of $\lambda$, translates to
\begin{equation}
\label{dens}
    \rho(\lambda)=-\frac{1}{2\pi i}(\lambda^{-1}+\lambda^{-m-1})\;.
\end{equation}
At this point, let us check the validity of \eqref{eq:lambdatsapp}.
Taking \eqref{eq:lambdatsapp} to the power of $m$, 
\begin{equation}
    \lambda^m=\mathrm{W}(-e^{\frac{2\pi m}{N} i(t-1/2)-1})^{-1}\,,
    \label{eq:lambdam}
\end{equation}
to the leading order in large $N$.
Manipulating the expression \eqref{eq:lambdam}, 
\begin{equation}\label{psol}
    -\lambda^{-m} e^{1+\lambda^{-m}}=e^{\frac{2\pi m}{N} i(t-1/2)}= e^{2\pi m i x} \left(1+{\mathcal O}(N^{-1})\right)\,,
\end{equation}
where we change variables $x = t/N$.
Taking a derivative with respect to $\lambda$, 
\begin{equation}
    (-m\lambda^{-1}-m\lambda^{-m-1})(-\lambda^{-m} e^{1+\lambda^{-m}})=2\pi m i\frac{dx}{d\lambda}e^{2\pi m i x}\left(1+{\mathcal O}(N^{-1})\right)\;\,,
\end{equation}
and noting that $\rho(\lambda) = dx/d\lambda$, we again obtain the same solution \eqref{dens}.
This concludes checking the validity of the solution \eqref{eq:lambdatsapp}.

Let us go back to \eqref{dens} and continue solving for the distribution of $\lambda$.
Inserting $\rho(\lambda) = dx/d \lambda$, where $x \in {\mathbb {R}}$ with $0 \le x \le 1$ in \eqref{dens} and integrating both sides, we obtain
\begin{equation}
    \lambda^{-m}e^{\lambda^{-m}}=\kappa \, e^{2\pi m i x} \,,
    \label{eq:lambdakappa}
\end{equation}
where $\kappa$ is some integration constant.
In \eqref{eq:lambdatsapp}, $\kappa$ was determined to be $\kappa=-e^{-1}$.
Therefore, in Fig.\,\ref{fig:c1e} and  Fig.\,\ref{fig:C2} we only had one curve $\gamma$ each. However,  \eqref{eq:lambdakappa} gives us a series of solutions because $\kappa$ is free.
A larger class of solutions means that now we should have many curves representing families of eigenvalues, with each curve corresponding to a particular value of $\kappa$.
The observation of \eqref{eq:lambdakappa} is that, 
given two possible values of $m$, $m_1$ and $m_2$, their solutions are related by $\lambda_1^{m_1}=\lambda_2^{m_2}$. Therefore, we can focus on $m=1$ without losing any essential information.
For $m=1$, we have
\begin{equation}\label{ras}
    \lambda^{-1}e^{\lambda^{-1}}=\kappa \,  e^{2\pi i x}\;.
\end{equation}
Then, each value of $\kappa$ is associated with a distinct solution for $C_m$. 
The eigenvalues of $C_m$ are represented by a curve  $\gamma_\kappa$. 
Equation \eqref{ras} can be inverted with the help of the $W$ Lambert function, 
\begin{equation}
    \lambda=W_b(\kappa \, e^{2\pi i x})^{-1}\,,
\end{equation}
where $W_b$ is a branch of $W$ and $b$ is an integer we assign to each branch. 
In Fig.\,\ref{fig:w}, we show a few of the infinitely many branch regions of the $W$ function with $z=\lambda^{-1}$. 
Let us assign $b=0$ to the branch where the point $z=0$ belongs to. 
As is shown in Fig.\,\ref{fig:abs}, only curves around the origin can be closed. 
These closed curves are the contours associated with the values of $|\kappa|=|ze^z|$ with $0\leq |\kappa|\leq e^{-1}$. The maximum value $|\kappa| =e^{-1}$ is associated with the branch point of $W_0$. Our previous solution \eqref{psol} picked this particular value $|\kappa| =e^{-1}$.
However, our current analysis indicates that all these closed curves around the origin are valid not just $|\kappa| =e^{-1}$.
In the previous solution \eqref{psol}, the condition $\kappa=-e^{-1}$ appeared because of the polynomials $s_{N/m}$ in \eqref{cha}. We speculate that the reason for the existence of multiple solutions is related to the fact that the large $N$ limit condition $\lim_{N\rightarrow\infty}\frac{1}{N}\tr [(C_m)^p]=\delta_{p,m}$ can be approached through different ways.
That is, $\kappa$ may be a function of $\tr [C_m]$ and $\tr [(C_m)^m]$, and this function might not have trivial limits. For example, $\tr [{(\tilde{C}_m)}^p ]=N \delta_{p,m} +\delta_{p,1}$ has the same large $N$ limit, $\lim_{N\rightarrow\infty}\frac{1}{N}\tr ((\tilde{C}_m)^p)=\delta_{p,m}$. 
The value of $\kappa$ may vary when changing the condition on $\tr C_m$.

\begin{figure}[t]
\begin{minipage}[t]{0.9\textwidth}
    \begin{subfigure}{0.45\textwidth}
    \centering
    \includegraphics[width=\linewidth]{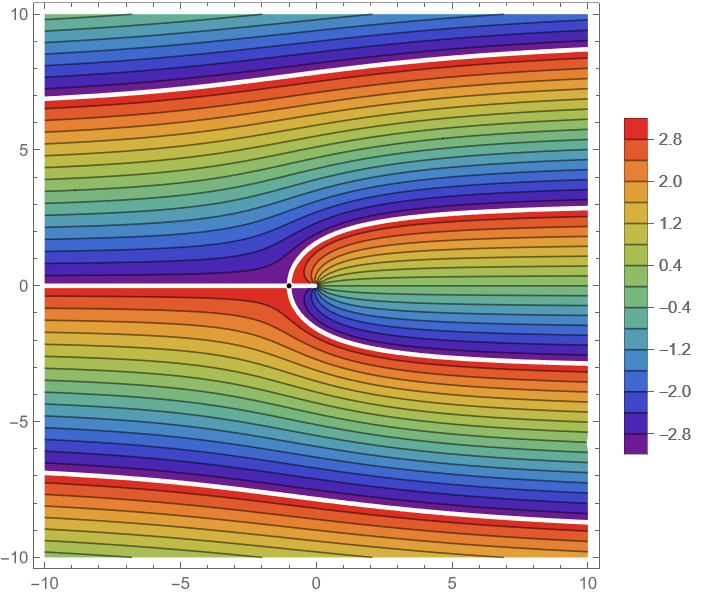}
    \caption{Contour lines of $\mathrm{arg}(ze^z)$. Complex plane of $z$ (${\mathrm {Re}}(z)$ on horizontal, and ${\mathrm {Im}}(z)$ on vertical).}
    \label{fig:arg}
    \end{subfigure}
    \begin{subfigure}{0.45\textwidth}
    \centering
    \includegraphics[width=0.97\linewidth]{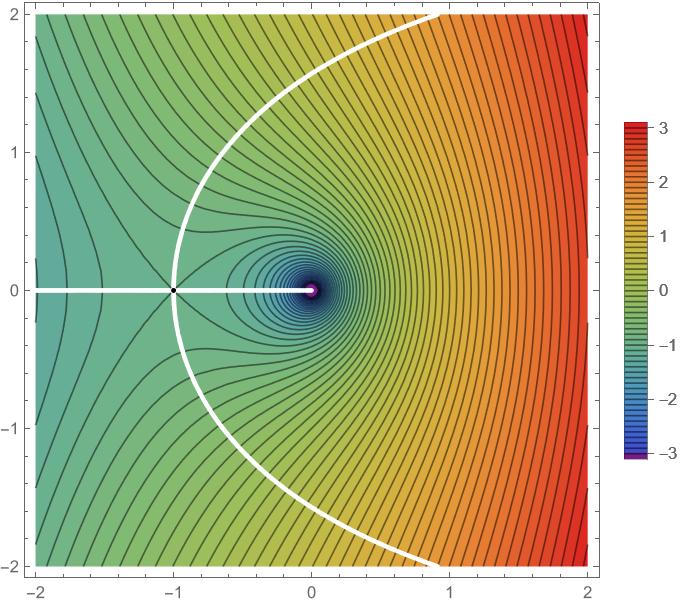}
    \caption{Contour lines of $\mathrm{ln}{|ze^z|}$. Complex plane of $z$ (${\mathrm {Re}}(z)$ on horizontal, and ${\mathrm {Im}}(z)$ on vertical).}
    \label{fig:abs}
    \end{subfigure}
\caption{Plots for the function $ze^z$. 
On the white curves, the image of $ze^z$ is on the negative real line. The black point is the branch point.}
\label{fig:w}
\end{minipage}
\end{figure}

\section{Character expansion of partition function}
\label{sec:characterexp}

For a $N \times N$ matrix $M$ with eigenvalues $M_k$, where $k \in \mathbb{Z}$ with its range $1 \le k \le N$ and given a set of $N$ increasing non-negative integers $\{h\}=\{h_1,...,h_N\}$,  
the character of this matrix can be given by
\begin{equation}\label{chard}
    \chi_{\{h\}}(M) = \frac{\det_{(k,l)}[(M_k)^{h_l}]}{\Delta (M)}\,,
\end{equation}
where
\begin{equation}\label{van}
    \Delta(M)= \det_{(k,l)}[(M_k)^{l-1}]= \prod_{1\leq i<j \leq N} (M_j-M_i)
\end{equation}
is the Vandermonde determinant. In general, the set $\{h\}$ can be used to uniquely label a representation of $\mathrm{GL}(N)$. We have as a result of character expansion 
\cite{Kazakov_1996}
\begin{equation}
    e^{\mathrm{Tr} M^2}\sim \sum_{\{h\}}c_{\{h\}}\chi_{\{h\}}(M) \,,
    \label{eq:characterexpansion}
\end{equation}
with a summation over all such sets $\{h\}$ and where, as we use from now on, $\sim$ indicates equality up to a constant. We separate the set $\{h\}$ into a set of only even integers $\{h\}^e$ and of only odd integers $\{h\}^o$. If the numbers of even and odd $h_l$ are equal or the numbers of even is one more than the number of odd, the coefficient $c_{\{h\}}$ is given by 
\begin{equation}
    c_{\{h\}}=\frac{\Delta (\{h\}^e)\Delta(\{h\}^o)}{\prod_{l=1}^{N} \left\lfloor \{h_l\}/2\right\rfloor ! }\,,
    \label{eq:charactercoefficent}
\end{equation}
where, in a similar manner to \eqref{van},  for a set of complex numbers $\{m\}=\{m_1,...,m_N\}$, the Vandermonde is defined as
\begin{equation}
    \Delta(\{m\})= \det_{(k,l)}[(m_k)^{l-1}]= \prod_{1\leq i<j \leq N} (m_j-m_i)\;.
\end{equation}
We use $\lfloor{\cdot}\rfloor$ to represent the floor function. When the condition above is not met, we simply have
\begin{equation}\label{c0}
    c_{\{h\}}=0\;.
\end{equation}
The character expansion for the more general case $e^{\mathrm{Tr}M^k}$, for integer $k\geq 1$ is also known \cite{Kazakov_1996}.

\subsection{Character expansion for the pure CDT-like matrix model of Benedetti-Henson}

The partition function \eqref{eq:CDTMMDarioABC2} for the pure CDT-like matrix model due to Benedetti-Henson studied in \cite{Benedetti:2008hc}, after integrating out $B$ is given by 
\begin{equation}
    Z = \int dA\; e^{-N \tr [\frac{1}{2}A^2-\frac{g^2}{2}(A^2C_2)^2]}\;.
    \label{eq:CDTMMDarioPFAC2}
\end{equation}
By applying the character expansion given in \eqref{eq:characterexpansion} with $M=\sqrt{\frac{N}{2}g^2}A^2C_2$, we expand the partition function \eqref{eq:CDTMMDarioPFAC2},
\begin{equation}
    Z \sim \sum_{\{h\}} \left (N\frac{g^2}{2}\right)^{\# h/2} c_{\{h\}} \int dA\; \chi_{\{h\}}(A^2C_2) e^{-N \tr [\frac{1}{2}A^2]}\,,
\end{equation}
where $\#h=\sum_{l=1}^N h_l -\frac{1}{2}N (N-1)$. With the diagonalization $A=\Omega\Lambda\Omega^\dagger$, with $\Omega \in \mathrm{U}(N)$ and $\Lambda$ a diagonal matrix, we preform a change of variables and rewrite
\begin{equation}
    Z \sim \sum_{\{h\}}\left (N\frac{g^2}{2}\right)^{\# h/2} c_{\{h\}} \int_{\mathbb{R}^N} d\Lambda\; \Delta(\Lambda)^2 e^{-N \tr [\frac{1}{2}\Lambda^2]} \int_{U(N)} d\Omega \;\chi_{\{h\}}(\Omega \Lambda^2 \Omega^\dagger C_2)\;.
\end{equation}
In order to factorize the dependence on $C_2$, let us 
use the Schur orthogonality relation, 
\begin{equation}\label{ort1}
    \int_{U(N)} d\Omega\; \chi_{\{h\}}(\Omega \Lambda^2 \Omega^\dagger C_2) = \frac{\chi_{\{h\}}(\Lambda^2)\chi_{\{h\}}(C_2)}{d_{h}}\,,
\end{equation}
where $d_{h}=\Delta(h)/\prod_{i=1}^{N-1}i!$ is the dimension of the representation given by $\{h\}$. We can then rewrite 
\begin{equation}
    Z \sim \sum_{\{h\}}\left (N\frac{g^2}{2}\right)^{\# h/2} \frac{c_{\{h\}}}{d_{h}} \, \chi_{\{h\}}(C_2) \int_{\mathbb{R}^N} d\Lambda \; \Delta(\Lambda)^2 \; \chi_{\{h\}}(\Lambda^2) \; e^{-N \tr [\frac{1}{2}\Lambda^2]} \,,
\label{eq:Zhy}
\end{equation}
now  with the integral independent of $C_2$ as desired, and $c_{\{h\}}$ is given in \eqref{eq:charactercoefficent} or \eqref{c0}, depending on which case $\{h\}$ satisfies. By changing the variables back to $A=\Omega\Lambda\Omega^\dagger$, the integral part of the expression  \eqref{eq:Zhy} can be written, up to a constant, as 
\begin{equation}
\label{in2}
\langle \chi_{\{h\}}(A^2)\rangle_0
=\frac{1}{Z_0}
\int dA\;\chi_{\{h\}}(A^2) \; e^{-N\tr \frac{1}{2} A^2}\;.
\end{equation} 
This integral had yet no known general solution in terms of the representation $\{h\}$. A conjecture was also given in \cite{Benedetti:2008hc}. It was proposed that
\begin{equation}\label{conj1}
    \langle \chi_{\{h\}}(A^2)\rangle_0 
    =
    k_N \frac{1}{N^{\#h}}\prod_{\epsilon=0}^{3}\Delta(2h^{(\epsilon)})^2\prod_i (2h_i^{(\epsilon)})!!\;,
\end{equation}
where the set of integers $\{h\}$ has been divided into four sets $\{h^{(0)}\},...,\{h^{(3)}\}$ according to equivalence modulo 4.

The conjecture stated in \eqref{conj1}, given in \cite{Benedetti:2008hc}, can be rewritten as
\begin{equation}\label{conj2}
k_N 
=
\frac
{ \langle \chi_{\{h\}}(A^2)\rangle_0 }
{ \frac{1}{N^{\#h}}\prod_{\epsilon=0}^{3}\Delta(2\{h\}^{(\epsilon)})^2\prod_i (2h_i^{(\epsilon)})!!}
     \;,
\end{equation}
where $\{h\}^{(\epsilon)}=\{a\in {\{h\}}\;|\; a=\epsilon \;\mathrm{mod}\; 4 \}$ and for a set $B$, we define $2B=\{2b\,|\,b\in B\}$. 
$k_N$ only depends on $N$, and not on $\{h\}$. 
We can test this conjecture by applying it to some representations which we already know the result. 
Fix $N=4$. For the trivial representation, which is given by $\{h\}=\{3,2,1,0\}$, the conjecture \eqref{conj2} tells us that the proportionality factor is $k_N=1/768$.
On the other hand, for the defining representation, which is given by $\{h\}=\{4,2,1,0\}$, the conjecture tells us that the proportionality factor should be $k_N=1/(8\cdot768)$.
Therefore, a contradiction.

Furthermore, let us look at representations that have size proportional to $N$. For example, the $k$th power determinant representation is one. The representation $\{h\}$ for which $\chi_{\{h\}}(A^2)=\mathrm{det}(A^{2q})$ is given by $\{h\}=\{q+N-1,q+N-2,...,q+1,q\}$. Evaluating the average integral \eqref{avdef} for $f(A)=\mathrm{det}(A^{2q})$ we obtain 
\begin{equation}
    \langle \mathrm{det}A^{2q}\rangle_0 = N^{-Nq} \prod_{k=0}^{\frac{N}{2}-1}\frac{(2q+2k+1)!!(2q+2k-1)!!}{(2k+1)!!(2k-1)!!}
\end{equation} 
According to \eqref{conj2}, we would have
\begin{equation}
k_N
=
\frac
{\prod_{k=0}^{\frac{N}{2}-1}\frac{(2q+2k+1)!!(2q+2k-1)!!}{(2k+1)!!(2k-1)!!}}
    {
     2^N2^{N\left(\frac{N}{4}-1\right)}\left(\prod_{i=0}^{\frac{N}{4}-1}i!\right)^8\prod_{k=1}^N (q+N-k)!
     }
     \;.
\end{equation}
We then see that $k_N$ depends on the representation through $q$, which again goes against the assumption that $k_N$ only depends on $N$.

\subsection{Character expansion for CDT-like matrix model coupled with Ising model}
For CDT-like matrix model coupled with Ising model, the expression is more involved. We write its partition function as
\begin{equation}
Z=\int dA_1dA_2 \;e^{-N\tr[\frac{1}{2}A_1^2+\frac{1}{2}A_2^2-\gamma A_1A_2-\frac{1}{2}\frac{g^2}{1-\gamma}(A_1^2C_2)^2-\frac{1}{2}\frac{g^2}{1-\gamma}(A_2^2C_2)^2-\frac{1}{2}\frac{\gamma g^2}{1-\gamma^2}((A_1^2+A_2^2)C_2)^2]}\,,
\end{equation}
and diagonalize the kinetic part,
\begin{equation}
\label{2nd}
Z=\int dUdV \;e^{-N\tr[\frac{1}{2}(1-\gamma)U^2+\frac{1}{2}(1+\gamma)V^2-\frac{1}{4}\frac{g^2}{1-\gamma}((U^2+V^2)C_2)^2-\frac{1}{4}\frac{g^2}{1+\gamma}((UV+VU)C_2)^2]}\,,
\end{equation}
where $U=(A_1+A_2)/\sqrt{2}$ and $V=(A_1-A_2)/\sqrt{2}$. 
The character expansion \eqref{eq:characterexpansion} for the latter partition function \eqref{2nd} 

\begin{equation}
\label{doubleexpansion}
    Z \sim  \sum_{\{h^1\},\{h^2\}}\left(\frac{N g^2}{4(1-\gamma)}\right)^{\# h^1/2}\left(\frac{N g^2}{4(1+\gamma)}\right)^{\# h^2/2} c_{\{h^1\}}c_{\{h^2\}} 
    \; 
    I_{\{h^1\},\{h^2\}} 
\end{equation}
where the sum is performed over representations labeled by $\{h^1\}$ and $\{h^2\}$, and we defined the integral
\begin{equation}
        I_{\{h^1\},\{h^2\}}=\int dUdV\;e^{-N\tr[\frac{1}{2}(1-\gamma)U^2+\frac{1}{2}(1+\gamma)V^2]}
        \; \chi_{\{h^1\}}((U^2+V^2)C_2) \; \chi_{\{h^2\}}((UV+VU)C_2)\;.
\end{equation}
Transform $U=\Omega_1\Lambda_1\Omega_1^\dagger$ and $V=\Omega_2\Lambda_2\Omega_2^\dagger$,
and furthermore, by substituting $\Omega_2=\Omega_1\Omega$, we have
\begin{equation}
    \begin{split}
        I_{\{h^1\},\{h^2\}}=&\int_{\mathbb{R}^N\times\mathbb{R}^N} d\Lambda_1d\Lambda_2\;\Delta(\Lambda_1)^2\Delta(\Lambda_2)^2e^{-N\tr[\frac{1}{2}(1-\gamma)\Lambda_1^2+\frac{1}{2}(1+\gamma)\Lambda_2^2]}\\
        &\int_{U(N)\times U(N)} d\Omega_1d\Omega\;\chi_{\{h^1\}}(\Omega_1(\Lambda_1^2+\Omega\Lambda_2^2\Omega^{\dag})\Omega_1^{\dag}C_2) \; \chi_{\{h^2\}}(\Omega_1(\Lambda_1\Omega\Lambda_2\Omega^{\dag}+\Omega\Lambda_2\Omega^{\dag}\Lambda_1)\Omega_1^{\dag}C_2)\,,
    \end{split}
    \label{eq:Ih1h2}
\end{equation}
that allows us to integrate over $\Omega_1$, however, the orthogonality relation  \eqref{ort1} is not enough, but one has to be equipped with a more general expression for Schur orthogonality relation.

\section{Unitary integral of degree $4$ monomial}
\label{sec:unitintmono}

A more general way to state the Schur orthogonality relation is that for a compact group $G$ and for matrix elements $\phi^\alpha_{ij}(g)$ of an element $g \in G$, where the upper index indicates the representation and the lower indices indicate the corresponding elements of a given orthogonal basis,
\begin{equation}
    \int_G \phi^{\alpha}_{i j}(g) \overline{\phi^{\beta}_{k l}(g)} dg = \frac{1}{d_{\alpha}}\delta_{\alpha \beta}\delta_{ik}\delta_{jl}.
    \label{eq:generalSchurOrtho}
\end{equation}
Here, with $G$ as $U(N)$ and for $\alpha=\beta=h$ by multiplying both sides by $\phi^h_{jl}(\Lambda^2)$ and $\phi^h_{ki}(C_2)$ and summing over the repeated lower indices, we recover \eqref{ort1}. Let us apply \eqref{eq:generalSchurOrtho} to the unitary integral of \eqref{eq:Ih1h2}, by defining the matrices which are independent of $\Omega_1$ as $M_1=\Lambda_1^2+\Omega\Lambda_2^2\Omega^{\dag}$ and $M_2=\Lambda_1\Omega\Lambda_2\Omega^{\dag}+\Omega\Lambda_2\Omega^{\dag}\Lambda_1$.
Then, we can factorize the expression
\begin{equation}
    \begin{split}
        \chi_{\{h^1\}}(\Omega_1M_1\Omega_1^{\dag}C_2)
        \;
        \chi_{\{h^2\}}(\Omega_1M_2\Omega_1^{\dag}C_2)=\;\phi^{h^1}_{ab}(\Omega_1)\phi^{h^1}_{bc}(M_1)\phi^{h^1}_{cd}(\Omega^{\dag}_1)\phi^{h^1}_{da}(C_2)\phi^{h^2}_{pq}(\Omega_1)\phi^{h^2}_{qr}(M_2)\phi^{h^2}_{rs}(\Omega_1^\dag)\phi^{h^2}_{sp}(C_2)\;.
    \end{split}
\end{equation}
Then, in \eqref{eq:Ih1h2}, we have an integral
\begin{equation}
   I_{ij}= \int_{U(N)} \phi^{\alpha}_{a b}(g) \phi^{\beta}_{c d}(g)\overline{\phi^{\alpha}_{\bar{a} \bar{b}}(g)\phi^{\beta}_{\bar{c} \bar{d}}(g)} dg\,,
   \label{eq:Iij}
\end{equation}
where $g =\Omega_1$.
This integral has some properties similar to the previous one \eqref{eq:generalSchurOrtho}, but has the double amount of matrix elements. This integral is also similar to the type of integral studied with the Weingarten calculus, however, here, more than one representations $\alpha$ and $\beta$ appear.

Putting all together, \eqref{eq:Ih1h2} becomes
\begin{multline}
    I_{\{h^1\},\{h^2\}}=I_{ij}\;
        \phi^{h^1}_{da}(C_2)\, 
        \phi^{h^2}_{sp}(C_2) \, \times\\\times
        \int_{\mathbb{R}^N\times\mathbb{R}^N} d\Lambda_1d\Lambda_2\;\Delta(\Lambda_1)^2\Delta(\Lambda_2)^2e^{-N\tr[\frac{1}{2}(1-\gamma)\Lambda_1^2+\frac{1}{2}(1+\gamma)\Lambda_2^2]}\int_{U(N)} d\Omega\;\phi^{h^1}_{bc}(M_1)\phi^{h^2}_{qr}(M_2)\,,
\end{multline}
and with the theorem \ref{unitaryint} below, we can evaluate $I_{ij}$.

\begin{proposition}
\label{unitaryint}
Given $g\in U(N)$ and a representation $\alpha$, set its matrix elements as $\phi^\alpha_{ab}(g)$, where the indices $a$ and $b$ are associated with some chosen basis. The following holds true for a Haar integral:
\begin{equation}
     \int_{U(N)} \phi^{\alpha}_{a b}(g) \phi^{\beta}_{c d}(g)\overline{\phi^{\alpha}_{\bar{a} \bar{b}}(g)\phi^{\beta}_{\bar{c} \bar{d}}(g)} dg=\sum_{r,m,n,p,q}{c^{rm*}_{acp}}c^{r n}_{\m a\m c p}{c^{rm}_{bdq}}c^{r n*}_{\m b\m d q}d_r^{-1}\;.
     \label{eq:haarint}
\end{equation}
\end{proposition}

\begin{proof}
In order to apply the Weingarten calculus \cite{collins2021weingarten} to the left hand side of \eqref{eq:haarint}, we need to know a basis for the invariant space of $\alpha\,\otimes\,\beta\,\otimes\,\bar{\alpha}\,\otimes\,\bar{\beta}$.
Indeed, finding the basis for the invariant space amounts to identifying the trivial representations $1^{\oplus l_1^{\alpha\beta \bar\alpha \bar\beta}}$, upon decomposing the tensor product representation into irreducible representations, 
\begin{equation} \alpha\,\otimes\,\beta\,\otimes\,\bar{\alpha}\,\otimes\,\bar{\beta}=\bigoplus_{\mathrm{all}\,\mathrm{irrep}\,r}r^{\oplus l_r^{\alpha\beta \bar\alpha \bar\beta}}\,,
\end{equation}
where $l_r^{\alpha\beta\bar\alpha\bar\beta}$ is the number of times the representation $r$ appears in the expansion of the tensor product, known as Littlewood-Richardson coefficients. 
One can decompose separately the products $\alpha \,\otimes\,\beta$ and $\bar{\alpha} \,\otimes\,\bar{\beta}$, i.e., 
\begin{equation}\label{dec}
    \alpha \,\otimes\,\beta = \bigoplus_r r ^{\oplus l_r
    ^{\alpha\beta}}\,, \qquad
    \bar{\alpha} \,\otimes\,\bar{\beta} = \bigoplus_{\bar{r}} \bar{r} ^{\oplus l_{\bar{r}}^{\bar{\alpha}\bar{\beta}}}
    = \bigoplus_r \bar{r} ^{\oplus l_{\bar{r}}^{\bar{\alpha}\bar{\beta}}} \,.
\end{equation}
Then, using $l_{\bar{r}}^{\bar{\alpha}\bar{\beta}}=l_{r}^{\alpha\beta}$, 
we can write 
\bea
\label{de2}
    \alpha\,\otimes\,\beta\,\otimes\,\bar{\alpha}\,\otimes\,\bar{\beta}
    &=&
    \bigoplus_{r,s} r ^{\oplus l_r
    ^{\alpha\beta}}\,\otimes\,\bar{s} ^{\oplus l_{\bar{s}}^{\bar{\alpha}\bar{\beta}}}
    \crcr
    &=&
    \bigoplus_{r} r ^{\oplus l_r
    ^{\alpha\beta}}\,\otimes\,\bar{r} ^{\oplus l_r
    ^{\alpha\beta}}\,\oplus\,\bigoplus_{r\neq s} r ^{\oplus l_r
    ^{\alpha\beta}}\,\otimes\,\bar{s} ^{\oplus l_s
    ^{\alpha\beta}}\,,
\eea
but the representation orthogonality relation tells us that for two irreducible representations $r$ and $s$, $l_1^{r\bar{s}}=\delta_{r s}$.
Therefore, 
we notice that the second big sum contains no trivial representations and the first contains one trivial representation for each product.
In the end, we arrive at
\begin{equation}
    \alpha\,\otimes\,\beta\,\otimes\,\bar{\alpha}\,\otimes\,\bar{\beta}=\bigoplus_{r} 1^{\oplus(l_r^{\alpha\beta})^2}\oplus \bigoplus_{\substack{\mathrm{all}\,\mathrm{irrep}\,r\\r\neq 1}}r^{\oplus l_r^{\alpha\beta \bar\alpha \bar\beta}}\;.
\end{equation}
Therefore, we conclude
\begin{equation}\label{dim}
    l_1^{\alpha\beta \bar\alpha \bar\beta}=\sum_r (l_r^{\alpha\beta})^2\,.
\end{equation}
Remark that there are infinite number of representations to sum over, however, only finite number of them have nonzero $l_r^{\alpha \beta}$. 
Equation \eqref{dim} is important to us because it allows to verify that the set we find later in \eqref{td} for the invariant space is indeed a basis because it contains as many elements as the dimension of the invariant space.

Let us analyze a basis for the invariant space $\bigoplus_{r} 1^{\oplus(l_r^{\alpha\beta})^2}$, since it allows us to evaluate $I_{ij}$ in \eqref{eq:Iij}.
First, let us define a basis for the vector space of the tensor product representation $\alpha\,\otimes\,\beta\,\otimes\,\bar{\alpha}\,\otimes\,\bar{\beta}$,
\begin{equation}\label{ten}
    |E_i\rangle = |E_{ab\,\bar{a}\bar{b}} \rangle = |e_a\rangle |e_b\rangle \langle e_{\bar{a}}|\langle e_{\bar{b}}|  \,,
\end{equation}
 where $a$, $\bar a$, $b$, $\bar b$, and $i$ are integers and $i$ is a relabelling of indices $a b \bar a \bar b$ with $1\leq i \leq (d_\alpha)^2(d_\beta)^2$, and $|e_a\rangle $, $|e_b\rangle$, $\langle e_{\bar{a}}|$ and $\langle e_{\bar{b}}|$, with $1\leq a\leq d_\alpha$, $1\leq \bar a \leq d_\alpha$, $1\leq b \leq d_\beta$, $1\leq \bar b \leq d_\beta$, are respectively a basis for representations $\alpha$, $\beta$, $\bar\alpha$ and $\bar\beta$. 
Writing $T_\mu$ as a basis  for this invariant space, where $1 \le \mu \le l_1^{\alpha\beta \bar\alpha \bar\beta}$,
let us define the $A$ matrix
\begin{equation}\label{am}
    A_{i\mu}=\langle E_i|T_\mu\rangle \,,
\end{equation}
and the Weingarten matrix $w$ through the inverse of a matrix 
\begin{equation}
    w^{-1}_{\mu \nu} = \langle T_\mu | T_\nu \rangle\;.
    \label{eq:Weingartenmatrix}
\end{equation}
The Weingarten theorem states that \eqref{eq:Iij} can be expressed as
\begin{equation}\label{teo}
    I_{ij}=\sum_{\mu,\nu}A_{i\mu} w_{\mu\nu}A^\dag_{\nu j}\,, 
\end{equation}
which is a function of $T$.
$T$ being invariant under $\alpha\,\otimes\,\beta\,\otimes\,\bar{\alpha}\,\otimes\,\bar{\beta}$ 
is equivalent to 
\begin{equation}
    M T M^{-1} = T\,,
\end{equation}
for every $M\in \alpha\,\otimes\,\beta$.
Furthermore,  there exists a $V \in U(d_\alpha \, d_\beta)$ such that, for every $M$,  
$V$ transforms $M$ into a block diagonal matrix,
\begin{equation}
    VMV^{-1}=M_D,
\end{equation}
where, 
\begin{equation}
    M_D=\begin{pmatrix}
M_{r_1}&0&\cdots&0\\
0&M_{r_2}&\cdots&0\\
\vdots&\vdots&\ddots&\vdots\\
0&0&\cdots&M_{r_n}\\
\end{pmatrix}\,,
\end{equation}
with $n=\sum_r l_r^{\alpha\beta}$.
For every representation $r$ appearing in the decomposition of $\alpha\,\otimes\,\beta$ as in \eqref{dec}, $M_r$ is $M$ in the specific representation $r$ and appears as a block element in the diagonal of $M_D$, appearing as many times as it appears in the decomposition \eqref{dec}, i.e., $l_r^{\alpha\beta}$ times.
Now, transform $T$ by a change of basis by $V$
\begin{equation}\label{bas}
    V T V^{-1} = T_D \,,
\end{equation}
which implies
\begin{equation}\label{dein}
    M_DT_DM_D^{-1}=T_D\;.
\end{equation}
$T_D$ can be found by focusing on the sector of a representation $r$ that appears in the decomposition, where a sector is the block consisting of $l_r^{\alpha \beta}$ number of block diagonal $M_r$ matrices. Evaluating
\begin{equation}
    M_D^{-1}=\begin{pmatrix}
M_{r_1}^{-1}&0&\cdots&0\\
0&M_{r_2}^{-1}&\cdots&0\\
\vdots&\vdots&\ddots&\vdots\\
0&0&\cdots&M_{r_n}^{-1}\\
\end{pmatrix}\,,
\end{equation}
and using properties of block matrix multiplication, each section can act independently. 
As an illustration, suppose that $r$ appears twice, i.e.,  $l_r^{\alpha\beta}=2$,
\begin{equation}
    M_D=\begin{pmatrix}
\ddots&0&0&0\\
0&M_{r}&0&0\\
0&0&M_r&0\\
0&0&0&\ddots\\
\end{pmatrix}\;.
\end{equation}
Then, denoting $\mathbb{1}_r$ as the identity in the representation $r$, it is easy to see that 
\begin{equation}
    T_D
    =\begin{pmatrix}
\ddots&0&0&0\\
0&\gamma_{11}\mathbb{1}_r&\gamma_{12}\mathbb{1}_r&0\\
0&\gamma_{21}\mathbb{1}_r&\gamma_{22}\mathbb{1}_r&0\\
0&0&0&\ddots\\
\end{pmatrix}
\label{eq:TD2}
\end{equation}
satisfies \eqref{dein} for any complex numbers $\gamma_{11}$, $\gamma_{12}$, $\gamma_{21}$, and $\gamma_{22}$.
In general, \eqref{eq:TD2} can be written as 
\begin{equation}\label{td}
   T_D
    =\gamma_{mn}
    ^r T_r^{mn}\,,
\end{equation}
where 
for any representation $r$, that appears $l_r^{\alpha\beta}$ times, we can construct $(l_r^{\alpha\beta})^2$ independent matrices
$T_r^{mn}$ by putting an identity matrix in an element of the $l_r^{\alpha\beta}\times l_r^{\alpha\beta}$ matrix block corresponding to the representation $r$. If we consider all representations, we find as many independent invariants as \eqref{dim}, thus we have a complete basis for the invariant space. 
Using $T_r^{mn}T_s^{pq}=\delta_{rst}\delta^{np}T^{mq}_t$ and ${T_r^{mn}}^\dag=T_r^{nm}$ (where $\delta^{r s t} = 1$ for $r = s = t$ and zero otherwise), the graham matrix elements are given by
\begin{equation}\label{gra}
    \langle T_s^{pq}|T_r^{mn} \rangle = \tr\; {T_s^{pq}}^\dag T_r^{mn} = \tr\; {T_s^{qp}} T_r^{mn} = \delta_{srt}\delta^{pm} \tr\; T^{qn}_t = \delta_{srt}\delta^{pm} \delta^{qn} d_t \,,
\end{equation}
where repeated indices are summed over.
Consequently, $T_r^{mn}$ is an orthogonal basis and 
\begin{equation}
    \langle T_r^{mn}|T_r^{mn} \rangle = d_r\,,
\end{equation}
where $d_r$ is the dimension of the representation $r$. 
Recalling \eqref{eq:Weingartenmatrix}, we can readily invert \eqref{gra} to obtain
\begin{equation}
    w_{rs}^{mnpq}=\delta_{srt}\delta^{pm} \delta^{qn} d_t^{-1}\;. 
\end{equation}
We also need to calculate the $A$ matrix defined in \eqref{am}, but from \eqref{bas} we see that 
\begin{equation}
    |T_\mu\rangle = V^{-1}|T_r^{mn}\rangle\,.
\end{equation}
Therefore,  the $A$ matrix is given by
\begin{equation}
    A_{i\mu}=\langle E_{i} |V^{-1}|T_r^{mn}\rangle = \tr\; E_{i}^\dag V^{-1}T_r^{mn}V= \tr\left[\left(VE_{i}V^{-1}\right)^\dag T_r^{mn}\right]\,,
    \label{eq:Aimu}
\end{equation}
where in the last expression, we insert
\begin{equation}
    VE_{i}V^{-1} = V|e_a\rangle |e_b\rangle \langle e_{\bar{a}}|\langle e_{\bar{b}}|V^{-1}=V|e_a\rangle |e_b\rangle  (V| e_{\bar{a}}\rangle| e_{\bar{b}}\rangle)^\dag\;.
\end{equation}
Notice that we can also write, now being more explicit on the representation,
\begin{equation}
    V|e_a^\alpha\rangle |e^\beta_b\rangle = c^{rk}_{abp}|e^{rk}_p\rangle\,,
\end{equation}
where $c^{rk}_{abp}$ are the Clebsch-Gordan coefficients, and the representations $\alpha$ and $\beta$ are written explicitly on the left but not on the right. We also introduce a new index $k$, an integer which satisfies $1\leq k \leq l_r^{\alpha \beta}$ to distinguish the same representation $r$ which appears $l_r^{\alpha \beta}$ times in the decomposition of $\alpha\otimes\beta$.
Then, we can compute \eqref{eq:Aimu}
\begin{equation}
A_{i \mu} = 
    \tr\;{c^{sk*}_{abp}}c^{\bar{s}\bar{k}}_{\bar{a}\bar{b}\m p}|e^{\m s\m k}_{\m p}\rangle\langle e^{sk}_p|T_r^{mn}={c^{sk*}_{abp}}c^{\bar{s}\bar{k}}_{\bar{a}\bar{b}\m p}\langle e^{sk}_p|T_r^{mn}|e^{\m s\m k}_{\m p}\rangle
\end{equation}
and using the coefficients for $T_r^{mn}$ in this basis,  $\langle e^{sk}_p|T_r^{mn}|e^{\m s\m k}_{\m p}\rangle=\delta^{s\m s r}\delta_{p \m p} \delta^{km}\delta ^{\m k n}$, we further compute
\begin{equation}
    A_{i\mu}=c^{sk*}_{abp}c^{\bar{s}\bar{k}}_{\bar{a}\bar{b}\m p}\delta^{r s\m s}\delta_{p \m p} \delta^{km}\delta ^{\m k n}=\delta^{rs\m s }{c^{sm*}_{abp}}c^{\m s n}_{\m a\m b p}={c^{rm*}_{abp}}c^{rn}_{\m a\m b p}\,,
    \label{eq:AimuCG}
\end{equation}
where in the last expression there is a sum only over $p$. 
However, remark that there is no known formula for these coefficients for a general representation, therefore one cannot compute explicitly $A_{i \mu}$  with the expression above \eqref{eq:AimuCG}.
Nevertheless, inserting \eqref{eq:AimuCG} in the expression \eqref{teo}, we get
\begin{equation}
    I_{ij}
    =
    \sum_{\mu,\nu}A_{i\mu} w_{\mu\nu} A^\dag_{\nu j}
    =
    \sum_{r,m,n,t,u,v}\delta^{rs\m s }{c^{sm*}_{abp}}c^{\m s n}_{\m a\m b p}\delta^{t k\m k }{c^{ku}_{bdq}}c^{\m k v*}_{\m b\m d q}\delta_{rtg}\delta^{um} \delta^{vn} d_g^{-1}\,,
\end{equation}
Performing some of the sums, we find that $I_{ij}$ in \eqref{eq:Iij} can be expressed in terms of the  (unknown) Clebsch-Gordan coefficients as
\begin{equation}\label{iijr}
    I_{ij}=\sum_{r,m,n,p,q}{c^{rm*}_{abp}}c^{r n}_{\m a\m b p}{c^{rm}_{bdq}}c^{r n*}_{\m b\m d q}d_r^{-1}\,.
\end{equation}
\end{proof}
We comment that 
the result in \eqref{iijr} is expressed in terms of the Clebsch-Gordan coefficients which are base dependent. The possibility of choosing different bases shows that there is a degree of freedom in the summation. This freedom is manifest in the summation over the indices $p$ and $q$. 
Further simplification could be achieved by expressing \eqref{iijr} in a base invariant way.

\section{
Evaluation of $\langle \chi_r(A) \rangle_0$ and $\langle \chi_r(A^2) \rangle_0$
}
\label{sec:characters}

In this section, we study and present several computational results on the large $N$ limit of the Hermitian Gaussian matrix model averages $\langle \chi_r (A)\rangle_0$ and $\langle \chi_r (A^2)\rangle_0$ as the latter appeared as a key quantity as in \eqref{in2} and the former can be an useful exercise and preparation for computing the latter. The Gaussian average of a character of $A$ for a given representation $\{h\}$ is given by
\begin{equation}
\langle \chi_{\{h\}}(A) \rangle_0 
    = \frac{1}{Z_0}
    \int dA\; e^{-N\tr \frac{1}{2} A^2} \; \chi_{\{h\}}(A) \,,
    \label{chiA}
\end{equation}
where $Z_0$ is an Gaussian integral of $A$;
\begin{equation}
\label{Z0}
Z_0 = \int dA\; e^{-N\tr \frac{1}{2} A^2}\,.
\end{equation}
Similarly to \eqref{chiA}, the Gaussian average of a character of $A^2$ for a given representation, given by
\begin{equation}
\langle \chi_{\{h\}}(A^2) \rangle_0 
    = \frac{1}{Z_0}
    \int dA\; e^{-N\tr \frac{1}{2} A^2} \; \chi_{\{h\}}(A^2) \,.
    \label{chiA2}
\end{equation}
For $\{h\}^e=\{a\in \{h\}|\,a\;\mathrm{is\;even}\}$ and $\{h\}^o=\{a\in \{h\}|\,a\;\mathrm{is\;odd}\}$, \eqref{chiA} evaluates to \cite{difrancesco1992generating}
\begin{equation}\label{in11}
    \langle \chi_{\{h\}}(A) \rangle_0 = N^{-\frac{n}{2}}
   \; d_h \;
    \frac{\prod_i (h_i^e-1)!!h_i^o!!}{\prod_{i,j}(h^e_i-h^o_j)}
    =N^{-\frac{n}{2}}
     \; d_h \;
    \frac{\chi_{\{h\}}(C_2)}{\chi_{\{h\}}(C_1)}
    \,,
\end{equation}
where $n=\sum_i{h_i}-N(N-1)/2$ is the size of the representation, and $d_h$ is the dimension of the representation $h$, which can be evaluated as
\begin{equation}\label{dh}
    d_h=\chi_{\{h\}}(\mathbb 1)=\frac{\Delta(h)}{\prod_{k=1}^{N-1}k!}\,.
\end{equation} 
We also used  
\begin{equation}\label{c2c1}
    \frac{\prod_i (h_i^e-1)!!h_i^o!!}{\prod_{i,j}(h^e_i-h^o_j)}
    =
    \frac{\chi_{\{h\}}(C_2)}{\chi_{\{h\}}(C_1)}
    \,,
\end{equation}
as presented in more detail in the appendix of \cite{Kazakov_1996}.

\begin{lemma}
\label{lem:lemmaI1I2}

If there exists a $GL(N)$ matrix $M$ such that Di Francesco-Itzykson integral \eqref{chiA} can be computed as the character of $M$; 
\begin{equation}\label{lem1}
    \langle \chi_{\{h\}}(A) \rangle_0  =\chi_{\{h\}}(M)\,,
\end{equation}
then the integral of our interest \eqref{in2} can be solved as:
\begin{equation}
    \langle \chi_{\{h\}}(A^2) \rangle_0 =\chi_{\{h\}}(M^2)\;.
\end{equation}

\end{lemma}
\begin{proof}
The function $\chi_{\{h\}}(A^2)$ is a class function of $A$, therefore there exists an expansion
\begin{equation}\label{expa2}
    \chi_{\{h\}}(A^2)=\sum_{\{h'\}}c_{\{h\},\{h'\}}\chi_{\{h'\}}(A)\;.
\end{equation}
This expansion is finite, since it is the same as the present in the identity $\chi_{\{h\}}(A^2)=\chi_{\mathrm{Sym}^2\{h\}}(A)-\chi_{\mathrm{Alt}^2\{h\}}(A)$. Thus, 
\begin{equation}
\begin{split}
   \langle \chi_{\{h\}}(A^2) \rangle_0 
    &=
    \frac{1}{Z_0}\int dA \;\chi_{\{h\}}(A^2) \; e^{-N\tr \frac{1}{2} A^2} \;\;
    =\frac{1}{Z_0}\int dA\;\sum_{\{h'\}}c_{\{h\},\{h'\}}\chi_{\{h'\}}(A)\;e^{-N\tr \frac{1}{2} A^2}\;\\
    &=\frac{1}{Z_0}\sum_{\{h'\}}c_{\{h\},\{h'\}}\int dA\;\chi_{\{h'\}}(A)\;e^{-N\tr \frac{1}{2} A^2}
    =\frac{1}{Z_0}\sum_{\{h'\}}c_{\{h\},\{h'\}}I^{(1)}_{\{h'\}} (\mathbb{1})
    \\
    &=\sum_{\{h'\}}c_{\{h\},\{h'\}}\chi_{\{h'\}}(M)=\chi_{\{h\}}(M^2)\;.
\end{split}
\end{equation}
We note that the exchange between the summation and the integral is legal, since the summation is finite.
\end{proof}

Lemma \ref{lem:lemmaI1I2} suggests that we shall study the expression in \eqref{in11} in order to see if we can find such a $M$ as in \eqref{lem1}. 
Let us present a more general normalized version of \eqref{chiA};
{\color{blue}\begin{equation}\label{gen}
    {I^{(1)}_{\{h\}} (B) =\ \int dA\;\chi_{\{h\}}(A)e^{-N\tr \frac{1}{2} AB^{-1}AB^{-1}}\;}
\;.
\end{equation}}
The generalized normalized Di Francesco-Itzykson integral 
\eqref{gen} can be expressed as a product of characters 
\begin{equation}\label{gen2}
    \frac{I^{(1)}_{\{h\}} (B)}{Z(B)} 
    =  
    \frac{\chi_{\{h\}}(B)\chi_{\{h\}}(C_2)}{\chi_{\{h\}}(C_1)} \,,
\end{equation}
where we used \eqref{c2c1} and $Z(B) = \int dA\;e^{-N\tr \frac{1}{2} AB^{-1}AB^{-1}}$ is a normalization factor.
Remark that on the right hand side of \eqref{gen2}, the  only part that depends on the choice of $B$ is $\chi_{\{h\}}(B)$.

If we set $B = \mathbb{1}$ in \eqref{gen2} , we indeed recover \eqref{in11}.
By setting $B=C_1$ in \eqref{gen2}, we find that
\begin{equation}
\label{rawc1c2}
I^{(1)}_{\{h\}} (C_1) 
= 
    \int dA\;\chi_{\{h\}}(A)
    \;
    e^{-N\tr \frac{1}{2} AC_1^{-1}AC_1^{-1}}\;
    =
    Z(C_1) \; 
    \chi_{\{h\}}(C_2)\;.
\end{equation}
Since $(C_2)^2=C_1$ in the large $N$ limit, by summing over all representations as in \eqref{expa2} on both sides of \eqref{rawc1c2} we obtain
\begin{equation}
    I^{(1)}_{\{h\}} (C_1)  \;
    \sim \;
    \chi_{\{h\}}(C_1)\,,
\end{equation}
or performing  change of variables,
\begin{equation}
I^{(1)}_{\{h\}} (C_1) 
\; \sim \;
    \int dA\;\chi_{\{h\}}((C_1A)^2)e^{-N\tr \frac{1}{2} A^2}\;
     \sim \; 
    \chi_{\{h\}}(C_1)\,,
\end{equation}
which can also be written as
\begin{equation}
    \langle \chi_{\{h\}}(C_1 A)^2 \rangle_0 
    \sim \; 
    \chi_{\{h\}}(C_1)\,,
\end{equation}
where $\sim$ denotes equal up to a constant in $N$, and $\langle \rangle_0$ denotes the Gaussian average over $A$ as defined in \eqref{avdef}.
This result is nearly the expression we aim at computing in \eqref{in2}, except for the extra $C_1$ in the left hand side.

On this note, let us now investigate further what  the condition in \eqref{in11} tells us about the matrix $M$ in \eqref{lem1}. 
\begin{proposition}
Given a representation $r$ of $\mathrm{GL}(N)$ defined through the set of shifted weights ${h}_i$, $i=1,...N,$, consider $\chi_r$ the character in that representation $r$. If $M$ is a $GL(N)$ matrix that satisfies
\begin{equation}
\int dA\;\chi_{\{h\}}(A)\; e^{-N\tr \frac{1}{2} A^2}
=
\chi_{\{h\}}(M)\,,
\end{equation}
then, for a positive integer $q$,
\begin{equation}
    \frac{1}{N}\tr (M^{2q}) = q^q\;.
\end{equation}
\end{proposition}
\begin{proof}
From \eqref{chard} we can deduce that
\begin{equation}\label{trf}
 \tr (M^p) =\sum_{k=1}^N\frac{\chi_{\{h^{p,k}\}}(M)}{\chi_{\{h\}}(M)}    
\end{equation}
by using a representation $\{h\}$ where $\chi_{\{h\}}(M)\neq 0$ and the modified representations $\{h^{p,k}\}$ where the shifted weights are related to the ones of $\{h\}$ by $h_i^{p,k}=h_i+p\,\delta_{k,i}$, for $i=1,...,N$. 

Using \eqref{dh} in \eqref{in11} we can deduce that, for some constant $c_N$,
\begin{equation}\label{chm}
    \chi_{\{h\}}(M)=c_N\,N^{-\frac{n}{2}} \Delta(\{h\}^e)\Delta(\{h\}^o)\prod_i (h_i^e-1)!!h_i^o!!\,.
\end{equation}
We assume $N$ and $p$ even. For $i$ even, we set $h_i=p (N-i)/2+1$, therefore $h_i$ is always odd for $i even$. For $i$ odd, we set $h_i=p {(N-i-1)}/{2}+2$, therefore $h_i$ is always even for $i$ odd. Consequently, for $i$ even, $h_i^{p,k}=p {(N-i)}/{2}+1+p\,\delta_{k,i}$, making $h_i^{p,k}$ odd and, for $i$ odd, $h_i^{p,k}=p (N-i-1)/2+2+p\,\delta_{k,i}$, making $h_i^{p,k}$ even. Notice that $\{h\}^e=\{a+1\,|\,a\in\{h\}^o\}$, thus
the Vandermonde determinants $\Delta(\{h\}^e)$ and $\Delta(\{h\}^o)$ are
\begin{equation}
    \Delta(\{h\}^o)= \Delta(\{h\}^e)=\prod_{i<j}({h}_i-{h}_j)=\prod_{i<j}\left(p \frac{N-i}{2}+1-p \frac{N-j}{2}-1\right)=\prod_{i<j}p \frac{j-i}{2},
\end{equation}
where the indices $i$ and $j$ are always even. Because of this, by setting $i=2m$ and $j=2n$ we find
\begin{equation}
   \Delta(\{h\}^o)= \Delta(\{h\}^e) =p^{\frac{1}{2}\frac{N}{2}\left(\frac{N}{2}-1\right)}\prod_{m<n}(n-m)=p^{\frac{1}{2}\frac{N}{2}\left(\frac{N}{2}-1\right)}\prod_{m=1}^{\frac{N}{2}-1{}}m!\;.
\end{equation}
We then get that
\begin{equation}\label{normde}
    \Delta(\{h\}^o)\Delta(\{h\}^e)=p^{\frac{N}{2}\left(\frac{N}{2}-1\right)}\prod_{m=1}^{\frac{N}{2}-1}(m!)^2\;,
\end{equation}
which will be necessary for the normalization in \eqref{trf}. The terms involving double factorials in \eqref{chm} are:
\begin{equation}
    \prod_i (h_i^e-1)!!=\prod_i h_i^o!!= \prod_{m=1}^{\frac{N}{2}}\left(p\left(\frac{N}{2}-m\right)+1\right)!!\;.
\end{equation}
The size of the representation is
\begin{equation}
    n=\sum_ih_i-\frac{N(N-1)}{2}=\frac{1}{2}\frac{N}{2}\left(p(N-2)+6\right)-\frac{N(N-1)}{2}\;.
\end{equation}
This way the character of $M$ in the representation $\{h\}$ is 
\begin{equation}\label{chp0}
    \chi_{\{h\}}(M)=c_N\;N^{-\frac{n}{2}}\;p^{\frac{N}{2}\left(\frac{N}{2}-1\right)}\prod_{m=1}^{\frac{N}{2}-1}(m!)^2\prod_{m=1}^{\frac{N}{2}}\left(p\left(\frac{N}{2}-m\right)+1\right)!!^2\,.
\end{equation}
Now, for the representations $\{{h^{p,k}}\}$, set $\{{h^{p,k}}\}^e=\{a\in\{h^{p,k}\}|\,a\;\mathrm{is\;even}\}$ and $\{{h^{p,k}}\}^o=\{a\in\{h^{p,k}\}|\,a\;\mathrm{is\;odd}\}$. For $k$ odd, $\{{h^{p,k}}\}^e=\{{h}\}^e$, then
\begin{equation}\label{eko}
    \Delta(\{{h^{p,k}}\}^e)=\Delta(\{{h}\}^e)=p^{\frac{1}{2}\frac{N}{2}\left(\frac{N}{2}-1\right)}\prod_{m=1}^{\frac{N}{2}-1{}}m!\;.
\end{equation}
For $k$ even, $\{{h^{p,k}}\}^o=\{{h}\}^o$, then
\begin{equation}\label{oke}
    \Delta(\{{h^{p,k}}\}^o)=\Delta(\{{h}\}^o)=p^{\frac{1}{2}\frac{N}{2}\left(\frac{N}{2}-1\right)}\prod_{m=1}^{\frac{N}{2}-1{}}m!\;.
\end{equation}
Another consequence for $k$ even is
\begin{equation}
    \Delta(\{{h^{p,k}}\}^e)
    =\prod_{\substack{i<j\\i,j\;\mathrm{even}}}({h^{p,k}}_i-{h^{p,k}}_j)
    =\prod_{\substack{i<j\\i,j\;\mathrm{even}}}(p \frac{N-i}{2}+2+p\,\delta_{k,i}-p \frac{N-j}{2}-2-p\,\delta_{k,j})
    =\prod_{\substack{i<j\\i,j\;\mathrm{even}}}\left(p \frac{j-i}{2}+p(\delta_{k,i}-\delta_{k,j})\right)\,.
\end{equation}
Setting $i=2m$, $j=2n$, and $k=2r$ we find
\begin{equation}
    \Delta(\{{h^{p,k}}\}^e)=p^{\frac{1}{2}\frac{N}{2}\left(\frac{N}{2}-1\right)}\prod_{m<n}\left( n-m+(\delta_{r,m}-\delta_{r,n})\right)\;.
\end{equation}
Let us emphasize that $n>1$ since $n>m\geq 1$. 
For $k=2$,
\begin{equation}\label{ek2}
    \Delta(\{{h^{p,2}}\}^e)=p^{\frac{1}{2}\frac{N}{2}\left(\frac{N}{2}-1\right)}\prod_{m<n}\left( n-m+\delta_{1,m}\right)=p^{\frac{1}{2}\frac{N}{2}\left(\frac{N}{2}-1\right)}\frac{N}{2}\prod_{i=1}^{\frac{N}{2}-1}i!\;.
\end{equation}
For $k> 2$, $k$ even, since the factor for $m=r-1$ and $n=r$  in the product  is zero,
\begin{equation}\label{ekg}
    \Delta(\{{h^{p,k}}\}^e)=0\;.
\end{equation}
It is also true that $\{{h^{p,k}}\}^o=\{a-1|\;a\in \{{h^{p,k+1}}\}^e\}$, then $\Delta(\{{h^{p,1}}\}^o)=\Delta(\{{h^{p,2}}\}^e)$, therefore
\begin{equation}\label{ok1}
    \Delta(\{{h^{p,1}}\}^o)=p^{\frac{1}{2}\frac{N}{2}\left(\frac{N}{2}-1\right)}\frac{N}{2}\prod_{i=1}^{\frac{N}{2}-1}i!\,.
\end{equation}
For $k$ odd with $k>2$, since $\Delta(\{{h^{p,k}}\}^o)=\Delta(\{{h^{p,k+1}}\}^e)$,
\begin{equation}\label{okg}
    \Delta(\{{h^{p,k}}\}^o)=0\;.
\end{equation}
Using \eqref{eko}, \eqref{oke}, \eqref{ek2}, \eqref{ekg}, \eqref{ok1} and \eqref{okg} we get 
\begin{equation}
    \Delta(\{{h^{p,k}}\}^e)\Delta(\{{h^{p,k}}\}^o)=\begin{cases}
p^{\frac{N}{2}\left(\frac{N}{2}-1\right)}\frac{N}{2}\prod_{i=m}^{\frac{N}{2}-1} (m!)^2 & \text{ if } k=1,2 \\
0 & \text{ if } k>2
\end{cases}\;.
\label{eq:vdmds}
\end{equation}
The consequence of \eqref{eq:vdmds} in the sum in \eqref{trf} is that only the terms for $k=1$ and $k=2$ are nonzero. 
For $k=1$ or $k=2$, the double factorial terms in \eqref{chm} is
\begin{equation}
    \prod_i (h_i^e-1)!! h_i^o!! =\frac{(p \frac{N}{2} +1)!!}{\left(p\left(\frac{N}{2}-1\right)+1\right)!!}\prod_{m=1}^{\frac{N}{2}} \left ( p \left(\frac{N}{2}-m\right)+1\right)!!^2\;.
\end{equation}
The size of the representation is
\begin{equation}
    n^{p,k}=\sum_ih^{p,k}_i-\frac{N(N-1)}{2}=p+\frac{1}{2}\frac{N}{2}\left(p(N-2)+6\right)-\frac{N(N-1)}{2}\;.
\end{equation}
This way the character of $M$ in the representation $\{{h^{p,k}}\}$ is, for $k=1$ or $k=2$,
\begin{equation}\label{chpk}
    \chi_{\{{h^{p,k}}\}}(M)=c_N\,N^{-\frac{n^{p,k}}{2}}\;p^{\frac{N}{2}\left(\frac{N}{2}-1\right)}\frac{N}{2}\prod_{m=1}^{\frac{N}{2}-1}(m!)^2\frac{(p \frac{N}{2} +1)!!}{\left(p\left(\frac{N}{2}-1\right)+1\right)!!}\prod_{m=1}^{\frac{N}{2}}\left(p\left(\frac{N}{2}-m\right)+1\right)!!^2\,,
\end{equation}
and zero for $k\geq 3$.
Hence, inserting \eqref{chp0} and \eqref{chpk} in \eqref{trf}, we find that
\begin{equation}
    \tr (M^p) 
    =N^{1-\frac{p}{2}}\frac{\left(p\frac{N}{2}+1\right)!!}{\left(p\left(\frac{N}{2}-1\right)+1\right)!!}\;.
\end{equation}
We evaluate a large $N$ limit by keeping only the largest order in $N$,
\begin{equation}
    \tr (M^p) = N^{1-\frac{p}{2}}\prod_{j=1}^{p/2}\left(p\left(\frac{N}{2}-1\right)+1+2j\right)= 
    N\left(\frac{p}{2}\right)^{\frac{p}{2}}\left(1+{\mathcal O}(N^{-1})\right)\;.
\end{equation}
Setting the integer $q=p/2$, we get the trace property
\begin{equation}
    \frac{1}{N}\tr (M^{2q}) = q^q\;.
\end{equation}
\end{proof}
In principle, this information should be enough to find $M$ up to matrix conjugation, similarly to what is done in \eqref{cmp} for the matrix $C_m$. We leave this possibility of computation of $M$ to future studies.

Let us now evaluate $\langle \chi_r(A^2)\rangle_0$ for a finite $N$. In particular, in Theorem \ref{thm:chiA2finiteN}, the integral over matrices has been evaluated and replaced with a summation over integers.

\begin{theorem}
\label{thm:chiA2finiteN}
Let $A$ be a random variable for a $N \times N$ Hermitian matrix under the Gaussian measure. Given a representation $r$ of $\mathrm{GL}(N)$ defined through the set of shifted weights $ h_i $, $i=1,...N,$ and considering $\chi_r$ the character in that representation $r$ the following holds true:
\begin{equation}\label{resfinn}
    \langle \chi_r(A^2)\rangle_0 =\frac{N^{\frac{N(N-1)}{2}}}{\prod_{k=0}^{N-1}k!}\frac{\prod_{i}(2h_i)!}{(2N)^{\sum_i h_i}}\underset{i,j}{\mathrm{Pf}}\sum_{\substack{k+l=2h_i\\u+v=2h_j}}\frac{(-1)^u-(-1)^{k}}{2}\frac{(k+u)!!(l+v-2)!!}{k!u!l!v!}\,.
\end{equation}
\end{theorem}

\begin{proof}
For a Hermitian matrix $A$ of size $N$ whose eigenvalues are denoted by $x$, we write $X$ as a diagonal matrix whose diagonal elements are the eigenvalues $x$'s. For a given representation $r$, we are interested in:
\bea
\langle \chi_r (A^2) \rangle_0
&=&
\frac{1}{Z_0}\int dA \, \chi_r(A^2) \; e^{-\frac{N}{2} {\rm Tr} A^2}
\crcr
&=&
\frac {{\tilde c}_N}{N!}
\int_{\mathbb{R}^N} d X\prod_{i <j} (x_i - x_j)^2
\frac{\underset{i,j}{\rm det} \, \left(x_i^{2 h_j}\right)}{\prod (x_i^2 - x_j^2)} 
e^{-\frac{N}{2} {\rm Tr} X^2}
\crcr
&=&
\frac {{\tilde c}_N}{N!}\int_{\mathbb{R}^N}  d X \,\underset{i,j}{\rm det} \, \Big(x_i^{2 h_j}\Big)
\prod_{i <j}\frac{ x_i - x_j}{x_i + x_j}
e^{-\frac{N}{2} {\rm Tr} X^2}
\, ,
\label{eq:avechiA2whatwewant}
\eea
where
\begin{equation}\label{ctil}
{\tilde c}_N=\frac{1}{Z_0}\frac {{\rm vol} (U(N))}{(2 \pi)^N}=N^{\frac{N^2}{2}}(2\pi)^{-\frac{N}{2}}(\prod_{k=0}^{N-1}k!)^{-1}.
\end{equation}
See \cite{Eynard:2015aea},
\cite{Anninos:2020ccj},
and \cite{zhang2017volumes} for the change of variables from $A$ to $X$.
We use de Bruijn's formula \cite{deBruijn}, 
\begin{equation}\label{bru}
   \frac{1}{N!} \int_{\mathbb{R}^N} d\mu (X) 
   \; \underset{i,j}{\mathrm{det}}f_i(x_j)
   \; \underset{i,j}{\mathrm{Pf}}A(x_i,x_j)
   =
   \underset{i,j}{\mathrm{Pf}}\int_{\mathbb{R}^2} d\mu(x)d\mu(y)\; f_i(x)\; A(x,y)\; f_j(y)\,,
\end{equation}
so we can reduce the integration over the N variables on the left hand side to a Pfaffian of an integral over two variables on the right hand side. $d\mu(x)$ sets the measure on $x$. Here, it is a Gaussian measure $d\mu(x)=e^{-\frac{N}{2}x^2}$. 
Since \cite{macdonald1998symmetric}
\begin{equation}\label{pffpr}
    \prod_{i<j}\frac{ x_i - x_j}{x_i + x_j}
    =
    \underset{i,j}{\mathrm{Pf}} \Bigg(\frac{ x_i - x_j}{x_i + x_j} \Bigg) \,,
\end{equation}
let us first rewrite \eqref{eq:avechiA2whatwewant}, 
\begin{equation}\label{wpf}
    \langle \chi_r(A^2)\rangle_0 
    =
\frac{{\tilde c}_N}{N!}\int_{\mathbb{R}^N}  d X \,\underset{i,j}{\rm det} \big(x_i^{2 h_j}\big)
\;
\underset{i,j}{\mathrm{Pf}} \Bigg(\frac{ x_i - x_j}{x_i + x_j} \Bigg)
\;
e^{-\frac{N}{2} {\rm Tr} X^2}\;.
\end{equation}
An important thing to notice is that when $x_i+x_j=0$, the Pfaffian has a divergence that is controlled by the zero in the determinant.
However, when we use the de Bruijn's formula the determinant is removed. Therefore, it is best to deal with this divergence already here.
We introduce a damping with a small constant $\epsilon$ to prevent a divergence from appearing. We take care of this by multiplying the elements in the Pfaffian by the term $\frac{(x_i+x_j)^2}{(x_i+x_j)^2+\epsilon^2}$. 
We regularize \eqref{wpf} by defining, for $\epsilon >0$,
\begin{equation}\label{chep}
    \langle \chi_r(A^2)\rangle_\epsilon 
    =
\frac{{\tilde c}_N}{N!}\int_{\mathbb{R}^N}  d X \,\underset{i,j}{\rm det} \big(x_i^{2 h_j}\big)
\;
\underset{i,j}{\mathrm{Pf}}\Bigg(\frac{ x_i^2 - x_j^2}{(x_i + x_j)^2+\epsilon^2}\Bigg) 
\;
e^{-\frac{N}{2} {\rm Tr} X^2}\,.
\end{equation}
We expand the Pfaffian in \eqref{chep}, a polynomial, into its monomials. We do the same for the Pfaffian in \eqref{wpf}. By comparing the absolute values of these integrands, term by term, we see that the ones from \eqref{chep} are bounded by the ones from \eqref{wpf}. Therefore, the dominated convergence theorem \cite{schilling2017measures} tells us that
\begin{equation}\label{lim}
    \langle \chi_r(A^2)\rangle_0=\lim_{\epsilon\rightarrow 0}\langle \chi_r(A^2)\rangle_\epsilon\;.
\end{equation}
Now we can use the de Bruijn's formula in \eqref{bru} and we obtain the damped average
\begin{equation}\label{2int}
    \langle \chi_r(A^2)\rangle_\epsilon 
    = 
    {\tilde c}_N 
    \; \underset{i,j}{\mathrm{Pf}}\int_{\mathbb{R}^2} dxdy\; \frac{x^2-y^2}{(x+y)^2+\epsilon^2}x^{2h_i}y^{2h_j} 
    \; e^{-\frac{N}{2}(x^2+y^2)}\;.
\end{equation}
We remark that at this point, if we take the $\epsilon \rightarrow 0$ limit, the integral in \eqref{2int} turns into a principal value integral. Let us define
\begin{equation}
    T_{ij}=\int_{\mathbb{R}^2} dxdy\; \frac{x^2-y^2}{(x+y)^2+\epsilon^2}x^{2h_i}y^{2h_j} e^{-\frac{N}{2}(x^2+y^2)}\;.
\end{equation}
Then, according to \eqref{2int},
\begin{equation}\label{Tf}
     \langle \chi_r(A^2)\rangle_\epsilon = {\tilde c}_N \underset{i,j}{\mathrm{Pf}}\;T_{ij}\;.
\end{equation}
We can simplify the $N$ dependence by changing variables $x,y\rightarrow N^{-\frac{1}{2}}x,N^{-\frac{1}{2}}y$,
\begin{equation}\label{Tij}
    T_{ij}=\frac{1}{N^{h_i+h_j+1}}\int_{\mathbb{R}^2} dxdy\; \frac{x^2-y^2}{(x+y)^2+N\epsilon^2}x^{2h_i}y^{2h_j} e^{-\frac{1}{2}(x^2+y^2)}\;.
\end{equation}
Introducing sources $\alpha$ and $\beta$ through the terms $\alpha x$ and $\beta y$ in the exponential, we can turn the factors $x^{2h_i}$ and $y^{2h_j}$ into derivatives:
\begin{equation}
    T_{ij}=\frac{1}{N^{h_i+h_j+1}}\left. \frac{d^{2h_i}}{d\alpha^{2h_i}}\frac{d^{2h_j}}{d\beta^{2h_j}}\int_{\mathbb{R}^2} dxdy\; \frac{x^2-y^2}{(x+y)^2+N\epsilon^2} e^{-\frac{1}{2}(x^2+y^2)+\alpha x+\beta y}\right|_{\alpha,\beta=0}\;.
\end{equation}
Let us also define
\begin{equation}
    {\rm II}(\alpha,\beta)=\int_{\mathbb{R}^2} dxdy\; \frac{x^2-y^2}{(x+y)^2+N\epsilon^2} e^{-\frac{1}{2}(x^2+y^2)+\alpha x+\beta y}\,,
\end{equation}
thus
\begin{equation}\label{If}
    T_{ij}
    =
    \frac{1}{N^{h_i+h_j+1}}\left. \frac{d^{2h_i}}{d\alpha^{2h_i}}\frac{d^{2h_j}}{d\beta^{2h_j}} {\rm II}(\alpha,\beta)\right|_{\alpha,\beta=0}\;.
\end{equation}
By changing integration variables to $u=(x+y)/\sqrt{2}$ and $v=(x-y)/\sqrt{2}$, and also defining
$a=(\alpha+\beta)/\sqrt{2}$ and $b=(\alpha-\beta)/\sqrt{2}$, it becomes
\begin{equation}
    {\rm II}(\alpha,\beta)=\int_{\mathbb{R}^2} dudv\; \frac{uv}{v^2+2N\epsilon^2} e^{-\frac{1}{2}(u^2+v^2)+au+bv}\;.
\end{equation}
Here we notice that the integration over $u$ and over $v$ are independent, hence we can separate them, 
\begin{equation}\label{iab}
    {\rm II}(\alpha,\beta)=\int_{\mathbb{R}} du\;u\;e^{-\frac{1}{2}u^2+au} \int_{\mathbb{R}} dv\; \frac{v}{v^2+2N\epsilon^2} e^{-\frac{1}{2}v^2+bv} 
    =
    {\rm II}^{(1)}(a)
    {\rm II}^{(2)}_{N\epsilon}(b)\,,
\end{equation}
where we define the integrals
\begin{equation}
    {\rm II}^{(1)}(a)
    =
    \int_{\mathbb{R}} du\;u\;e^{-\frac{1}{2}u^2+au}\qquad \mathrm{and} \qquad {\rm II}^{(2)}_\epsilon(b)=\int_{\mathbb{R}} dv\; \frac{v}{v^2+2N\epsilon^2} e^{-\frac{1}{2}v^2+bv}\;.
\end{equation}
The first integral is easily evaluated as
\begin{equation}\label{ia}
    {\rm II}^{(1)}(a)=\sqrt{2\pi}\,a\;e^{\frac{1}{2}a^2}\,.
\end{equation}
The second integral we solve by introducing another integral,
\begin{equation}\label{blim}
   {\rm II}^{(2)}_\epsilon(b)=\int_0^b d\tilde b\int_{\mathbb{R}} dv\; \frac{v^2}{v^2+2N\epsilon^2} e^{-\frac{1}{2}v^2+{\tilde b}v}\;.
\end{equation}
At this point we can evaluate the $\epsilon \rightarrow 0$ limit due to the dominated convergence theorem. Thus, by defining
\begin{equation}
    {\rm II}^{(2)}_0(b)=\lim_{\epsilon\rightarrow 0} {\rm II}^{(2)}_\epsilon(b)\,,
\end{equation}
we find that
\begin{equation}\label{aflim}
    {\rm II}^{(2)}_0(b)=\int_0^b d\tilde b\int_{\mathbb{R}} dv\; e^{-\frac{1}{2}v^2+{\tilde b}v}\;.
\end{equation}
The integral over $v$ in \eqref{aflim} is a simple Gaussian, and we can evaluate it to find
\begin{equation}\label{ib}
    {\rm II}^{(2)}_0(b)  =\sqrt{2\pi} \int_0^bd{\tilde b}\; e^{\frac{1}{2}{\tilde b}^2}\;.
\end{equation}
This function is, up to normalization conventions, the imaginary error function. Joining \eqref{ia} and \eqref{ib} in \eqref{iab}, we deduce that
\begin{equation}
    {\rm II}(\alpha,\beta)= 2\pi\,a\;e^{\frac{1}{2}a^2}\int_0^bd{\tilde b}\; e^{\frac{1}{2}{\tilde b}^2}\;.
\end{equation}
Going back through \eqref{If} and \eqref{Tf}, we obtain
\begin{equation}\label{alf}
    \langle \chi_r(A^2)\rangle_0 = \left.{\tilde c}_N\underset{i,j}{\mathrm{Pf}}\frac{2\pi}{N^{h_i+h_j+1}}\frac{\partial^{2h_i}}{\partial \alpha ^{2h_i}}\frac{\partial^{2h_j}}{\partial \beta ^{2h_j}}\,a \, e^{\frac{1}{2}a^2}\int_0^b d\tilde{b}\; e^{\frac{1}{2}{\tilde b}^2}\right|_{\alpha,\beta=0}\,,
\end{equation}
and by using that
\begin{equation}
    a \, e^{\frac{1}{2}a^2}=\sum_{k=0}^\infty \frac{a^{2k+1}}{2^k k!}\qquad\mathrm{and}\qquad \int_0^b d\tilde{b}\; e^{\frac{1}{2}{\tilde b}^2}=\sum_{k=0}^\infty \frac{b^{2k+1}}{(2k+1) 2^k k!}
\end{equation}
and that $a=(\alpha+\beta)/\sqrt{2}$ and $b=(\alpha-\beta)/\sqrt{2}$, we can evaluate the derivatives in \eqref{alf} to find that

\begin{equation}\label{chiA20final}
    \langle \chi_r(A^2)\rangle_0 ={\tilde c}_N\underset{i,j}{\mathrm{Pf}}\frac{2\pi}{N^{h_i+h_j+1}}\frac{(2h_i)!(2h_j)!}{2^{h_i+h_j}}\sum_{\substack{k+l=2h_i\\u+v=2h_j\\k+u\mathrm{\,is\,odd}}}(-1)^u\frac{(k+u)!!(l+v-2)!!}{k!u!l!v!}\,,
\end{equation}
which becomes \eqref{resfinn} by using \eqref{ctil}, Pfaffian properties and that $(-1)^k=-(-1)^u$ for $k+u$ odd and $(-1)^k=(-1)^u$ for $k+u$ even.
\end{proof}

We wish to remark here about the quantity 
$ \langle \chi_r(A^2)\rangle_0 $ which we computed by taking the limit of $\epsilon \rightarrow 0$ in the expression $ \langle \chi_r(A^2)\rangle_\epsilon $ given in \eqref{Tf}.
For a given (in other words, finite) $N$, the expression we obtained in \eqref{chiA20final} is valid and therefore Theorem \ref{thm:chiA2finiteN}.
However, if we pay attention the expression \eqref{Tij}, we notice that $N$ comes with $\epsilon$. 
Then, one notices that once we send $N$ to infinity, this procedure becomes sensitive to the ratio in which $N \rightarrow \infty$ and $\epsilon \rightarrow 0$ are sent.

\medskip

One naturally wonders if the expression  obtained in Theorem \ref{thm:chiA2finiteN} may become simpler in  large $N$ limit.
Let us explore this possibility in Proposition \ref{thm:epilonN}.

\begin{proposition}
\label{thm:epilonN}
Let $A$ be a random variable for a $N \times N$ Hermitian matrix under the Gaussian measure. Given a representation $r$ of $\mathrm{GL}(N)$ defined through the set of normalized shifted weights ${\tilde h}_i = h_i/N$, $i=1,...N,$. Consider $\chi_r$ the character in that representation $r$. Defining
\begin{equation}
    \langle \chi_r(A^2)\rangle_\epsilon =\frac{1}{Z_0}
\frac {{\rm vol} (U(N))}{N! \, (2 \pi)^N}
\int  d X \,\underset{i,j}{\rm det} \, x_i^{2 N \tilde{h}_j}
\underset{i,j}{\mathrm{Pf}}\frac{ x_i^2 - x_j^2}{(x_i + x_j)^2+\epsilon^2}
e^{-\frac{N}{2} {\rm Tr} X^2}\,,
\end{equation}
which satisfies $ \langle \chi_r(A^2)\rangle_0=\lim_{\epsilon\rightarrow 0}\langle \chi_r(A^2)\rangle_\epsilon$, the following holds true:
\begin{equation}\label{lnap}
    \lim_{N\rightarrow \infty}
    \frac{\langle \chi_r(A^2)\rangle_\epsilon}
    {{\tilde c}_N\prod_k 2e^{-N {\tilde h}_k}(2 {\tilde h}_k)^{N {\tilde h}_k}
    \underset{i,j}{
    \rm Pf}\left[\frac{( {\tilde h}_i- {\tilde h}_j)( {\tilde h}_i+ {\tilde h}_j+\epsilon^2/2)}{( {\tilde h}_i- {\tilde h}_j)^2+\epsilon^2( {\tilde h}_i+ {\tilde h}_j)+\epsilon^4/4}\right]} = 1\;,
\end{equation}
where ${\tilde c}_N=N^{\frac{N^2}{2}}(2\pi)^{-\frac{N}{2}}(\prod_{k=0}^{N-1}k!)^{-1}$.
\end{proposition}

\begin{proof}
We apply the saddle point method to compute \eqref{2int}. Let us first rescale  integers $h_i$ to ${\tilde h}_i = h_i/N$;
\begin{equation}
    \langle \chi_r(A^2)\rangle_\epsilon 
    = 
    {\tilde c}_N\; \underset{i,j}{\mathrm{Pf}}\int_{\mathbb{R}^2} dxdy\; \frac{x^2-y^2}{(x+y)^2+\epsilon^2}x^{2N {\tilde h}_i}y^{2N {\tilde h}_j} 
    \; e^{-\frac{N}{2}(x^2+y^2)}
\end{equation}
and prepare in a form proper to use the saddle point approximation,
\begin{equation}
    \langle \chi_r(A^2)\rangle_\epsilon 
    = 
    {\tilde c}_N
    \;
    {\rm Pf}
    \int_{\mathbb{R}^2} dxdy\; 
    \frac{x^2-y^2}{(x+y)^2+\epsilon^2}
    \;
    e^{-N(\frac{1}{2}x^2+\frac{1}{2}y^2-2 {\tilde h}_i{\rm ln}|x|-2 {\tilde h}_j{\rm ln}|y|)}\,.
    \label{eq:chisaddle}
\end{equation}
Laplace's method of integration \cite{bleistein1986asymptotic} can be expressed as
\begin{equation}\label{sadd}
    \int_{\mathbb{R}^d} dX \, g(X) \, e^{-N f(X)}
    =
    \sum_{X_0}
    \left(\frac{2\pi}{N}\right)^{d/2}
    \frac{g(X_0)e^{-Nf(X_0)}}{\sqrt{{\rm det}\big(H(f)(X_0)\big)}}
    (1+{\mathcal O}(N^{-1}))
    \,,
\end{equation}
where $X$ is a set of $d$ real variables, $f$ is a twice-differentiable complex valued function of $X$, $H(f)$ is the Hessian matrix of $f$, the points $X_0$ are local maxima of $f$, and  $g$ is a complex valued function of $X$ nonzero at $X_0$. Comparing \eqref{eq:chisaddle} and \eqref{sadd} we identify
\begin{equation}
    f(X)
    =
    \frac{1}{2}x^2+\frac{1}{2}y^2-2 {\tilde h}_i{\rm ln}|x|-2 {\tilde h}_j{\rm ln}|y|
    \quad 
    \mathrm{and}
    \quad 
    g(X)
    = 
    \frac{x^2-y^2}{(x+y)^2+\epsilon^2}\;.
\end{equation}
Computing the saddle point equations $x-2 {\tilde h}_i x^{-1}=0$ and $y-2 {\tilde h}_j y^{-1}=0$,
we find four saddle points:
\begin{equation}
    x=\pm \sqrt{2 {\tilde h}_i}\quad \mathrm{with}\quad y=\pm \sqrt{2 {\tilde h}_j}\;.
\end{equation}
Therefore, we see that the term associated with $f$ is the same for any $X_0$ and is
\begin{equation}
    e^{-Nf(X_0)}=(2 {\tilde h}_i)^{N {\tilde h}_i}(2 {\tilde h}_j)^{N {\tilde h}_j}e^{-N( {\tilde h}_i+ {\tilde h}_j)}\,,
\end{equation}
and therefore, for the Hessian,
\begin{equation}
    H(f)(X_0)=
    \begin{pmatrix}
        2&0\\
        0&2
    \end{pmatrix}\;
\quad 
{\rm {and}} 
\quad 
    \sqrt{ {\rm det}\big(H(f)(X_0)\big)}=2\;.
\end{equation}
Additionally, 
\begin{equation}
    \sum_{X_0}g(X_0)=2\frac{2 {\tilde h}_i-2 {\tilde h}_j}{(\sqrt{2 {\tilde h}_i}+\sqrt{2 {\tilde h}_j})^2+\epsilon^2}+2\frac{2 {\tilde h}_i-2 {\tilde h}_j}{(\sqrt{2 {\tilde h}_i}-\sqrt{2 {\tilde h}_j})^2+\epsilon^2}    =4\frac{( {\tilde h}_i- {\tilde h}_j)( {\tilde h}_i+ {\tilde h}_j+\epsilon^2/2)}{( {\tilde h}_i- {\tilde h}_j)^2+\epsilon^2( {\tilde h}_i+ {\tilde h}_j)+\epsilon^4/4} 
    \,.
\end{equation}
Joining everything, we obtain the saddle point approximate for the regularized character of $A^2$ as 
\begin{equation}
        \langle \chi_r(A^2)\rangle_\epsilon 
        = 
        {\tilde c}_N \underset{i,j}{
    \rm Pf}\left[\frac{1}{2}e^{-N( {\tilde h}_i+ {\tilde h}_j)}(2 {\tilde h}_i)^{N {\tilde h}_i}(2 {\tilde h}_j)^{N {\tilde h}_j}4\frac{( {\tilde h}_i- {\tilde h}_j)( {\tilde h}_i+ {\tilde h}_j+\epsilon^2/2)}{({\tilde h}_i-{\tilde h}_j)^2+\epsilon^2({\tilde h}_i+{\tilde h}_j)+\epsilon^4/4}\right] (1+{\mathcal O}(N^{-1}))\;.
    \label{eq:chisaddlept}
\end{equation}
Using some Pfaffian properties we can simplify the expression \eqref{eq:chisaddlept} and get
\begin{equation}
    \langle \chi_r(A^2)\rangle_\epsilon 
    = 
    {\tilde c}_N2^{\frac{N}{2}}\prod_k e^{-N {\tilde h}_k}(2 {\tilde h}_k)^{N {\tilde h}_k}
    \underset{i,j}{
    \rm Pf}\left[\frac{( {\tilde h}_i- {\tilde h}_j)( {\tilde h}_i+ {\tilde h}_j+\epsilon^2/2)}{( {\tilde h}_i- {\tilde h}_j)^2+\epsilon^2( {\tilde h}_i+ {\tilde h}_j)+\epsilon^4/4}\right](1+{\mathcal O}(N^{-1}))\;.
\end{equation}
Applying the limit $N\rightarrow \infty$, we find
\begin{equation}
    \lim_{N\rightarrow \infty}
    \frac{\langle \chi_r(A^2)\rangle_\epsilon}
    {{\tilde c}_N2^{\frac{N}{2}}\prod_k e^{-N {\tilde h}_k}(2 {\tilde h}_k)^{N {\tilde h}_k}
    \underset{i,j}{
    \rm Pf}\left[\frac{( {\tilde h}_i- {\tilde h}_j)( {\tilde h}_i+ {\tilde h}_j+\epsilon^2/2)}{( {\tilde h}_i- {\tilde h}_j)^2+\epsilon^2( {\tilde h}_i+ {\tilde h}_j)+\epsilon^4/4}\right]} = 1\;.
\label{eq:limN}
\end{equation}
\end{proof}

One may wish to apply $\epsilon \rightarrow 0$ limit to \eqref{eq:limN} and if the limits $\epsilon \rightarrow 0$ and $N \rightarrow \infty$ commute, then, \eqref{eq:limN} can be manipulated to say that in the large $N$ limit, 
 ${\langle \chi_r(A^2)\rangle_0}$ is equal to
\beq
\label{eq:conj}
   {{\tilde c}_N2^{\frac{N}{2}}\prod_k e^{-N {\tilde h}_k}(2 {\tilde h}_k)^{N {\tilde h}_k}
    \underset{i,j}{
   \rm Pf}\left[\frac{ {\tilde h}_i+ {\tilde h}_j}{ {\tilde h}_i- {\tilde h}_j}\right]} \;.
\eeq

We computed $\langle \chi_r(A^2)\rangle_0$ using \eqref{resfinn} and \eqref{eq:conj} for $1 \le N \le 30$ (we should note here that if we use integral form for $\langle \chi_r(A^2)\rangle_0$, then even $N=6$ the computation becomes slow) using Mathematica for trivial, defining, and determinant representations.
However, for the above values of $N$ that we tested, the values of the expression \eqref{eq:conj} do not converge to the values computed using the expression achieved in \eqref{resfinn}.

\medskip

\vskip 20pt

Now, we present first a new way of computing $\langle \chi_R(A)\rangle_0$ below in Proposition \ref{theorem:chiAR}, whose similar technique is used to compute $\langle \chi_R(A^2)\rangle_0$ in Theorem \ref{thm:chiA2R}.

\begin{proposition}
\label{theorem:chiAR}
Let $A$ be an $N \times N$ Hermitian matrix under the Gaussian measure. Given a representation $R$ of $\mathrm{GL}(N)$ defined through the set of shifted weights $h_i$, $i=1,...N,$ and considering $\chi_R$ the character in the representation $R$, the following holds true:
\begin{equation}
    \langle \chi_R(A)\rangle_0 
    = 
    \frac{\chi_R(\mathbb{1})\chi_R(C_2)}{\chi_R(C_1)}= (-1)^{\frac{1}{2}\frac{N}{2}\left(\frac{N}{2}-1\right)}N^{-\frac{n}{2}}\frac{\Delta(h)}{\prod_{i=0}^{N-1}i!}\frac{\prod_i(h_i^e-1)!!h_i^o!!}{\prod_{i,j}(h_i^o-h_j^e)}\,,
\end{equation}
where the numbers $h$ are separated in a set of $\lceil N/2\rceil$ even numbers $h^e$ and $\lfloor N/2 \rfloor$ odd numbers $h^o$. If such a separation is not possible, then the average is $0$.
\end{proposition}

\begin{proof}
For $A\in \mathrm{GL}(N)$, using character orthogonality for the symmetric group $S_n$, we can write the character of $A$ as 
\bea
      \chi_R(A) 
       &=&
       \sum_{r\vdash n} \delta_{R,r}\chi_r(A)
       \crcr
       &=&
       \sum_{r\vdash n} 
       \Big(
       \sum_{\sigma\in S_n}\frac{1}{n!}\bar{\chi}_R(\sigma)\chi_r(\sigma)
       \Big)
       \chi_r(A)
       \,,
\eea
where the bar on
$\bar{\chi}_R(\sigma)$ denotes complex conjugate.
Interchanging the sums and  using Schur-Weil duality,
\begin{equation}\label{swd}
   \sum_{r\vdash n} \chi_r(\sigma)\chi_r(A)=\mathrm{Tr}(\sigma A^{\otimes n})
   \,,
\end{equation}
we obtain
\begin{equation}
   \chi_R(A) =\sum_{\sigma\in S_n}\frac{1}{n!}\bar{\chi}_R(\sigma)\mathrm{Tr}(\sigma A^{\otimes n})
   \,.
   \label{eq:chiATrsigmaAn}
\end{equation}
Now, we take the average of the above quantity. Wick's probability theorem tells us that
\begin{equation}
    \langle A^{\otimes n} \rangle_0 = N^{-\frac{n}{2}}\sum_{\gamma\in [2^{\frac{n}{2}}]}\gamma
    \,,
\end{equation}
where $[2^\frac{n}{2}]$ is the conjugacy class of permutations in $S_n$ with $n/2$ 2-cycles. Therefore we obtain
\bea
   \langle \chi_R(A) \rangle_0
   &=&
   \sum_{\sigma\in S_n}\sum_{\gamma\in [2^{\frac{n}{2}}]}N^{-\frac{n}{2}}\frac{1}{n!}\bar{\chi}_R(\sigma)\mathrm{Tr}(\sigma \gamma)
   \crcr
   &=&
   \sum_{\sigma\in S_n}\sum_{\gamma\in [2^{\frac{n}{2}}]}\sum_{r\vdash n}N^{-\frac{n}{2}}\frac{1}{n!}\bar{\chi}_R(\sigma)\chi_r(\sigma \gamma)\chi_r(\mathbb{1})
   \,,
\eea
where we used \eqref{swd} in the last equality, with $A ={\mathbb 1}$.
Again using orthogonality relations, rewrite
\begin{equation}\label{ort}
    \sum_{\sigma\in S_n}\frac{1}{n!}\bar{\chi}_R(\sigma)\chi_r(\sigma \gamma)=\delta_{R,r}\frac{1}{s_R}\chi_R(\gamma)
    \,,
\end{equation}
where we denote the dimension of the $S_n$ representation $s_R=\chi_R({\rm id})$, where ${\rm id}$ is the identity permutation, and 
\bea
   \langle \chi_R(A) \rangle_0
   &=&
   \sum_{\gamma\in [2^{\frac{n}{2}}]}\sum_{r\vdash n}N^{-\frac{n}{2}}\delta_{R,r}\frac{1}{s_R}\chi_R(\gamma) \; d_r
   \crcr
   &=&
   N^{-\frac{n}{2}}\frac{d_R}{s_R}\sum_{\gamma\in [2^{\frac{n}{2}}]}\chi_R(\gamma)
   \,.
   \label{avc}
\eea
In order to perform $\sum_{\gamma\in [2^{\frac{n}{2}}]}$, we use the following trick.
Let us first return to the relation \eqref{eq:chiATrsigmaAn} for some matrix $M$,
\begin{equation}\label{chch2}
   \chi_R(M) =\sum_{\sigma\in S_n}\frac{1}{n!}\bar{\chi}_R(\sigma)\mathrm{Tr}(\sigma M^{\otimes n})\;.
\end{equation}
We observe that \eqref{chch2} sums over all elements in $S_n$, whereas \eqref{avc} sums over a subset of $S_n$.
We aim to extract \eqref{avc} from \eqref{chch2}. 
We choose $M$ such that the summation in $S_n$ is restricted to elements of $[2^{\frac{n}{2}}]$, achieved by  $\mathrm{Tr}(\sigma M^{\otimes n})= 0$ for $\sigma$ not in $[2^\frac{n}{2}]$ and constant for when $\sigma$ is in $[2^\frac{n}{2}]$. Using the notation $[\prod_k k^{c_k}]$ for the cycle $[\sigma]$ which $\sigma$ belongs to, we wish to find $M$ such that
\begin{equation}
    \mathrm{Tr}(\sigma M^{\otimes n})= a \;\delta_{[\sigma],[2^{\frac{n}{2}}]}=a\, \delta_{c_2,\frac{n}{2}}\prod_{k\neq 2} \delta_{c_k,0}
    \,.
\end{equation}
for some constant $a$ (that might depend on $N$ or $n$ but not on $\sigma$). But also, we can write
\begin{equation}\label{cytr}
    \mathrm{Tr}(\sigma M^{\otimes n})=\prod_k \mathrm{Tr}(M^k)^{c_k}
    \,,
\end{equation}
leading us to the identification, for $k\neq 2$ ,
\begin{equation}
    \mathrm{Tr}(M^k)^{c_k}= \delta_{c_k,0}
    \,.
\end{equation}
Hence, the possibility of a nonzero $c_k$ with a general permutation $\sigma$ tells us that
\begin{equation}
    \mathrm{Tr}(M^k)=0\;
\end{equation}
for $k\neq2$. Then, recalling 
\begin{equation}\label{c2r}
    \mathrm{Tr}(C_2^k)=N\delta_{k,2}
    \,,
\end{equation}
we conclude $M=C_2$ with $a=N^\frac{n}{2}$ is a possible solution. Then, setting $M=C_2$ in \eqref{chch2} with \eqref{cytr}, we find that 
\begin{equation}\label{c2sl}
   \chi_R(C_2) =\sum_{\sigma\in S_n}\frac{1}{n!}\bar{\chi}_R(\sigma)\prod_k \mathrm{Tr}(C_2^k)^{c_k}\;.
\end{equation}
Finally, using \eqref{c2r} 
\begin{equation}
   \chi_R(C_2) =\sum_{\sigma\in [2^{\frac{n}{2}}]}\frac{1}{n!}\bar{\chi}_R(\sigma)N^{\frac{n}{2}}\,,
\end{equation}
we succeeded in restricting the sum to $[2^\frac{n}{2}]$. 
Identifying the last summation in \eqref{avc} as
\begin{equation}\label{s2n2}
   \sum_{\sigma\in [2^{\frac{n}{2}}]}\chi_R(\sigma)
   =
   n!N^{-\frac{n}{2}}\bar{\chi}_R(C_2)\,,
\end{equation}
we achieve an expression for the average of character,
\begin{equation}\label{drlr}
   \langle \chi_R(A) \rangle_0 = n!N^{-n}\frac{d_R}{s_R}\bar{\chi}_R(C_2)\;.
\end{equation}

Let us now compute $s_R$.
We perform a similar trick as above, by considering $[1^n]$ instead of $[2^\frac{n}{2}]$.
This time there is no summation in $[1^n]$ because the only element in this class is the identity permutation. 
We can show in a similar fashion that we find the result through the matrix $C_1$ instead of $C_2$. 
Skipping to the equivalent of \eqref{c2sl}, just changing the condition from $k=2$ to $k=1$ in relevant places, we are able to isolate the character of the identity permutation,
\begin{equation}
    \chi_R(C_1)=\frac{1}{n!}\bar{\chi}_R(id)N^n\;.
\end{equation}
Identifying $s_R = \bar{\chi}_R(id) = {\chi}_R(id)$, 
the dimension of the $S_n$ representation can be written in terms of $\chi_R(C_1)$ by
\begin{equation}\label{lrr}
    s_R=n!N^{-n}\bar{\chi}_R(C_1)\;.
\end{equation}
Finally, using \eqref{lrr} in \eqref{drlr} also with $d_R=\chi_R(\mathbb{1})$ we find that 
\begin{equation}
   \langle \chi_R(A) \rangle_0 = \frac{\chi_R(\mathbb{1})\chi_R(C_2)}{\chi_R(C_1)}\;.
   \label{eq:chiAprodchi}
\end{equation}

Additionally, 
with a formula for the character of $C_m$ 
\cite{Kazakov_1996}
\begin{equation}
    \chi_R(C_m)=c\;\; \left(\frac{N}{m}\right)^{\frac{1}{m}\sum_i h_i}\;\prod_{\epsilon=0}^{m-1}\frac{\Delta(h^{(\epsilon)})}{\prod_i\left(\frac{h_i^{(\epsilon)}-\epsilon}{m}\right)!}\,\mathrm{sgn}\left[\prod_{0\leq \epsilon_1<\epsilon_2\leq m-1}\prod_{i,j}(h^{(\epsilon_2)}_i-h^{(\epsilon_1)}_j)\right]
    \,,
    \label{eq:chiCm}
\end{equation}
where $\{h^{(\epsilon)}\}=\{a\in {\{h\}}\;|\; a=\epsilon \;\mathrm{mod}\; m \}$,
we can also find a different expression for $\langle \chi_R(A) \rangle_0$, by inserting \eqref{eq:chiCm} into \eqref{eq:chiAprodchi}
\begin{equation}
\label{chchch}
    \langle\chi_R(A)\rangle_0
    =
    (-1)^{\frac{1}{2}\lceil\frac{N}{2}\rceil(\lceil\frac{N}{2}\rceil-1)}N^{-\frac{n}{2}}\frac{\Delta(h)}{\prod_{i=0}^{N-1}i!}\frac{\prod_i (h^{(e)}_j-1)!!h^{(o)}!!}{\prod_{i,j}(h^{(o)}_i-h^{(e)}_j)}\;.
\end{equation}
This expression is written in terms of the highest weights $\{h\}$.
\end{proof}

Following a similar argument as Proposition \ref{theorem:chiAR}, we compute $\langle\chi_r(A^2)\rangle_0$.

\begin{theorem}
\label{thm:chiA2R}
Let $A$ be an $N \times N$ Hermitian matrix under the Gaussian measure. Given a representation $r$ of $\mathrm{GL}(N)$ defined through the set of shifted weights $h_i$, $i=1,...N,$ and considering $\chi_r$ the character in that representation, the following holds true:
\begin{equation}
   \langle\chi_r(A^2)\rangle_0 =\frac{\chi_r(\mathbb{1})^2}{\chi_r(C_1)}\;\big(1+{\mathcal O}(N^{-2})\big)=N^{-n}\prod_i h_i!\,\Delta(h)\;\big(1+{\mathcal O}(N^{-2})\big)\;.
\end{equation}
\end{theorem}

\begin{proof}
We start from the relation
\begin{equation}
   \chi_r(A^2) =\sum_{\sigma\in S_n}\frac{1}{n!}\bar{\chi}_R(\sigma)\mathrm{Tr}(\sigma (A^2)^{\otimes n})\;.
   \label{eq:chiA2start}
\end{equation}
We wish to express in terms of $A$ instead of $A^2$. 
We use the relation
\begin{equation}
    \mathrm{Tr}\big(\sigma (A^2)^{\otimes n}\big)=\mathrm{Tr}\big((\sigma\otimes 1) A^{\otimes 2n} \alpha\big)
\end{equation}
where $\alpha=(1,1+n)(2,2+n)\cdots (n,2n)$ and on the left the trace considers the vector space $\mathbb{C}_N^{\otimes n}$, while on the right the vector space considered is $\mathbb{C}_N^{\otimes 2n}$. This relation can be shown by using the definition of trace on both sides.
We then just write \eqref{eq:chiA2start},
\begin{equation}
\label{eq:chaA2n}
   \chi_r(A^2) =\sum_{\sigma\in S_n}\frac{1}{n!}\bar{\chi}_R(\sigma)\mathrm{Tr}((\sigma\otimes 1) A^{\otimes 2n} \alpha)\,,
\end{equation}
and just like before in Theorem \ref{theorem:chiAR}, taking the average and using Wick's theorem we obtain
\begin{equation}
   \langle\chi_r(A^2)\rangle_0 =\sum_{\sigma\in S_n}\sum_{\gamma\in [2^n]}N^{-n}\frac{1}{n!}\bar{\chi}_R(\sigma)\mathrm{Tr}((\sigma\otimes 1)\, \gamma\, \alpha)\;.
\end{equation}
At this point we need a relation that we prove in Appendix \ref{sec:app}, which is
\begin{equation}\label{app2}
    \sum_{\gamma \in [2^n]}\mathrm{Tr}((\sigma\otimes 1)\,\gamma\,\alpha)=\sum_{\rho \in S_n}\mathrm{Tr}(\rho)\mathrm{Tr}(\sigma\,\rho)\;\big(1+{\mathcal O}(N^{-2})\big)\,,
\end{equation}
where on the right hand side, the traces act on $\mathbb{C}_N^{\otimes n}$. Using this equation we then get
\begin{equation}
   \langle\chi_r(A^2)\rangle_0 =\sum_{\sigma\in S_n}\sum_{\rho \in S_n}N^{-n}\frac{1}{n!}\bar{\chi}_R(\sigma)\mathrm{Tr}(\rho)\mathrm{Tr}(\sigma\,\rho)\;\big(1+{\mathcal O}(N^{-2})\big)\;.
\end{equation}
The summation over $\sigma$ can then be evaluated by using the summation over representations \eqref{swd} with the identity matrix in place of $A$ and the summation over permutation \eqref{ort};
\begin{equation}
   \langle\chi_r(A^2)\rangle_0 
   =
   \sum_{\rho \in S_n}N^{-n}\mathrm{Tr}(\rho)\frac{d_R}{s_R}\chi_r(\rho)\;\big(1+{\mathcal O}(N^{-2})\big)\;.
\end{equation}
Repeating for the summation in $\rho$, and using $\chi_r(id)=s_R$ and \eqref{lrr}, we find
\begin{equation}
   \langle\chi_r(A^2)\rangle_0 
   =
   n!N^{-n}\frac{d_R^2}{s_R}\;(1+{\mathcal O}(N^{-2}))=
   \frac{\chi_r(\mathbb{1})^2}{\chi_r(C_1)}\;(1+{\mathcal O}(N^{-2}))\;.
\end{equation}
Finally, using \eqref{dh} and \eqref{eq:chiCm} we find that
\begin{equation}
   \langle\chi_r(A^2)\rangle_0 
   =
   N^{-n}\prod_i h_i!\,\Delta(h)\;\big(1+{\mathcal O}(N^{-2})\big)\;.
\end{equation}
This interesting new result is valid considering $N\gg n$.
\end{proof}

\section{Conclusions}
\label{sec:conc}

A comprehensive understanding of quantum gravity-matter systems is not only important for phenomenological applications but also to check the consistency of the the underlying quantum theory. A fundamental theory that encompasses quantum-gravitational degrees of freedom as well as matter fluctuations depends on the synergy between the collective quantum fluctuations of the full theory. A paradigmatic example is QCD that loses  asymptotic freedom and, thus, UV-completeness when the number of fermions is sufficiently large. Likewise, an asymptotically safe theory of quantum gravity can lose its UV fixed point for a sufficiently intricate matter content. Therefore, understanding the impact of matter degrees of freedom to the continuum limit of a candidate theory of quantum gravity is an essential ingredient in its construction. Conversely, quantum-gravity fluctuations can affect the dynamics of the underlying matter fields and trigger non-trivial effects that can leave imprints in the infrared. Furthermore, the large body of evidence reported by the CDT community reveals that the implementation of causality constraints can be central for a suitable continuum limit of the lattice-regularized path integral of quantum gravity. Thus, in view of what was said above, it is natural to investigate the continuum limit of CDT coupled to matter degrees of freedom. Actually, it is an active field of interest to understand the relation, if any, between the Euclidean and Lorentzian formalisms, see, e.g., \cite{Hamber:2012zm,Baldazzi:2018mtl,Visser:2021ucg,Gerard:2022fgu,Kontsevich:2021dmb}.

In the present work, we have defined a toy model that generated two-dimensional CDT configurations coupled to Ising model degrees of freedom. This was proposed under the framework of dually-weighted multi-matrix models. 
Matrix models enjoy rich analytical tools and, in many occasions, allow for the establishment of mathematical rigorous results such as 
eigenvalue decomposition techniques to explicitly compute matrix-integrals, large $N$ expansion where the sphere topology dominates at leading order in a Feynman expansion of the partition function \cite{tHooft:1973alw,Francesco_1995,Brezin:1977sv,KAZAKOV1985295,DAVID1985543}, double-scaling limit in which all topologies are taken into account\cite{BREZIN1990144,Shenker1991,Gross:1989vs}, topological recursion which lets us recursively solve the loop equations of matrix models and bridges combinatorial maps with algebraic and enumerative geometry \cite{Eynard:2007kz}, Weingarten calculus to compute matrix integrals \cite{collins2021weingarten}, etc.
Matrix models indeed generate the Brownian sphere at criticality and rigorously proven to be equivalent to 2d Liouville gravity \cite{miller2021liouville,Miller:2016dez,Marckert_2006,gall2006scaling,Le_Gall_2007,gall2012scaling}. Equipped with such abundant and rich tools, they therefore provide with us an enticing platform to further explore interesting mathematics and physics, with also possibly more physically relevant (higher dimensional) generalization in the context of quantum gravity; Matrix models have a higher dimensional generalization, i.e., tensor models \cite{Gurau:2009tw,Gurau:2011xp,Bonzom:2011zz,Rivasseau:2011hm,Rivasseau:2012yp,Rivasseau:2013uca,Rivasseau_2016,Rivasseau:2016wvy} where it is of interest to understand how to incorporate causality.

The initial hope of formulating CDT coupled to the Ising model as a multi-matrix model was to find an exact/analytical solution for it (in other words compute exactly the partition function), possibly using new techniques developed in the last decade since the work by Benedetti-Henson in \cite{Benedetti:2008hc}. For example, one new interesting result \cite{collins2022tensor} in recent works is the generalization of the Harish-Chandra–Itzykson–Zuber (HCIZ) integral, where the generalization considers integrals over tensor powers of $U(N)$. They studied the integral $\int dU\;e^{t\tr (AUBU^*)}$, where $U\in U(N)^{\otimes D}$, $A$ and $B$ are self-adjoint operators in $(\mathbb{C}^N)^{\otimes D}$, and $N$ and $D$ are positive integers. The HCIZ result is recovered by setting $D=1$. 
For our purposes, a generalization of the HCIZ integral for a trace with a different power would have been useful. That is, $\int dU\;e^{t\tr (AUBU^*)^k}$, where $U\in U(N)$, $A$ and $B$ are self-adjoint operators in $\mathbb{C}^N$, and $N$ and $k$ are positive integers.
Instead, we made use of Weingarten calculus (in Proposition \ref{unitaryint}), Schur-Weyl duality (in Proposition \ref{theorem:chiAR}, Theorem \ref{thm:chiA2R}), theory of symmetric group algebra (in Proposition \ref{theorem:chiAR}, Theorem \ref{thm:chiA2R}, and Lemma \ref{lemapp}), etc to achieve our results in this paper.

\medskip

Even though matrix models enjoy vast amount of analytical results, it seems that solving exactly the partition functions of the pure CDT-like matrix model of Benedetti-Henson and the CDT-like matrix model coupled with Ising models remains quite difficult.

One problem we faced was the necessity to do the character expansion more than once for the CDT-like matrix model with Ising model. See \eqref{2nd} and \eqref{doubleexpansion}.
In Proposition \ref{unitaryint}, the appearance of multiple representations from the expansion required us to use the theory of decomposition of reducible representations. The current state of this theory is not sufficient for our purposes, since it lacks closed formulas in terms of the representations involved; the existence of algorithms is not enough. 
For example, given a representation $R$ in terms of a set of shifted highest weights $\{h\}$, the decomposition of the representation $R\otimes R$ is doable through the Littlewood-Richardson rule. 
However, this process is not written in terms of functions which we could apply the standard theory of 
calculus (e.g., derivatives and integrals) to; for example,  closed formulas for the Littlewood-Richardson rule, as well as for the Clebsch-Gordan coefficients are not known.

Indeed, we foresee a similar problem when we consider solving for $\langle \chi_{\{h\}}(A^2)\rangle$, by expressing the character in terms of the symmetric square and alternating representations (which are, in general, reducible). 
There exists a relation $\{h\}\otimes\{h\}=\mathrm{Sym}^2\{h\}\oplus\mathrm{Alt}^2\{h\}$. 
This expression tells us that $\chi_{\{h\}}(A)^2=\chi_{\mathrm{Sym}^2\{h\}}(A)+\chi_{\mathrm{Alt}^2\{h\}}(A)$. Another known formula is $\chi_{\{h\}}(A^2)=\chi_{\mathrm{Sym}^2\{h\}}(A)-\chi_{\mathrm{Alt}^2\{h\}}(A)$. 
This last decomposition can be done with the Carré--Leclerc domino tableaux algorithm, in a similar way the Littlewood-Richardson rule gives an algorithm to decompose the product of irreducible representations into a direct sum of irreducible representations. 
Nevertheless, this is just an algorithm, and does not give a closed formula to which representations contribute, and it seems that such formula is not known.
Therefore, we seem to encounter the similar problem as we did in the Proposition \ref{unitaryint}, where Clebsch-Gordan coefficients are unknown. In order to solve for the partition functions of the CDT-like matrix models, we do need general explicit formulas for such coefficients.

Another difficulty lies in  the fact that the representations that contribute most to the character expansion are  expected to have the size proportional to $N^2$. 
More specifically, the number of shifted weights is $N$ and each shifted weight grows proportionally to $N$. See Theorem \ref{thm:chiA2R}.
When applying the theory of symmetric group algebra, one of the symmetric groups considered is $S_n$, where $n$ is the size of the representation (corresponding to the number of boxes of Young diagram of a given representation). 
In principle, the size of the representation $n$ should be proportional to $N^2$ \cite{Kazakov_1999} \cite{Douglas_1993}. 
Therefore, if we want to compute the partition functions \eqref{doubleexpansion} and \eqref{eq:Zhy}, then we need to take care of the $N$ dependence in $n$.
However, in this current work, our results in Theorem \ref{thm:chiA2R} and Lemma \ref{lemapp} relied on us taking large $N$ limit but keeping $n$ finite.
Working with $S_n$ for large $n$ is challenging because the dependence of $n$ is involved in expressions like \eqref{eq:chaA2n}.

\medskip

Nevertheless, we have achieved several new results concerning the averages of character of Hermitian matrix or Hermitian matrix squared for a given representation, which may be of interest in the mathematics community and alike.

In particular, the results in Theorems \ref{thm:epilonN}, \ref{thm:chiA2finiteN} and \ref{thm:chiA2R} add to the known properties of the decomposition of the product of representations $r\otimes r$ in terms of irreducible representations.

In Theorem \ref{thm:chiA2finiteN}, the expression of $\langle\chi_R(A^2)\rangle$ in \eqref{resfinn}  is for finite $N$ and $n$. This result reduces the average integral in \eqref{avdef} to a finite sum over integers, therefore, more computable. The next step is either evaluate this sum, that has only products and ratios of factorials as terms, or possibly evaluating the Pfaffian without computing the sum. 

We also found an expression for $\langle \chi_R(A^2)\rangle$ at large $N$ and $n$  in Proposition \ref{thm:epilonN}. 
The expression itself of the Pfaffian ${\rm Pf}\left[\frac{{\tilde h}_i+{\tilde h}_j}{{\tilde h}_i-{\tilde h}_j}\right]$ is similar to the known result in \eqref{pffpr}, except that in \eqref{lnap} each matrix element is inverse to the respective element in \eqref{pffpr} and is unknown to our knowledge.

In Theorem \ref{thm:chiA2R}, we explicitly found the leading order expression of the character of $A^2$ for a given representation in the large $N$ limit for a finite $n$. 
This result is related to cell modules in Brauer algebra \cite{bulgakova:tel-02554375}.

Additionally, the results Theorem \ref{thm:Cm} on the matrix $C_m$ (which is responsible for yielding the causal structure to the Feynman graphs generated) are new to our knowledge. 
We explicitly found the distribution of eigenvalues of $C_m$.

\medskip

For future works, we can consider several possible extensions and additions of our current work.
It should be possible to extend the result of Theorem \ref{thm:chiA2R} to any powers of $A$ by generalizing the expression used in \eqref{app2}. This expression, shown in Appendix \ref{sec:app}, is obtained by combinatorial evaluations, which we expect can be applied to extend our result to higher powers of $A$.

In \cite{Benedetti:2008hc}, the critical value of the coupling $g$ in \eqref{eq:CDTMMDarioABC2} was computed. By solving for the distribution of highest weights, they were able to observe a critical behavior where the distribution extended over the entire positive real line. This criticality happened for coupling $g=1/2$, which is the same as the criticality found in \cite{Ambj_rn_1998}.
We may be able to obtain the critical property of partition functions \eqref{eq:CDTMMDarioABC2} if we can understand the properties of Pfaffians ${\rm Pf}\left[\frac{{\tilde h}_i+{\tilde h}_j}{{\tilde h}_i-{\tilde h}_j}\right]$ in \eqref{lnap} well enough, even if we do not know the explicit expression for it.

In summary, the present work puts forward a dually-weighted multi-matrix model that, for a particular choice of weight, corresponds to two-dimensional CDT coupled to the Ising model. To the best of our knowledge, there is no known exact solution for CDT coupled to Ising even in two dimensions. Hence, one possibility to be explored is whether the model here proposed can be well-suited for different approximations of the CDT-Ising model as, e.g., by means of Monte-Carlo simulations. A reformulation of the underlying theory by a different formalism can lead to new insights and this is worth exploring in the future. Moreover, we have now the possibility to couple more sophisticated matter models to two-dimensional CDT in the matrix-model frameowork, such as the Potts model. This will be reported elsewhere. Lastly, the dually-weighted multi-matrix model can be investigated on its own, i.e., beyond the choice of the causal constraint as mainly emphasized in this work.

\section*{Acknowledgements}

The authors would like to thank 
Dario Benedetti, Taro Kimura, Benoit Collins, Sanjaye Ramgoolam, Mark Wildon, Luca Lionni, Paul Wedrich, Liron Speyer, and Christopher Chung, Dmitri Grinko, Cihan Pazarbasi and Shinobu Hikami for discussions and valuable inputs. 
We would also like to thank the Insitut Henri Poincar\'e thematic program ``Quantum Gravity, Random Geometry and Holography", 
2023, 
where
part of the work was done
and 
many 
discussions took place.  
ADP acknowledges CNPq under the grant PQ-2 (312211/2022-8), FAPERJ under the “Jovem Cientista do Nosso Estado” program (E26/202.800/2019 and E-26/205.924/2022), and NWO under the VENI Grant (VI.Veni.192.109) for financial support.

\section*{Appendix}
\label{sec:app}

We show how to derive the relation in $\eqref{app2}$. 
First we show a more general relation in the following Lemma \ref{lemapp}, and we obtain \eqref{app2} by performing left multiplication on both sides of the equation \eqref{ger} by $\sigma$ and taking the trace.

\begin{lemma}\label{lemapp}
Consider the trace $\tr$ over a vector space $C_N^{\otimes2n}=C_N^{\otimes n}\otimes C_N^{\otimes n}$. Let $\mathrm{P_2 Tr}$ be the partial trace over the second $C_N^{\otimes n}$ vector space, and $\alpha=(1,1+n)(2,2+n)\cdots (n,2n)$ be the permutation that switches the first and second $C_N^{\otimes n}$ vector spaces. For a given $n$,
\begin{equation}\label{ger}
    \sum_{\gamma \in [2^n]}\mathrm{P_2Tr}(\gamma\,\alpha)
    =
    \sum_{\rho \in S_n}\mathrm{Tr}(\rho)\,\rho\;\big(1+{\mathcal O}(N^{-2})\big)\;.
\end{equation}
\end{lemma}
Before we begin our proof, let us see an example of \eqref{ger}, checking it for $n=2$. The permutations in $[2^2]$ are shown in Fig.\,\ref{22}.
For each permutation $\sigma$ in $[2^2]$, $\mathrm{P_2 Tr} (\sigma\,\alpha)$ is shown in Fig.\,\ref{tr22}.
This way, we indeed check
\begin{equation}    \sum_{\gamma\in[2^2]}\mathrm{P_2Tr}(\gamma\;\alpha)= [1+N^2](1)(2)+N(1\, 2) = [N^2(1)(2)+N(1\,2)][1+{\mathcal O}(N^{-2})]\;.
\end{equation}

\begin{figure}[h]
    \begin{subfigure}{0.3\textwidth}
    \centering
    \includegraphics[width=.4\linewidth]{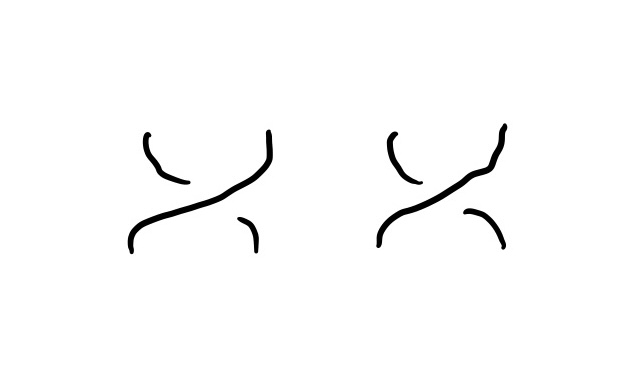}
    \caption{$(1\, 2)(3\, 4)$}
    \end{subfigure}
    \begin{subfigure}{0.3\textwidth}
    \centering
    \includegraphics[width=.4\linewidth]{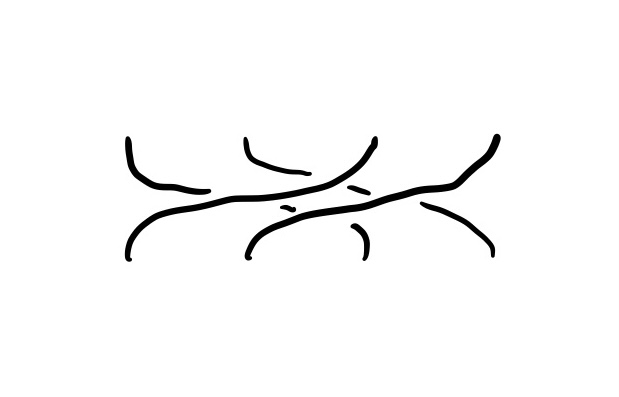}
    \caption{$(1\, 3)(2\, 4)$}
    \end{subfigure}
    \begin{subfigure}{0.3\textwidth}
    \centering
    \includegraphics[width=.4\linewidth]{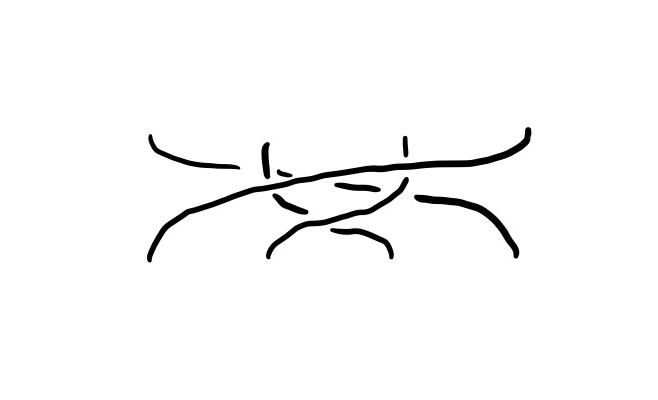}
    \caption{$(1\, 4)(2\, 3)$}
    \end{subfigure}
    \caption{Permutations in $[2^2]$.}
    \label{22}
\end{figure}
\begin{figure}[t]
    \begin{subfigure}{0.3\textwidth}
    \centering
    \includegraphics[width=.6\linewidth]{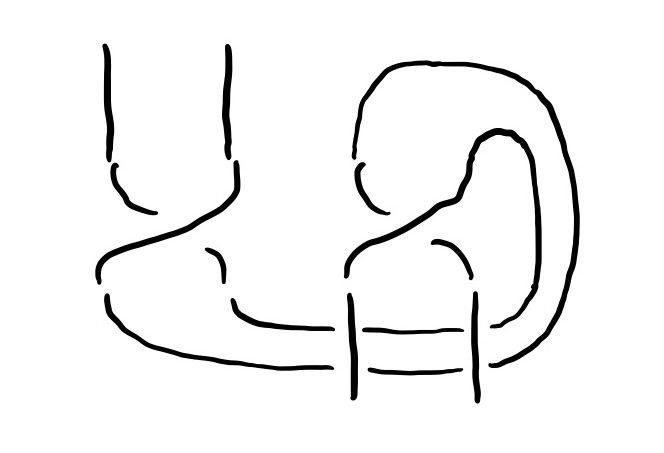}
    \caption{$(1)(2)$}
    \end{subfigure}
    \begin{subfigure}{0.3\textwidth}
    \centering
    \includegraphics[width=.6\linewidth]{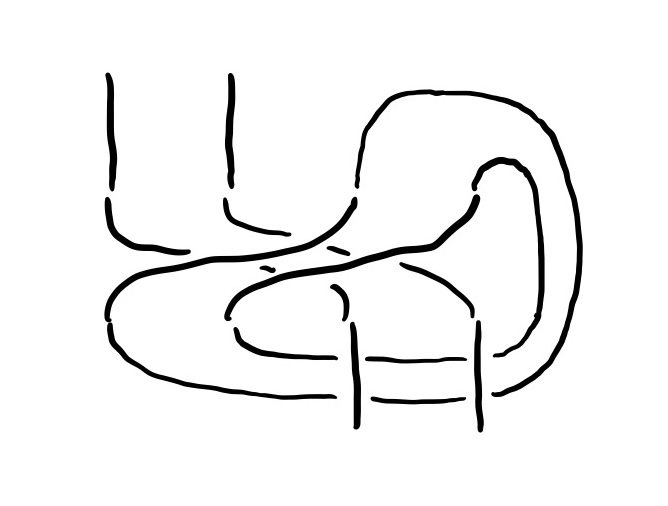}
    \caption{$N^2(1)(2)$}
    \end{subfigure}
    \begin{subfigure}{0.3\textwidth}
    \centering
    \includegraphics[width=.6\linewidth]{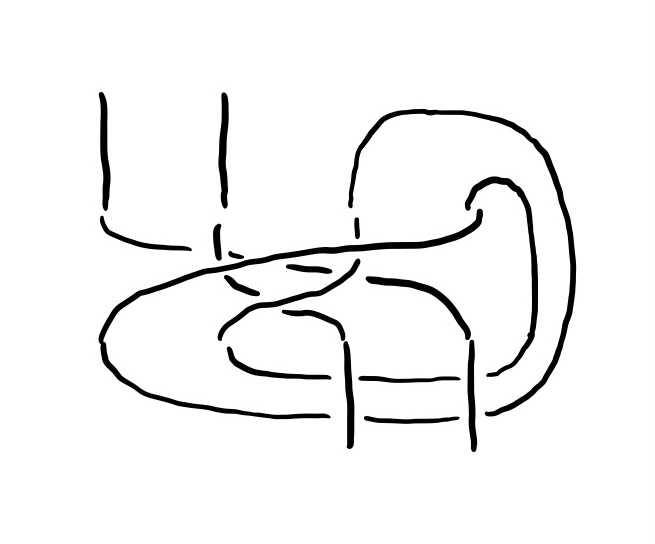}
    \caption{$N(1\, 2)$}
    \end{subfigure}
    \caption{Partial trace of permutations in $[2^2]$.}\label{tr22}
\end{figure}

\begin{proof}
Let $\gamma \in [2^n]$. Defining the interval of integers $[a,b]=\{x|x\in\mathbb{Z}, a\leq x \leq b\}$ and the sets $p_f=\{x|x\in[1,n], \gamma x \in[1,n]\}$ and $p_i=\{\alpha x | x\in [n+1,2n]| \gamma x \in [n+1,2n]\}$, where $\alpha=(1,1+n)(2,2+n)\cdots (n,2n)$. 
Remark $\alpha = \alpha^{-1}$ and $\gamma = \gamma^{-1}$.
Noticing the cardinalities of $p_f$ and $p_i$ are equal, we call $s=|p_f|=|p_i|$. Define the permutation $\nu_f\in [2^\frac{s}{2}]$ by
\begin{equation}\label{nuf}
    \nu_f x =\begin{cases}
\gamma x  & \text{ if } x\in p_f \\
x & \text{ if } x\in [1,n]\setminus p_f 
\end{cases}\;.
\end{equation}
Define the permutation $\nu_i\in[2^\frac{s}{2}]$ by
\begin{equation}\label{nui}
    \nu_i x =\begin{cases}
\alpha \gamma \alpha x & \text{ if } x\in p_i \\
x & \text{ if } x\in [1,n]\setminus p_i 
\end{cases}\;.
\end{equation}
Define also $\mu\in S_n$ by
\begin{equation}\label{mu}
    \mu x =\begin{cases}
 \gamma\alpha x & \text{ if } x\in [1,n]\setminus p_i   \\
(\gamma \alpha)^{-k(x)}x & \text{ if } x\in p_i 
\end{cases}\,,
\end{equation}
where $k(x)=\mathrm{min}\{m|m\geq0, (\gamma\alpha)^{-m} x\in p_f\}$. 
Checking each case we see that
\begin{equation}
    \gamma x =\begin{cases}
 \nu_f x & \text{ if } x\in p_f \\
 \alpha \nu_i\alpha x & \text{ if } x\in \alpha p_i \\
\mu\alpha x & \text{ if } x\in [n+1,2n]\setminus \alpha p_i \\
(\mu\alpha)^{-1} x & \text{ if } x\in [1,n]\setminus p_f 
\end{cases}\,,
\end{equation}
and applying a multiplication by $\alpha$,
\begin{equation}
    \gamma \alpha x =\begin{cases}
 \nu_f \alpha x & \text{ if } x\in \alpha p_f \\
 \alpha \nu_i x & \text{ if } x\in  p_i \\
\mu x & \text{ if } x\in [1,n]\setminus p_i \\
\alpha\mu^{-1}\alpha x & \text{ if } x\in [n+1,2n]\setminus \alpha p_f 
\end{cases}\;.
\end{equation}
Let us review the definition of partial trace of an operator $\eta\in\mathrm{End}(\mathbb{C}_N^{2n})$. The tensor coefficients of $\eta$ are defined by
\begin{equation}
    \eta_{i_1,...i_{2n}}^{j_1,...,j_{2n}}=\langle e_{j_1}, ..., e_{j_{2n}}|\,\eta\,| e_{i_1}, ..., e_{i_{2n}}\rangle \,,
    \label{eq:partialtracecoeff}
\end{equation}
where $| e_{i_1}, ..., e_{i_{2n}}\rangle$, for $1\leq i_k \leq N$ with $1\leq k \leq 2n$, is a basis of ${\mathbb{C}}_N^{2n}$.
In the case $\eta$ is a permutation, the coefficients \eqref{eq:partialtracecoeff} can be evaluated as
\begin{equation}
    \eta_{i_1,...,i_{2n}}^{j_1,...,j_{2n}}= \delta_{i_1}^{j_{\eta(1)}}...\delta_{i_{2n}}^{j_{\eta(2n)}}=\delta_{l_1}^{j_{1}}...\delta_{l_{2n}}^{j_{2n}}\delta_{i_1}^{l_{\eta(1)}}...\delta_{i_{2n}}^{l_{\eta(2n)}}\;.
\end{equation}
We define the partial trace of $\eta$ as a trace over its last $n$ indices,
\begin{equation}
    \mathrm{P_2Tr}(\eta)_{i_1,...,i_n}^{j_1,...,j_n}=\sum_{k_{n+1},...,k_{2n}}\eta_{i_1,...,i_n,k_{n+1},...,k_{2n}}^{j_1,...,j_n,k_{n+1},...,k_{2n}}\;.
\end{equation}
If $\eta$ is a permutation in $S_n$ represented in $\mathrm{End}(\mathbb{C}_N^{2n})$, its partial trace is equal to
\begin{equation}
    \mathrm{P_2Tr}(\eta)_{i_1,...,i_n}^{j_1,...,j_n}=\sum_{\substack{k_{n+1},...,k_{2n}\\l_1,...,l_{2n}}}\delta_{l_1}^{j_{1}}...\delta_{l_n}^{j_{n}}\delta_{l_{n+1}}^{k_{n+1}}...\delta_{l_{2n}}^{k_{2n}}\delta_{i_1}^{l_{\eta(1)}}...\delta_{i_n}^{l_{\eta(n)}}\delta_{k_{n+1}}^{l_{\eta(n+1)}}...\delta_{k_{2n}}^{l_{\eta(2n)}}\,,
\end{equation}
Let $\mathrm{cy}(\mu)$ be the number of cycles in $\mu$ (note that indeed $\mathrm{cy} (\mu) = \mathrm{cy} (\mu^{-1})$). Checking case by case, we then see that for $\eta = \gamma \alpha$, 
\begin{equation}
    \mathrm{P_2Tr}(\gamma \alpha) x =  \begin{cases}
 N^{\mathrm{cy}(\mu^{mu})-s}\, \mu \, x & \text{ if } x\in [1,n]\setminus p_i \\
N^{\mathrm{cy}(\mu^{-1})-s} \, \nu_f \, \mu^{-k(x)}\, \nu_i \, x & \text{ if } x \in p_i 
\end{cases}\;,
\end{equation}
where the power of N, which is $\mathrm{cy}(\mu)-s$, is the number of cycles of $\alpha\gamma$ which contains only elements from $n+1$ to $2n$. This power also counts the number of closed loops in the graphic representation, see Fig. \ref{tr22}.
By noticing that the two cases are both equal to $\nu_f\mu\nu_i$ restricted to their conditions, simply we can write
\begin{equation}\label{res}
    \mathrm{P_2Tr}(\gamma \alpha) = N^{\mathrm{cy}(\mu)-s} \, \nu_f \, \mu \, \nu_i\;.
\end{equation}
If $s=0$, then $\nu_f=\nu_i=id$ and $\gamma\alpha=\mu\otimes\mu^{-1}$, therefore
\begin{equation}
    \mathrm{P_2Tr}(\alpha\gamma) = N^{{\mathrm{cy}}(\mu)}\mu\;.
\end{equation}
If $s>0$, the coefficient of $\nu_f\mu\nu_i$ is of lower order than from the from $(\nu_f\,\mu\,\nu_i)\otimes(\nu_f\,\mu\,\nu_i)^{-1}$, therefore
\begin{equation}
    \sum_{\gamma\in[2^n]}\mathrm{P_2Tr}(\gamma \alpha) = \sum_{\mu\in S_n}\mathrm{P_2Tr}(\mu\otimes\mu^{-1})\;\left(1+{\mathcal O}(N^{-2})\right)=\sum_{\mu\in S_n}N^{\mathrm{cy}(\mu)}\mu\;\left(1+{\mathcal O}(N^{-2})\right)\,,
\end{equation}
where the order is ${\mathcal O}(N^{-2})$ since $s$ is even.

\end{proof}

As an illustration of the proof above, we consider a permutation $\gamma$ in $[2^{10}]$ given by 
\begin{equation}
\gamma=(1\;5)(2\;17)(3\;18)(4\;15)(6\;12)(7\;10)(8\;13)(9\;11)(14\;16)(19\;20)    
\end{equation}
and is represented graphically in Fig.\,\ref{210c}.
For all the lines connecting the bottom to the top, call $x$ and $y$ the values at the bottom and top of the line, respectively. Each pair $x$ and $y$ falls into one of four distinct cases: $1\leq x\leq 10$ and $1\leq y \leq 10$; $11\leq x\leq 20$ and $11\leq y \leq 20$; $11\leq x\leq 20$ and $1\leq y \leq 10$; $1\leq x\leq 10$ and $11\leq y \leq 20$. In Fig.\,\ref{210c} we highlight the first two cases in red, the third case in green and the fourth case in pink.
\begin{figure}[t]
    \centering
    \includegraphics[width=0.7\linewidth]{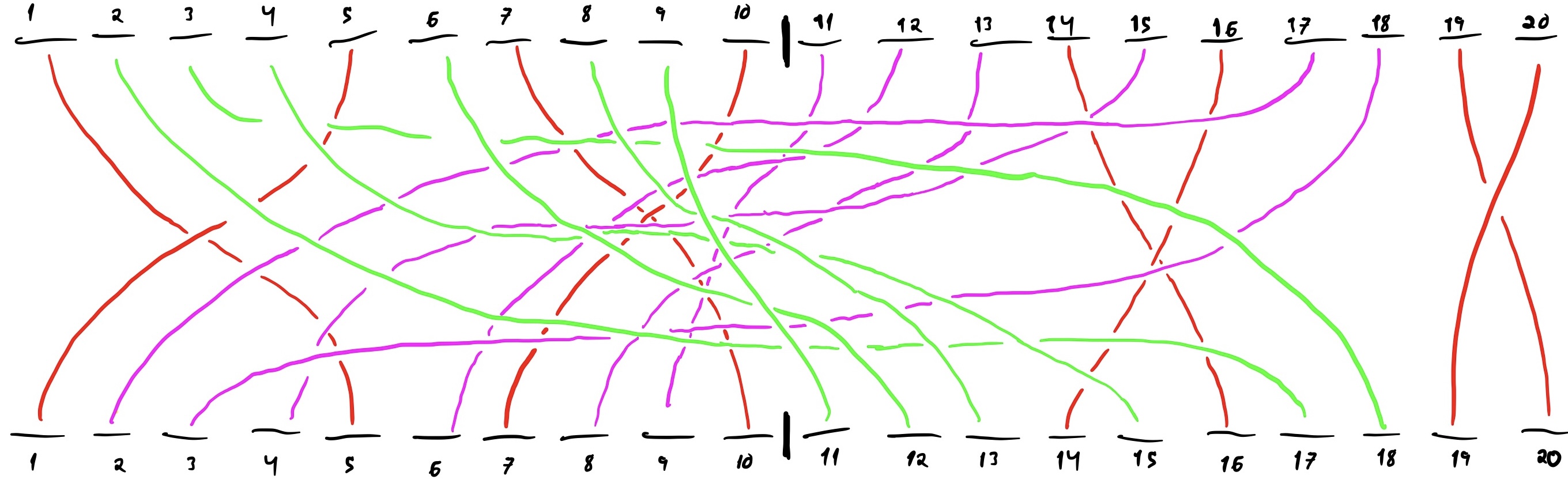}
    \caption{An example of a permutation in $[2^{10}]$.}
    \label{210c}
\end{figure}
The sets $p_i$ and $p_f$ used in the lemma \ref{lemapp} are $p_i=\{4,6,9,10\}$ and $p_f=\{1,5,7,10\}$. The permutation $\nu_f$ defined in \eqref{nuf} is the permutation associated with the case $1\leq x\leq 10$ and $1\leq y \leq 10$, shown in red on the left side of Fig.\,\ref{210c}, and given by
\begin{equation}
    \nu_f=(1\;5)(7\;10)(2)(3)(4)(6)(8)(9)\,.
\end{equation}
The permutation $\nu_i$ defined in \eqref{nui} is the permutation associated with the case $11\leq x\leq 20$ and $11\leq y \leq 20$, shown in red on the right side of Fig.\,\ref{210c}, and given by
\begin{equation}
    \nu_i=(4\;6)(9\;10)(1)(2)(3)(5)(7)(8)\,.
\end{equation}
The permutation $\mu$ defined in \eqref{mu} is the permutation associated with the case $11\leq x\leq 20$ and $1\leq y \leq 10$, shown in green in Fig.\,\ref{210c}, and given by
\begin{equation}
    \mu=(1\;9)(5\;4)(7\;2\;6)(3\;8)(10)\,.
\end{equation}
The permutation $\gamma\alpha$ is 
\begin{equation}
    \gamma\alpha=(1\;9\;20\;7\;2\;6\;14\;15)(5\;4\;16\;12\;17\;10\;19\;11)(3\;8)(18\;13)\;,
\end{equation}
and according to \eqref{res}, its partial trace is
\begin{equation}
    \mathrm{P_2Tr}(\gamma \alpha) = N\;(1\;9\;7\;2\;6)(3\;8)(4\;10\;5)\,.
\end{equation}

\bibliography{refs}

    \end{document}